\documentclass[11pt]{article}
\pdfoutput=1
\usepackage{a4wide}
\usepackage[utf8]{inputenc}
\usepackage[T1]{fontenc}

\usepackage{inconsolata}
\usepackage{libertine}

\usepackage{amsthm}
\usepackage{amssymb}
\usepackage{amsmath}
\usepackage{mathtools}
\usepackage{mathrsfs}
\usepackage{thmtools,thm-restate}

\usepackage{tikz}
\usepackage{graphicx}
\graphicspath{ {./images/} }

\usepackage{algorithm}
\usepackage[noend]{algpseudocode}
\usepackage{caption}
\usepackage{subcaption}
\usepackage[noadjust]{cite}
\usepackage{enumitem}
\usepackage{xspace}
\usepackage{xcolor}
\usepackage{xparse}
\usepackage{xfrac}
\usepackage[disable]{todonotes}
\usepackage{textpos}

\usepackage[colorlinks = true,linkcolor = blue,urlcolor = blue, citecolor = blue]{hyperref}

\usepackage[capitalize, nameinlink]{cleveref}

\newcommand\tk[1]{\todo[inline,size=\scriptsize,backgroundcolor=yellow]{#1 - \textbf{Tuukka}}}

\newcommand\wn[1]{\todo[inline,size=\scriptsize,backgroundcolor=teal!5!white]{#1 - \textbf{Wojtek}}}
\renewcommand\mp[1]{\todo[inline,size=\scriptsize,backgroundcolor=cyan]{#1 - \textbf{Micha\l}}}
\newcommand\ms[1]{\todo[inline,size=\scriptsize,backgroundcolor=green]{#1 - \textbf{Marek}}}

\newtheorem{theorem}{Theorem}[section]

\newtheorem{lemma}[theorem]{Lemma}
\newtheorem{corollary}[theorem]{Corollary}
\newtheorem{fact}[theorem]{Fact}
\newtheorem{definition}{Definition}

\crefname{claim}{Claim}{Claims}
\newtheorem{claim}[theorem]{Claim}
\newcommand{\cqed}{\ensuremath{\lhd}}
\makeatletter
\newenvironment{claimproof}{\par
	\pushQED{\cqed}%
	\normalfont \topsep6\p@\@plus6\p@\relax
	\trivlist
	\item\relax
	{\itshape
		Proof of the claim\@addpunct{.}}\hspace\labelsep\ignorespaces
}{%
	\hfill\popQED\endtrivlist\@endpefalse
}
\makeatother

\newcommand{\ceil}[1]{\left\lceil #1 \right\rceil}

\newcommand{\bag}{\ensuremath{\mathtt{bag}}}
\newcommand{\bags}{\ensuremath{\mathtt{bags}}}
\newcommand{\edges}{\ensuremath{\mathtt{edges}}}

\newcommand{\Oh}[2][]{{\cal O}_{#1}(#2)}
\newcommand{\adhesion}[2][]{\mathtt{adh}_{#1}(#2)}
\newcommand{\component}[2][]{\mathtt{cmp}_{#1}(#2)}
\newcommand{\parent}[2][]{\mathtt{parent}_{#1}(#2)}
\newcommand{\lca}[2]{\mathrm{lca}(#1, #2)}
\newcommand{\tw}[1]{\mathrm{tw}(#1)}
\newcommand{\neighopen}[2][]{N_{#1}(#2)}
\newcommand{\neighclosed}[2][]{N_{#1}[#2]}
\newcommand{\torso}[2][]{\mathtt{torso}_{#1}(#2)}
\newcommand{\cc}[1]{\mathrm{cc}(#1)}
\newcommand{\blockages}[2]{\mathsf{Blockages}(#1, #2)}
\newcommand{\exploration}[2]{\mathsf{Exploration}(#1, #2)}
\newcommand{\explorationgraph}[2]{\mathsf{ExplorationGraph}(#1, #2)}
\newcommand{\collectedcomps}[2]{\mathsf{Coll}(#1, #2)}
\newcommand{\collectedsym}{\ensuremath{\mathfrak{C}}\xspace}
\newcommand{\collectedmap}{\mathsf{collmap}}
\newcommand{\noclosuremsg}{\textsf{No closure}\xspace}
\newcommand{\Tpref}{\ensuremath{T_{\mathrm{pref}}}\xspace}
\newcommand{\App}{\ensuremath{\mathtt{App}}}
\newcommand{\height}{\ensuremath{\mathsf{height}}}
\newcommand{\weight}{\ensuremath{\mathsf{weight}}}
\newcommand{\size}{\ensuremath{\mathsf{size}}}
\newcommand{\cmpsize}{\ensuremath{\mathsf{cmpsize}}}

\newcommand{\Top}{\ensuremath{\mathsf{top}}}

\newcommand{\lheight}{\mathsf{lheight}}
\newcommand{\leaflabel}{\mathsf{label}}

\newcommand{\funrestriction}[2]{#1\vert_{#2}}

\newcommand{\D}{\mathbb{D}}
\newcommand{\F}{\mathbb{F}}
\newcommand{\Lb}{\mathbb{L}}
\newcommand{\M}{\mathbb{M}}
\newcommand{\N}{\mathbb{N}}

\newcommand{\Z}{\mathbb{Z}}

\newcommand{\Aa}{\mathcal{A}}
\newcommand{\Cc}{\mathcal{C}}
\newcommand{\Fc}{\mathcal{F}}

\newcommand{\Hc}{\mathcal{H}}
\newcommand{\Hh}{\Hc}
\newcommand{\Ic}{\mathcal{I}}

\newcommand{\Tc}{\mathcal{T}}

\newcommand{\Wc}{\mathcal{W}}

\renewcommand{\leq}{\leqslant}
\renewcommand{\geq}{\geqslant}

\renewcommand{\le}{\leqslant}
\renewcommand{\ge}{\geqslant}

\newcommand{\CMSO}{\mathsf{CMSO}}
\newcommand{\MSO}{\mathsf{MSO}}

\newcommand{\bnd}{\partial}
\newcommand{\tup}[1]{\overline{#1}}

\newcommand{\Daux}{\D^{\mathrm{aux}}}
\newcommand{\Dexplore}{\D^{\mathrm{explore}}}

\newcommand{\invinterface}[1]{\mathsf{Interface}^{-1}(#1)}
\newcommand{\interface}[1]{\mathsf{Interface}(#1)}
\newcommand{\Cinterface}{\interface{\collectedsym}}

\newcommand{\myparagraph}[1]{\paragraph*{#1}\mbox{}\\ \noindent}

\newcommand{\myopt}[1]{\ensuremath{#1_{\mathrm{opt}}}\xspace}
\newcommand{\myshallow}[1]{\ensuremath{#1 _{\mathrm{shallow}}}\xspace}

\newcommand{\Tdnew}[1]{{\cal T}^{#1}\xspace}

\newcommand{\Tnew}[1]{T^{#1}\xspace}
\newcommand{\bagnew}[1]{\bag^{#1}\xspace}
\newcommand{\Topt}[1]{\myopt{T^{#1}}}

\newcommand{\Tshallow}[1]{\myshallow{T^{#1}}}

\newcommand{\Tnewpref}[1]{\Tnew{#1}_{\mathrm{pref}}}

\newcommand*\circled[1]{\tikz[baseline=(char.base)]{
            \node[shape=circle,draw,inner sep=2pt] (char) {#1};}}

\newcommand{\origin}{\mathsf{origin}}
\newcommand{\origininv}{\mathsf{origin}^{-1}}

\newcommand{\neighall}[2]{\Ic(#1, #2)}

\newcounter{wcounter}
\newcommand{\wformat}{(WIDTH)}
\newcommand{\wcnt}{\wformat}
\newcommand{\witem}{\refstepcounter{wcounter}\item[\textit{\wcnt\label{width:width}}]}
\crefname{wcounter}{}{}
\creflabelformat{wcounter}{#2\wformat#3}
\newcommand{\wref}{\cref{width:width}}

\newcounter{amcounter}
\newcommand{\amformat}[1]{(POT)}
\newcommand{\amcnt}{\amformat{\arabic{amcounter}}}
\newcommand{\amitem}[1]{\refstepcounter{amcounter}\item[\textit{\amcnt\label{am:#1}}]}
\crefname{amcounter}{}{}
\creflabelformat{amcounter}{#2\amformat{#1}#3}
\newcommand{\amref}[1]{\cref{am:#1}}

\newcounter{pcounter}
\newcommand{\pformat}[1]{(P#1)}
\newcommand{\pcnt}{\pformat{\arabic{pcounter}}}
\newcommand{\pitem}[1]{\refstepcounter{pcounter}\item[\textit{\pcnt\label{prop:#1}}]}
\crefname{pcounter}{}{}
\creflabelformat{pcounter}{#2\pformat{#1}#3}
\newcommand{\pref}[1]{\cref{prop:#1}}

\newcounter{rtcounter}
\newcommand{\rtformat}[1]{(RT#1)}
\newcommand{\rtcnt}{\rtformat{\arabic{rtcounter}}}
\newcommand{\rtitem}[1]{\refstepcounter{rtcounter}\item[\textit{\rtcnt\label{rt:#1}}]}
\crefname{rtcounter}{}{}
\creflabelformat{rtcounter}{#2\rtformat{#1}#3}
\newcommand{\rtref}[1]{\cref{rt:#1}}

\newcounter{stepcounter}
\newcommand{\stepformat}[1]{Step \circled{#1}}
\newcommand{\stepcnt}{\stepformat{\arabic{stepcounter}}}
\newcommand{\stepitem}[1]{\refstepcounter{stepcounter}\stepcnt\label{step:#1} \xspace}
\crefname{stepcounter}{}{}
\creflabelformat{stepcounter}{#2\stepformat{#1}#3}
\newcommand{\stepref}[1]{\cref{step:#1}}

\newcommand{\Fprim}{\widetilde{F}}

\newcommand{\pullNei}{\mathsf{pull}}

\begin{document}
\title{\huge{Dynamic treewidth}\thanks{The work of Konrad Majewski, Wojciech Nadara, Micha\l{} Pilipczuk and Marek Soko\l{}owski on this manuscript is a part of a project that has received funding from the European Research Council (ERC), grant agreement No 948057 --- BOBR. Tuukka Korhonen was supported by the Research Council of Norway via the project
BWCA (grant no. 314528).}}
\author{Tuukka Korhonen\thanks{Department of Informatics, University of Bergen, Norway (\texttt{tuukka.korhonen@uib.no})} \and
Konrad Majewski\thanks{Institute of Informatics, University of Warsaw, Poland (\texttt{k.majewski@mimuw.edu.pl})} \and
Wojciech Nadara\thanks{Institute of Informatics, University of Warsaw, Poland (\texttt{w.nadara@mimuw.edu.pl})} \and
Michał Pilipczuk\thanks{Institute of Informatics, University of Warsaw, Poland (\texttt{michal.pilipczuk@mimuw.edu.pl})} \and
Marek Sokołowski\thanks{Institute of Informatics, University of Warsaw, Poland (\texttt{marek.sokolowski@mimuw.edu.pl})}}
\date{}
\maketitle
\thispagestyle{empty}

 \begin{textblock}{20}(-1.9, 8.2)
  \includegraphics[width=40px]{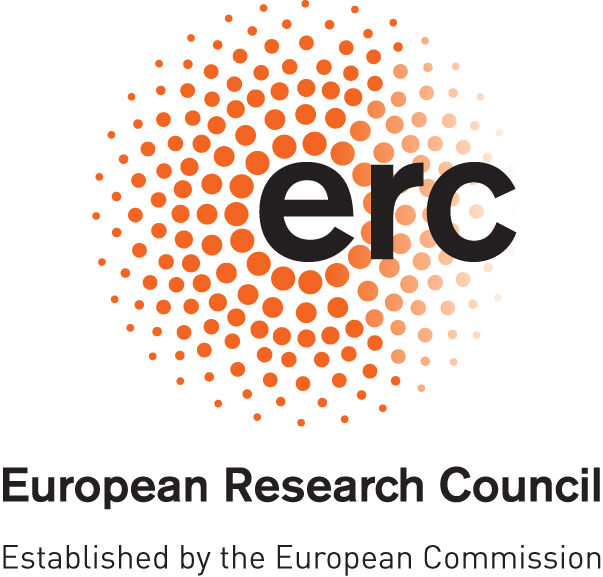}%
 \end{textblock}
 \begin{textblock}{20}(-2.15, 8.5)
  \includegraphics[width=60px]{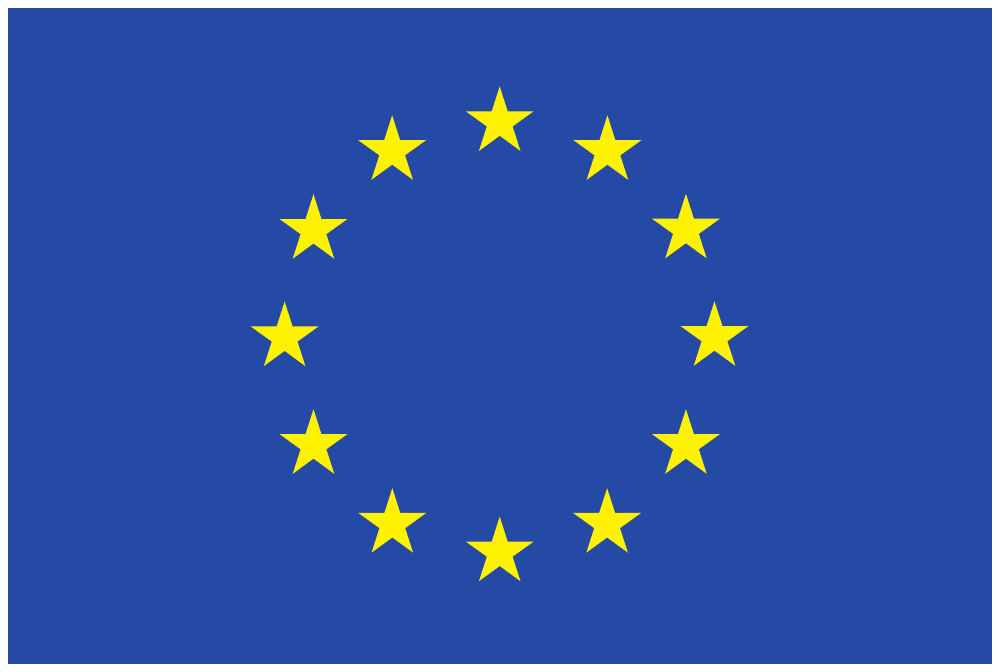}%
 \end{textblock}

\begin{abstract}
We present a data structure that for a dynamic graph $G$ that is updated by edge insertions and deletions, maintains a tree decomposition of $G$ of width at most $6k+5$ under the promise that the treewidth of $G$ never grows above $k$. The amortized update time is $\Oh[k]{2^{\sqrt{\log n}\log\log n}}$, where $n$ is the vertex count of $G$ and the $\Oh[k]{\cdot}$ notation hides factors depending on $k$. In addition, we also obtain the dynamic variant of Courcelle's Theorem: for any fixed property $\varphi$ expressible in the $\CMSO_2$ logic, the data structure can maintain whether $G$ satisfies $\varphi$ within the same time complexity bounds. To a large extent, this answers a question posed by Bodlaender~[WG~1993].

\end{abstract}

\newpage
\pagenumbering{arabic}

\newcommand{\ttl}{{\sf{``Treewidth too large''}}}

\section{Introduction}

\paragraph*{Treewidth.} {\em{Treewidth}} is a graph parameter that measures how much the structure of a graph resembles that of a tree. This resemblance is expressed through the notion of a {\em{tree decomposition}}, which is a tree $\Tc$ of {\em{bags}} --- subsets of vertices --- such that for every vertex $u$, the bags containing $u$ form a connected subtree of $\Tc$, and for every edge $uv$, there is a bag containing both $u$ and $v$. The measure of complexity of a tree decomposition is its {\em{width}} --- the maximum size of a bag minus~$1$ --- and the treewidth of a graph is the minimum possible width of its tree decomposition.

Treewidth and tree decompositions, in the form described above, were first introduced by Robertson and Seymour in their Graph Minors series, concretely in~\cite{RobertsonS86}, as a  key combinatorial tool for describing and analyzing graph structure. However, it quickly became clear that these notions are of great importance in algorithm design as well. This is because many problems which are hard on general graphs, like {\sc{Independent Set}}, {\sc{3-Coloring}}, or {\sc{Hamiltonicity}}, can be solved efficiently on graphs of bounded treewidth using dynamic programming. The typical approach here is to process the given tree decomposition in a bottom-up manner while keeping, for every node $x$ of the decomposition, a table expressing the existence of partial solutions with different ``fingerprints'' on the bag of $x$. The number of fingerprints is usually an (at least) exponential function of the bag size, hence obtaining a tree decomposition of small width is essential for designing efficient algorithms. We refer to the textbook of Cygan et al.~\cite[Section~7.3]{platypus} for a broad exposition of this framework.


Dynamic programming on tree decompositions is one of fundamental techniques of parameterized complexity, used ubiquitously in all branches of this field of research. Typically, tree-decomposition-based procedures serve as important subroutines in larger algorithms, where they are used to solve the problem on instances that already got simplified to having bounded treewidth; see~\cite[Section~7.7 and~7.8]{platypus} for several examples of this principle. The combinatorial aspects of treewidth, expressed for instance through the Grid Minor Theorem~\cite{RobertsonS86a}, can be used here to analyze what useful structures can be found in a graph when the treewidth is large. 

Outside of parameterized algorithms, dynamic programming on tree decompositions is used in approximation algorithms, especially in approximation schemes for planar and geometric problems. It is the cornerstone of the fundamental Baker's technique~\cite{Baker94}, and of countless other approximation schemes following the principle: simplify the input graph to a bounded-treewidth graph while losing only a tiny fraction of the optimum, and apply dynamic programming on a tree~decomposition.

%

The design of dynamic programming procedures on tree decompositions can be done in a semi-automatic way. In the late 80s, Courcelle observed that whenever a graph problem can be expressed in $\CMSO_2$ --- monadic second-order logic with edge quantification and counting modular predicates --- then the problem can be solved in time $f(k)\cdot n$ on graphs of treewidth at most $k$, for a computable function $f$~\cite{Courcelle90}. Here, $\CMSO_2$ is a powerful logic that extends first-order logic by the possibility of quantifying over vertex and edge subsets; consequently, it can concisely express difficult problems such as the aforementioned {\sc{3-Coloring}} or {\sc{Hamiltonicity}}. See~\cite[Section~7.4]{platypus},~\cite[Chapter~13]{DowneyF13}, and~\cite[Chapter~11]{FlumG06} for a broader discussion of Courcelle's Theorem and~$\CMSO_2$. 

In the proof of Courcelle's Theorem, one transforms the given sentence $\varphi$ of $\CMSO_2$ into an appropriate tree automaton $\Aa_\varphi$, and runs $\Aa_{\varphi}$ on a tree decomposition of the input graph. Therefore, apart from being a powerful meta-technique for an algorithm designer, Courcelle's Theorem also brings a variety of tools related to logic and tree automata to the setting of bounded treewidth graphs. This creates a fundamental link between combinatorics, algorithms, and finite model theory, and logical aspects of graphs of bounded treewidth are a vibrant research area to this day. We refer the reader to the book of Courcelle and Engelfriet~\cite{CEbook} for a broader introduction to this topic.

\paragraph*{Computing tree decompositions.} Before applying any dynamic programming procedure on a tree decomposition of a graph, one needs to compute such a decomposition first. Therefore, the problem of computing a tree decomposition of optimum or approximately optimum width has gathered significant attention over the years. While computing treewidth exactly is $\mathsf{NP}$-hard~\cite{ArnborgP89}, there is a vast literature on various approximation and fixed-parameter algorithms for computing tree decompositions. Let us name a few prominent results.
\begin{itemize}[nosep]
 \item The exact fixed-parameter algorithm of Bodlaender~\cite{bodlaender-tw-opt} finds a tree decomposition of optimum width of a treewidth-$k$ graph in time $2^{\Oh{k^3}}\cdot n$. Recently, Korhonen and Lokshtanov~\cite{DBLP:journals/corr/abs-2211-07154} improved the parametric factor to $2^{\Oh{k^2}}$ at the cost of increasing the dependency on $n$ to $n^4$.
\item The $4$-approximation algorithm of Robertson and Seymour~\cite{RobertsonS86} finds a tree decomposition of width at most $4k+3$ of a treewidth-$k$ graph in time $\Oh{8^k\cdot k^2\cdot n^2}$. This algorithm presents the most basic approach to approximating treewidth, is featured in textbooks on algorithm design~\cite{platypus,KleinbergTardos}, and provides a general framework for approximating other width parameters of graphs, such as rankwidth or branchwidth~\cite{OumS06,Oum08}. Recently, Korhonen~\cite{Korhonen21} gave a $2$-approximation algorithm for treewidth running in time $2^{\Oh{k}}\cdot n$.
 \item As far as polynomial-time approximation algorithms are concerned, the best known approximation ratio is $\Oh{\sqrt{\log k}}$~\cite{FeigeHL08}. It is known that treewidth cannot be approximated to a constant factor in polynomial time assuming the Small Set Expansion Conjecture~\cite{WuAPL14}.
\end{itemize}
For more literature pointers, see the introductory sections of the works listed above, particularly of~\cite{DBLP:journals/corr/abs-2211-07154}. A gentle introduction to classic approaches to computing tree decompositions can be also found in an overview article by Pilipczuk~\cite{Pilipczuk20a}.

Apart from the classic setting described above, the problem of computing a low-width tree decomposition of a given graph was also considered in other algorithmic paradigms. For instance, Elberfeld et al.~\cite{ElberfeldJT10} gave a {\em{slicewise logspace}} algorithm: an algorithm that computes a tree decomposition of optimum width in time $n^{f(k)}$ and {\em{space}} $f(k)\log n$, for some function $f$. Parallel algorithms for computing tree decompositions were studied by Bodlaender and Hagerup~\cite{bodlaender-hagerup}.

\paragraph*{Dynamic treewidth.} A classic algorithmic paradigm where the problem of computing tree decompositions has received relatively little attention is that of {\em{dynamic data structures}}. The basic setting would be as follows. We are given a dynamic graph $G$ that is updated over time by edge insertions and deletions. We are given a promise that the treewidth of $G$ is at all times upper bounded by $k$, and we would like to maintain a tree decomposition ${\cal T}$ of $G$ of width bounded by a function of~$k$. 
Ideally, together with ${\cal T}$ we would like to also maintain the run of any reasonable dynamic programming procedure on ${\cal T}$, so as to obtain a dynamic variant of Courcelle's Theorem: in addition, the data structure would be able maintain whether $G$ satisfies any fixed $\CMSO_2$ property $\varphi$.

This basic problem of {\em{dynamic treewidth}} is actually an old one: it was first asked by Bodlaender in 1993~\cite{Bodlaender93a}. Later, the question was raised again by Dvo\v{r}\'ak et al.~\cite{DvorakKT14}, and then repeated by Alman et al.~\cite{AlmanMW20}, Chen et al.~\cite{ChenCDFHNPPSWZ21}, and Majewski et al.~\cite{MajewskiPS23}. Despite efforts, the problem has so far remained almost unscratched, and to the best of our knowledge, no non-trivial --- with update time sublinear in $n$ --- data structure was known prior to this work.
Let us review the existing literature.
\begin{itemize}[nosep]
 \item None of the aforementioned static algorithms for computing tree decompositions is known to be liftable to the dynamic setting.
 \item There is a wide range of data structures for dynamic maintenance of forests and of various dynamic programming procedures working on them, see e.g.~\cite{BrodalF99, SleatorT83, AlstrupHLT05, Niewerth18}. Unfortunately, the simple setting of dynamic forests omits the main difficulty of the dynamic treewidth problem: the need of reconstructing the tree decomposition itself upon updates. Consequently, we do not see how any of these approaches could be lifted to graphs of treewidth higher than $1$.
 \item Bodlaender~\cite{Bodlaender93a} showed that on graphs of treewidth at most $2$, tree decompositions of width at most $11$ can be maintained with worst-case update time $\Oh{\log n}$. This result also comes with a dynamic variant of Courcelle's Theorem: the satisfaction of any $\CMSO_2$-expressible property $\varphi$ on graphs of treewidth at most $2$ can be maintained with worst-case update time $\Oh[\varphi]{\log n}$. The approach of Bodlaender relies on a specific structure theorem for graphs of treewidth at most $2$,
 which unfortunately does not carry over to larger values of the treewidth. In~\cite{Bodlaender93a}, Bodlaender also observed that for $k>2$, update time $\Oh[k]{\log n}$ can be achieved  in the decremental setting, when only edge deletions are allowed. But this again avoids the main difficulty of the problem, as in this setting no rebuilding of the tree decomposition is necessary.
 \item Independently of Bodlaender, Cohen et al.~\cite{CohenSTV93}\footnote{Unfortunately, the authors of this manuscript were not able to find access to reference~\cite{CohenSTV93}. Our description of the content of~\cite{CohenSTV93} is based on that of Bodlaender~\cite{Bodlaender93a}.}  tackled the case $k=2$ with worst-case update time $\Oh{\log^2 n}$, and the case $k=3$ in the incremental setting with worst-case update time~$\Oh{\log n}$. Frederickson~\cite{Frederickson98} studied dynamic maintenance of properties of graphs of treewidth at most~$k$, but the updates considered by him consist of direct manipulations of tree decompositions; this again avoids the main difficulty.
 \item Dvo\v{r}\'ak et al.~\cite{DvorakKT14} gave a data structure that in a dynamic graph of {\em{treedepth}} at most $d$, maintains an {\em{elimination forest}} (a decomposition suited for treedepth) with worst-case update time $f(d)$, for some computable $f$. Here, treedepth is a graph parameter that, intuitively, measures the depth of a tree decomposition, rather than its width. Treedepth of a graph is never smaller than its treewidth, and is upper bounded by the treewidth times $\log n$. The update time of the data structure of Dvo\v{r}\'ak et al. was later improved to $d^{\Oh{d}}$ by Chen et al.~\cite{ChenCDFHNPPSWZ21}. Along with their data structures, Dvo\v{r}\'ak et al. and Chen et al. also gave a dynamic variant of Courcelle's Theorem for treedepth: on a dynamic graph of treedepth at most $d$, the satisfaction of any fixed $\MSO_2$-expressible\footnote{$\MSO_2$ is a fragment of $\CMSO_2$ where modular counting predicates are not allowed.} property $\varphi$ can be maintained with worst-case update time~$\Oh[\varphi,d]{1}$.
 \item Building upon an earlier work of Alman et al.~\cite{AlmanMW20} on the dynamic feedback vertex set problem, Majewski et al.~\cite{MajewskiPS23} showed a dynamic variant of Courcelle's Theorem for the parameter feedback vertex number: the minimum size of a deletion set to a forest. This is another parameter that is lower bounded by the treewidth. That is, they showed that in a dynamic graph of feedback vertex number at most $\ell$, the satisfaction of any fixed $\CMSO_2$-expressible property $\varphi$ can be maintained with amortized update time $\Oh[\varphi,\ell]{\log n}$.
 \item Recently, Goranci et al.~\cite{GoranciRST21} gave the first non-trivial result on the general dynamic treewidth problem: using a dynamic algorithm for {\em{expander hierarchy}}, they can maintain a tree decomposition with $n^{o(1)}$-approximate width with amortized update time $n^{o(1)}$, under the assumption that the graph has bounded maximum degree. Note, however, that this result is rather unusable in the context of (dynamic) parameterized algorithms, because dynamic programming procedures on tree decompositions work typically in time exponential in the decomposition's width. Consequently, the result does not imply any dynamic variant of Courcelle's Theorem.
\end{itemize}
We remark that the dynamic treewidth problem seems to lie at the foundations of the emerging area of {\em{parameterized dynamic data structures}}. In this research direction, the goal is to design efficient data structures for parameterized problems where the update time is measured both in terms of relevant parameters and in terms of the total instance size. See~\cite{AlmanMW20,Bodlaender93a,ChenCDFHNPPSWZ21,CohenSTV93,DvorakKT14,DvorakT13,GrezMPPR22,IwataO14,MajewskiPS23,OlkowskiPRWZ23} for examples of this type of results.

\paragraph*{Our contribution.} In this work we give a resolution to the dynamic treewidth problem with subpolynomial amortized time complexity of the updates. That is, we present a data structure that for a fully dynamic graph $G$ of treewidth~$k$, maintains a constant-factor-approximate tree decomposition of $G$. The amortized update time is subpolynomial in $n$ for every fixed $k$. As a consequence, we prove the dynamic variant of Courcelle's Theorem for treewidth: the satisfaction of any fixed $\CMSO_2$ property $\varphi$ can be maintained within the same complexity bounds.
The statement below presents our main result in full formality.

 
\begin{theorem}\label{thm:main}
 There is a data structure that for an integer $k\in \N$, fixed upon initialization, and a dynamic graph $G$, updated by edge insertions and deletions, maintains a tree decomposition of $G$ of width at most $6k+5$ whenever $G$ has treewidth at most $k$. More precisely, at every point in time the data structure either contains a tree decomposition of $G$ of width at most $6k+5$, or a marker \ttl, in which case it is guaranteed that the treewidth of $G$ is larger than $k$. The data structure can be initialized on $k$ and an edgeless $n$-vertex graph $G$ in time $g(k)\cdot n$, and then every update takes amortized time $2^{f(k)\cdot \sqrt{\log n\log\log n}}$, where $g(k)\in 2^{k^{\Oh{1}}}$ and $f(k)\in k^{\Oh{1}}$ are computable functions.
 
 Moreover, upon initialization the data structure can be also provided a $\CMSO_2$ sentence $\varphi$, and it can maintain the information whether $\varphi$ is satisfied in $G$ whenever the marker \ttl{} is not present. In this case, the initialization time is $g(k,\varphi)\cdot n$ and the amortized update time is $h(k,\varphi)\cdot 2^{f(k)\cdot \sqrt{\log n\log\log n}}$, where $g$, $h$, and $f$ are computable functions.
\end{theorem}

Observe that the bound of $2^{f(k)\cdot \sqrt{\log n\log\log n}}$ on the amortized update time can be also expressed differently, in order to avoid having a factor depending on $k$ multiplied by factors depending on $n$ in the exponent. For example, if $f(k)\leq \sqrt{\log \log n}$, then we have $2^{f(k)\cdot \sqrt{\log n\log\log n}}\leq 2^{\sqrt{\log n}\log\log n}$, and otherwise, when $f(k)>\sqrt{\log \log n}$, we have $2^{f(k)\cdot \sqrt{\log n\log\log n}}\leq 2^{f(k)^2\cdot 2^{f(k)^2/2}}$; this gives a unified bound of
$2^{f(k)\cdot \sqrt{\log n\log\log n}}\leq 2^{f(k)^2\cdot 2^{f(k)^2/2}}+2^{\sqrt{\log n}\log\log n}$.
Other tradeoffs are possible by comparing $f(k)$ with other functions of $n$.

Further, note that the statement of \cref{thm:main} appears stronger than the setting discussed before: the data structure persists even at times when the treewidth grows above $k$,  instead of working under the assumption that this never happens. In fact, throughout the paper we work in the latter weaker setting, as it can be lifted to the stronger setting discussed in \cref{thm:main} in a generic way using the  technique of delaying invariant-breaking updates, proposed by Eppstein et al.~\cite{EppsteinGIS96}; see also \cite[Section~11]{abs-2006-00571}. We discuss how this technique applies to our specific problem in \cref{sec:ultimate-proof}.

%

Within the data structure of \cref{thm:main} we modify the maintained tree decomposition ${\cal T}$ only in a restricted fashion: through {\em{prefix-rebuilding updates}}. These amount to rebuilding a prefix\footnote{All our tree decompositions are rooted, and a {\em{prefix}} of a tree decomposition is an ancestor-closed subset of bags.} ${\cal S}$ of ${\cal T}$ into a new prefix ${\cal S}'$, and reattaching all trees of ${\cal T}-{\cal S}$ to ${\cal S}'$ without modifying them. As a consequence, while in \cref{thm:main} we only discuss maintenance of $\CMSO_2$ properties, in fact we can maintain the run of any standard dynamic programming procedure on ${\cal T}$. In \cref{sec:dynamic-dynamic-programming} we present a general automata-based framework for dynamic programming on tree decompositions that can be combined with our data structure. Automata verifying $\CMSO_2$-expressible properties are just one instantiation of this framework. 

Thus, \cref{thm:main} answers the question posed by Bodlaender in~\cite{Bodlaender93a} up to optimizing the update time. We conjecture that the update time can be improved to polylogarithmic in $n$, or even close to the $\Oh{\log n}$ bound achieved by Bodlaender for graphs of treewidth $2$.

\paragraph*{Applications.} Dynamic programming on tree decompositions is a fundamental technique used in countless parameterized algorithms. Therefore, the possibility of maintaining a run of a dynamic programming procedure in a fully dynamic graph of bounded treewidth opens multiple new avenues in the design of parameterized dynamic data structures. In particular, one can attempt to lift a number of classic treewidth-based approaches to the dynamic setting. Here is an example of how using our data structure, one can give a dynamic variant of the classic win/win approach to minor containment testing based on the Grid Minor Theorem.

\begin{corollary}\label{cor:minor-testing}
 Let $H$ be a fixed planar graph. There exists a data structure that for an $n$-vertex graph~$G$, updated by edge insertions and deletions, maintains whether $H$ is a minor of $G$. The initialization time on an edgeless graph is $\Oh[H]{n}$, and the amortized update time is $2^{\Oh[H]{\sqrt{\log n\log \log n}}}$.
\end{corollary}
\begin{proof}
 Since $H$ is planar, there exists $\ell\in \N$ such that the $\ell\times \ell$ grid contains $H$ as a minor. Consequently, by the Grid Minor Theorem~\cite{RobertsonS86a}, there exists $k\in \N$ such that every graph of treewidth larger than $k$ contains $H$ as a minor. Further, it is easy to write a $\CMSO_2$ sentence $\varphi_H$ that holds in a graph $G$ if and only if $G$ contains $H$ as a minor. It now suffices to set up the data structure of \cref{thm:main} for the treewidth bound $k$ and sentence $\varphi_H$. Note there that if this data structure contains the marker \ttl, it is necessary the case that $G$ contains $H$ as a minor.
\end{proof}

\cref{cor:minor-testing} is only a simple example, but dynamic programming on tree decompositions is a basic building block of multiple other, more complicated techniques; examples include bidimensionality, shifting, irrelevant vertex rules, or meta-kernelization. Consequently, armed with \cref{thm:main} one can approach dynamic counterparts of those developments. We expand on this in \cref{sec:conclusions}.

\paragraph*{Organization.} 
In \cref{sec:overview} we give an overview of our algorithm and present the key ideas behind the result.
In \cref{sec:preliminaries} we present notation and preliminary results, including our framework of ``prefix-rebuilding data structures'' and the statement of the main lemma (\cref{lem:weak-treewidth-ds}) that will imply \cref{thm:main}.
In \cref{sec:closures} we prove combinatorial results on objects called ``closures'', which will then be leveraged in \cref{sec:refinement} to build our main algorithmic tool called the ``refinement operation''.
Then, in \cref{sec:height} we use the refinement operation to build a height reduction operation for dynamic tree decompositions, and in \cref{sec:wrap-up} we put these results together and finish the proof of \cref{lem:weak-treewidth-ds}.
We discuss conclusions and future research directions in \cref{sec:conclusions}.
In \cref{sec:dynamic-dynamic-programming} we present our framework for dynamic maintenance of dynamic programming on tree decompositions and in \cref{sec:ultimate-proof} we complete the proof of \cref{thm:main} from~\cref{lem:weak-treewidth-ds}.

\newcommand{\Tt}{\mathcal{T}}

\section{Overview}
\label{sec:overview}
In this section we give an overview of our algorithm.
We first give a high-level description of the whole algorithm in \Cref{subsec:overview_main}, and then in \Cref{subsec:overview_refin,subsec:overview_heightredu} we sketch the proofs of the most important technical ingredients.

\subsection{High-level description}
\label{subsec:overview_main}
Let $n$ be the vertex count and $k$ a given parameter that bounds the treewidth of the considered dynamic graph $G$.
Our goal is to maintain a rooted tree decomposition $\Tt$ of height $2^{\Oh[k]{\sqrt{\log n \log\log n}}}$ and width at most $6k+5$, and at the same time any dynamic programming scheme, or more formally, a tree decomposition automaton with $\Oh[k]{1}$ evaluation time, on $\Tt$.
We will also require $\Tt$ to be binary, i.e., that every node has at most two children.

The goal of maintaining such a tree decomposition is reasonable because of a well-known lemma of Bodlaender and Hagerup~\cite{bodlaender-hagerup}: For every graph of treewidth $k$, there exists a binary tree decomposition of height $\Oh{\log n}$ and width at most $3k+2$.
Then, assuming $\Tt$ is a binary tree decomposition of height $h$ and width $\Oh{k}$, the operations of adding an edge or deleting an edge can be implemented in time $\Oh[k]{h}$ as follows.
Let us assume that we store the existence of an edge $uv$ in the highest node whose bag contains both $u$ and $v$.
Now, when deleting the edge $uv$, it suffices to find the highest node of $\Tt$ whose bag contains both $u$ and $v$, update information about the existence of this edge stored in this node, and then update dynamic programming tables of the nodes on the path from this node to the root, taking $\Oh[k]{h}$ time.
In the edge addition operation between vertices $u$ and $v$, we let $P_u$ (resp. $P_v$) be the path in $\Tt$ from the highest node containing $u$ (resp. $v$) to the root, add $u$ and $v$ to all bags on $P_u \cup P_v$, add the information about the existence of the edge $uv$ to the root node, and update dynamic programming tables on $P_u \cup P_v$, again taking in total $\Oh[k]{h}$ time.
Let us emphasize that only the highest bag containing both $u$ and $v$ is ``aware'' of the existence of the edge $uv$, as opposed to the more intuitive alternative of all the bags containing both $u$ and $v$ being ``aware'' of $uv$.
This is crucial for the fact that the dynamic programming tables of only $\Oh{h}$ nodes have to be recomputed after an edge addition/deletion.

Now, the only issue is that the edge addition operation could cause the width of $\Tt$ to increase to more than $6k+5$.
By maintaining Bodlaender-Kloks dynamic programming~\cite{DBLP:journals/jal/BodlaenderK96} on $\Tt$, we can detect if the treewidth of the graph actually increased to more than $k$ and terminate the algorithm in that case\footnote{By the standard technique of delaying updates we actually do not need to terminate the algorithm, but in this overview let us assume for simplicity that we are allowed to just terminate the algorithm if the width becomes more than $k$.}.
The more interesting case is when the treewidth of the graph is still at most $k$, in which case we have to modify $\Tt$ in order to make its width smaller while still maintaining small height.
The main technical contribution of this paper is to show that such changes to tree decompositions can indeed be implemented efficiently.

Let us introduce some notation.
We denote a rooted tree decomposition by a pair $\Tt=(T,\bag)$, where $T$ is a rooted tree and $\bag \colon V(T) \rightarrow 2^{V(G)}$ is a function specifying the bag $\bag(t)$ of each node $t$.
For a set of nodes $W \subseteq V(T)$, we denote by $\bags(W) = \bigcup_{t \in W} \bag(t)$.
A \emph{prefix} of a rooted tree $T$ is a a set of nodes $\Tpref \subseteq V(T)$ that contains the root and induces a connected subtree, and a prefix of a rooted tree decomposition $\Tt=(T,\bag)$ is a prefix of $T$.
For a rooted tree $T$, we denote by $\height(T)$ the maximum number of nodes on a root-leaf path, and for a node $t \in V(T)$, $\height(t)$ is the height of the subtree rooted at $t$.

Recall that in the edge addition operation, we increased the sizes of bags in a subtree consisting of the union $P_u \cup P_v$ of two paths, each between a node and the root.
In particular, all of the nodes with too large bags are contained in the prefix $P_u \cup P_v$ of size at most $2h$, where $h$ is the height of our tree decomposition.
Now, a natural idea for improving the width would be to replace the prefix $P_u \cup P_v$ by a tree decomposition of $G[\bags(P_u \cup P_v)]$ with height $\Oh{\log n}$ and width $3k+2$ given by the Bodlaender-Hagerup lemma.
While this form of the idea is too naive, we show that surprisingly, something that is similar in the spirit can be achieved.

The main tool we develop for maintaining tree decompositions in the dynamic setting is the \emph{refinement operation}.
The definition and properties of the operation are technical and will be described in \Cref{subsec:overview_refin}, but let us give here an informal description of what is achieved by the operation.
The refinement operation takes as an input a prefix $\Tpref$ of the tree decomposition $\Tt$ that we are maintaining, and informally stated, replaces $\Tpref$ by a tree decomposition of width at most $6k+5$ and height at most $\Oh{\log n}$.
The operation also edits other parts of $\Tt$, but in a way that makes them only better in terms of sizes of bags.
In particular, if we use the refinement operation on $\Tpref = P_u \cup P_v$ after an edge addition operation that made the width exceed $6k+5$ in the nodes in $P_u \cup P_v$, the operation brings the width of $\Tt$ back to at most $6k+5$.
The amortized time complexity of the refinement operation is $\Oh[k]{|\Tpref|}$, and it can increase the height of $\Tt$ by at most $\Oh{\log n}$.

With the refinement operation, we have a tool for keeping the width of the maintained tree decomposition $\Tt$ bounded by $6k+5$.
However, each application of the refinement operation can increase the height of $\Tt$ by $\log n$, so we need a tool also for decreasing the height.
We develop such a tool by a combination of a carefully chosen potential function and a strategy to decrease the potential function ``for free'' by using the refinement operation if the height is too large.
In particular, the potential function we use is
\[\Phi(\Tt) = \sum_{t \in V(T)} (\gamma \cdot k)^{|\bag(t)|} \cdot \height(t),\]
where $\gamma < 1000$ is a fixed constant defined in \cref{sec:refine-potential}.
This function has the properties that it does not increase too much in the edge addition operation
(the increase is at most $\Oh[k]{\height(T)^2}$), it plays well together with the details of the amortized analysis of the refinement operation (the factor $(\gamma \cdot k)^{|\bag(t)|}$ comes from there), and because of the factor $\height(t)$, it naturally admits smaller values on trees of smaller height.
In \Cref{subsec:overview_heightredu} we outline a strategy that, provided the height of $\Tt$ exceeds $2^{\Oh[k]{\sqrt{\log n \log\log n}}}$, selects a prefix $\Tpref$ so that applying the refinement operation to $\Tpref$ decreases the value of $\Phi(\Tt)$, and moreover the running time of the refinement operation can be bounded by this decrease.
In particular, this means that as long as the height is more than $2^{\Oh[k]{\sqrt{\log n \log\log n}}}$, we can apply such a refinement operation ``for free'', in terms of amortized running time, and moreover decrease the value of the potential.
As the potential cannot keep decreasing forever, repeated applications of such an operation eventually lead to improving the height to at most $2^{\Oh[k]{\sqrt{\log n \log\log n}}}$.

\subsection{The refinement operation}
\label{subsec:overview_refin}
In this subsection we overview the refinement operation.
First, let us note that our refinement operation builds on the tree decomposition improvement operation recently introduced by Korhonen and Lokshtanov for improving static fixed-parameter algorithms for treewidth~\cite{DBLP:journals/corr/abs-2211-07154}, which in turn builds on the 2-approximation algorithm for treewidth of Korhonen~\cite{Korhonen21} and in particular on the use of the techniques of Thomas~\cite{Thomas90} and Bellenbaum and Diestel~\cite{BellenbaumD02} in computing treewidth.
Our refinement operation generalizes the improvement operation of~\cite{DBLP:journals/corr/abs-2211-07154} by allowing to refine a continuous subtree instead of only a single bag, which is crucial for controlling the height of the tree decomposition.
We also adapt the operation from being suitable to compute treewidth exactly into a more approximative version (the $6k+5$ width comes from the combination of this and the Bodlaender-Hagerup lemma~\cite{bodlaender-hagerup}) in order to attain structural properties that are needed for efficient running time.
For this, the Dealternation Lemma of Boja{ń}czyk and Pilipczuk~\cite{DBLP:journals/lmcs/BojanczykP22} is used.
In the rest of this section we do not assume knowledge of these previous results.

Recall the goal of the refinement operation: Given a prefix $\Tpref$ of the tree decomposition $\Tt=(T,\bag)$ that we are maintaining, we would like to, in some sense, recompute the tree decomposition on this prefix.
Observe that we cannot simply select an induced subgraph like $G[\bags(\Tpref)]$ and hope to replace $\Tpref$ by any tree decomposition of it, because the tree decomposition needs to take into account also the connectivity provided by vertices outside of $\bags(\Tpref)$.
For this, the correct notion will be the \emph{torso} of a set of vertices.
Let $G$ be a graph and $X \subseteq V(G)$ a set of vertices.
The graph $\torso[G]{X}$ has the vertex set $X$ and has an edge $uv$ if there is a $u$-$v$-path in $G$ whose internal vertices are outside of $X$.
In other words, $\torso[G]{X}$ is the supergraph of $G[X]$ obtained by making for each connected component $C$ of $G-X$ the neighborhood $N(C)$ into a clique.

\begin{figure}[htb]
\begin{center}
\includegraphics[width=0.99\textwidth]{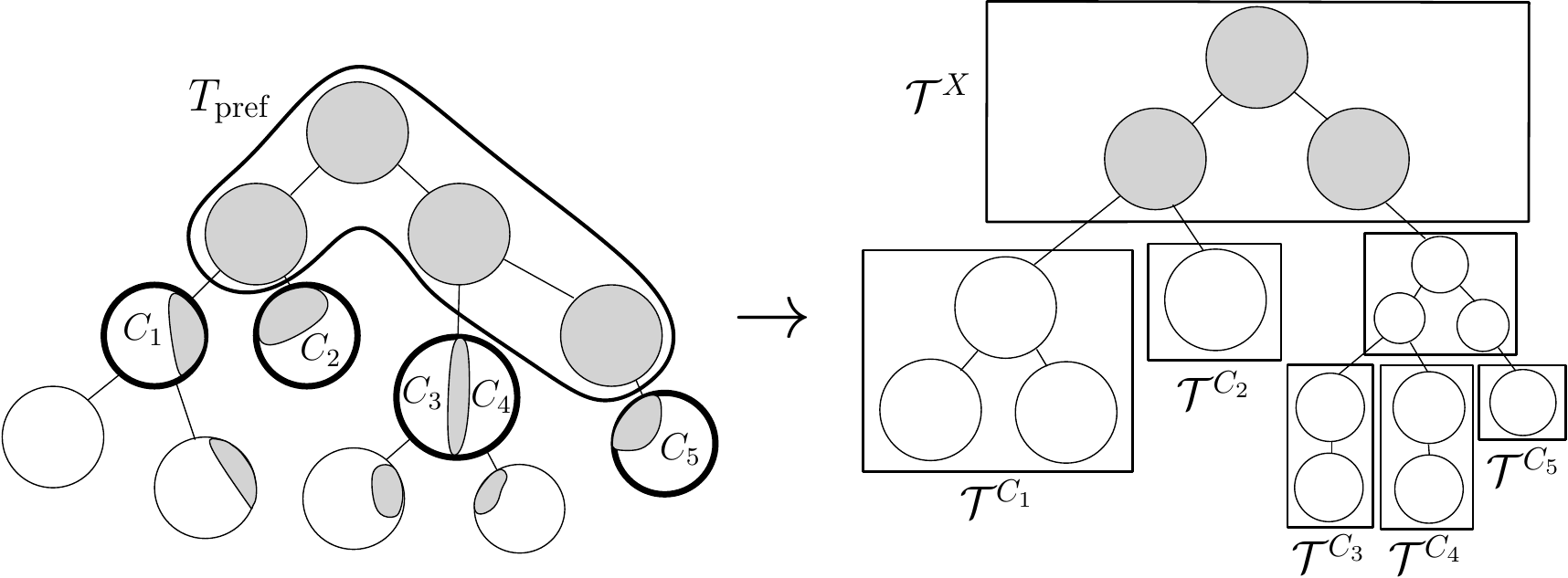}
\caption{The refinement operation. The left picture illustrates the tree decomposition~$\Tt$, with the prefix $\Tpref$ encircled and the vertices in $X \supseteq \bags(\Tpref)$ depicted in gray. The appendices of $\Tpref$ are circled by boldface, and the components of $G - X$ are denoted by $C_1,\ldots,C_5$.
The right picture illustrates the tree decomposition constructed from $\Tt$ by the refinement using $X$, in particular, by taking the tree decomposition $\Tt^X$ of $\torso[G]{X}$, and gluing the tree decompositions $\Tt^{C_i}$ for components $C_i$ of $G - X$ to it.
The subtree consisting of the three nodes above $\Tt^{C_3}$, $\Tt^{C_4}$, $\Tt^{C_5}$ is constructed in order to keep the tree binary after reattaching $\Tt^{C_3}$, $\Tt^{C_4}$, $\Tt^{C_5}$.
\label{fig:overview:refinement}}
\end{center}
\end{figure}

We then outline the refinement operation (see \cref{fig:overview:refinement} for an illustration of it).
Given $\Tpref$, we find a set of vertices $X \supseteq \bags(\Tpref)$ so that $\torso[G]{X}$ has treewidth at most $2k+1$ (here we have $2k+1$ instead of $k$ for a technical reason we will explain).
Then, we compute an optimum-width tree decomposition $\Tt^X$ of $\torso[G]{X}$ and use the Bodlaender-Hagerup lemma~\cite{bodlaender-hagerup} to make its height $\Oh{\log n}$, resulting in $\Tt^X$ having width at most $6k+5$.
We root $\Tt^X$ at an arbitrary node, and it will form a prefix of the new refined tree decomposition.
What remains, is to construct tree decompositions $\Tt^C$ for each connected component $C$ of $G - X$ and attach them into $\Tt^X$.

We say that a node $a \in V(T)$ of $\Tt$ is an {\em{appendix}} of $\Tpref$ if $a$ is not in $\Tpref$ but the parent of $a$ is.
For each appendix $a$ of $\Tpref$, denote by $\Tt_a = (T_a, \bag_a)$ the restriction of $\Tt$ to the subtree rooted at $a$.
Now, note that because $X \supseteq \bags(\Tpref)$, for each connected component $C$ of $G-X$ there exists a unique appendix $a$ of $\Tpref$ such that $C$ is contained in the bags of $\Tt_a$.
Moreover, the restriction $(T_a,\funrestriction{\bag_a}{N[C]})$ to the closed neighborhood $N[C]$ of $C$ is a tree decomposition of the induced subgraph $G[N[C]]$ minus the edges inside $N(C)$.
Then, observe that because $\Tt^X$ is a tree decomposition of $\torso[G]{X}$, it must have a bag that contains $N(C)$.
Now, our goal is to attach $(T_a,\funrestriction{\bag_a}{N[C]})$ into this bag.
In order to achieve this while satisfying the connectedness condition of tree decompositions, we need to have the set $N(C)$ in the root of $(T_a,\funrestriction{\bag_a}{N[C]})$.
We denote by $\Tt^C = (T^C,\bag^C)$ the tree decomposition obtained from $(T_a,\funrestriction{\bag_a}{N[C]})$ by ``forcing'' $N(C)$ to be in the bag $\bag^C(a)$ of the root node $a$, in particular, by inserting $N(C)$ to $\bag^C(a)$ and then fixing the connectedness condition by inserting each vertex $v \in N(C)$ to all bags on the unique path from the root to the subtree of the other bags containing $v$.
Then, $\Tt^C$ is a tree decomposition of $G[N[C]]$ whose root bag contains $N(C)$, and therefore it can be attached to the bag of $\Tt^X$ that contains $N(C)$.
These attachments may make the degree of the resulting tree decomposition higher than $2$, so finally these high-degree nodes need to be expanded into binary trees.
This concludes the informal description of the refinement operation.
The actual definition is a bit more involved, as it is necessary for obtaining efficient running time to (1) treat in some cases multiple different components $C$ in $\Tt_a$ as one component, and (2) prune out some unnecessary bags of $\Tt^C$.

From the description of the refinement operation sketched above, it should be clear that the resulting tree decomposition is indeed a tree decomposition of $G$.
It is also easy to see that the height of the refined tree decomposition is at most $\height(\Tt) + \Oh{\log n}$: This is because $\Tt^X$ has height at most $\Oh{\log n}$, and each of the attached decompositions $\Tt^C$ has height at most $\height(\Tt)$.
Recall that our goal is that if all of the bags of width more than $6k+5$ of $\Tt$ are contained in $\Tpref$, then the width of refined tree decomposition is at most $6k+5$.
The widths of the bags in $\Tt^X$ are clearly at most $6k+5$. 
However, because of the additional insertions of vertices in $N(C)$ to bags in $\Tt^C$, it is not clear why those bags would have width at most $6k+5$. In fact, we cannot guarantee this without additional properties of $X$ we outline next.

Let us call a set of vertices $X \subseteq V(G)$ a $k$-closure of $\Tpref$ if $X \supseteq \bags(\Tpref)$ and the treewidth of $\torso[G]{X}$ is at most $2k+1$.
In particular, the set $X$ in the refinement operation is a $k$-closure of~$\Tpref$.
We say that a $k$-closure $X$ is \emph{linked} into $\Tpref$ if for each component $C$ of $G - X$, the set $N(C)$ is linked into $\bags(\Tpref)$ in the sense that there are no separators of size $<|N(C)|$ separating $N(C)$ from $\bags(\Tpref)$.
The key property for controlling the width of $\Tt^C$ is that if $X$ is linked into $\Tpref$, then each bag of $\Tt^C$ has width at most the width of the corresponding bag in $\Tt$.
In particular, let $\Tt_a$ be a subtree of $\Tt$ hanging on an appendix $a$ of $\Tpref$, and $C$ be a component of $G-X$ contained in $\Tt_a$.
Recall that we construct the tree decomposition $\Tt^C = (T^C,\bag^C)$ by taking $T^C = T_a$, and then for each $t \in V(T^C)$ setting
\[\bag^C(t) = (\bag_a(t) \cap N[C]) \cup \{v \in N(C) \mid \text{the highest bag containing } v \text{ is below } t\}.\]

Then, we can bound the size of $\bag^C(t)$ as follows.
\begin{lemma}
\label{lem:overview:linkedstandard}
If $X$ is linked into $\Tpref$, then $|\bag^C(t)| \le |\bag_a(t)|$.
\end{lemma}
\begin{proof}[Proof sketch]
Let $N_t = \{v \in N(C) \mid \text{the highest bag containing } v \text{ is below } t\}$ and note that it suffices to prove $|N_t| \le |\bag_a(t) \setminus N[C]|$.
By linkedness and Menger's theorem, there are $N(C)$ vertex-disjoint paths from $N(C)$ to $\bags(\Tpref)$, and because $\bags(\Tpref)$ is disjoint from $C$, all internal vertices of these paths are in $V(G) \setminus N[C]$.
Moreover, $|N_t|$ of these paths are from $N_t$ to $\bags(\Tpref)$, and because $\bag_a(t)$ separates $N_t$ from $\bags(\Tpref)$ and is disjoint from $N_t$, each such path contains an internal vertex in $\bag_a(t)$, implying $|\bag_a(t) \setminus N[C]| \ge |N_t|$.
\end{proof}

\Cref{lem:overview:linkedstandard} shows that if $\Tpref$ contains all bags of width more than $6k+5$ and $X$ is linked into $\Tpref$, the resulting tree decomposition will have width at most $6k+5$.
We note that the existence of a $k$-closure of $\Tpref$ that is linked into $\Tpref$ is non-trivial, but before going into that let us immediately generalize the notion of linkedness in order to obtain a stronger form of \Cref{lem:overview:linkedstandard} that will be useful for analyzing the potential function.
The {\em{depth}} $d(v)$ of a vertex $v \in V(G)$ in $\Tt = (T,\bag)$ is the depth of the highest node of $T$ whose bag contains $v$ (i.e., the distance from this node to root).
For a set of vertices $S$, we denote $d(S) = \sum_{v \in S} d(v)$.
Then, we say that a $k$-closure $X$ is {\em{$d$-linked}} into $\Tpref$ if it is linked into $\Tpref$, and additionally for each neighborhood $N(C)$, there are no separators $S$ with $|S| = |N(C)|$ and $d(S) < d(N(C))$ separating $N(C)$ from $\bags(\Tpref)$.

Recall our potential $\Phi(\Tt) = \sum_{t \in V(T)} (\gamma \cdot k)^{|\bag(t)|} \cdot \height(t)$.
Using $d$-linkedness we are able to prove that the actual definition of the refinement operation satisfies the following properties.

\begin{lemma}[Informal]
\label{lem:overview:linkedpotential}
Let $X$ be a $d$-linked $k$-closure of $\Tpref$, $a$ an appendix of $\Tpref$, and $C_1, \ldots, C_\ell$ the connected components of $G-X$ that are contained in $\Tt_a$.
It holds that $\sum_{i=1}^\ell \Phi(\Tt^{C_i}) \le \Phi(\Tt_a)$, and moreover, the tree decompositions $\Tt^{C_i}$ for all $i$, together with their updated dynamic programming tables, can be constructed in time $\Oh[k]{\Phi(\Tt_a) - \sum_{i=1}^\ell \Phi(\Tt^{C_i})}$.
\end{lemma}

Let us note that \Cref{lem:overview:linkedpotential} would hold also for a potential function without the $\height(t)$ factor; this factor is included in the potential only for the purposes of the height reduction scheme that will be outlined in \Cref{subsec:overview_heightredu}.
Also, in the actual refinement operation each $C_i$ in \Cref{lem:overview:linkedpotential} can actually be the union of multiple different components with the same neighborhood $N(C_i)$, and we can actually charge a bit extra from the potential for each of these ``connected components''; this extra potential will be used for constructing the binary trees for the high-degree attachment points.

After ignoring these numerous technical details, the main takeaway of \Cref{lem:overview:linkedpotential} is that constructing the decompositions $\Tt^C$ is ``free'' in terms of the potential.
The only place where we could use a lot of time or increase the potential a lot is finding the set $X$ and constructing the tree decomposition $\Tt^X$.
For bounding this, we give a lemma asserting that we can assume $X$ to have size at most $\Oh[k]{|\Tpref|}$, and in addition to have an even stronger structural property that will be useful in the height reduction scheme.
For a node $a$ of $\Tt$, denote by $\component{a} \subseteq V(G)$ the vertices that occur in the bags of the subtree rooted at $a$, but not in bag of the parent of $a$.
We prove the following statement using the Dealternation Lemma of Boja\'nczyk and Pilipczuk~\cite{DBLP:journals/lmcs/BojanczykP22}.
We note that the bound $2k+1$ in the definition of $k$-closure comes from this proof.

\begin{lemma}
\label{lem:overview:smallclosure}
Let $G$ be a graph of treewidth at most $k$, and $\Tt$ a tree decomposition of $G$ of width~$\Oh{k}$.
For any prefix $\Tpref$ of $\Tt$, there exists a $k$-closure $X$ of $\Tpref$ so that for each appendix $a$ of $\Tpref$ it holds that $|X \cap \component{a}| \le \Oh{k^4}$. 
\end{lemma}

In particular, as $\Tt$ is a binary tree, $\Tpref$ has at most $|\Tpref|+1$ appendices. So by \Cref{lem:overview:smallclosure}, $\Tpref$ admits a $k$-closure with at most $|\Tpref| \cdot (k+1) + \Oh{|\Tpref| \cdot k^4} = \Oh{|\Tpref| \cdot k^4}$ vertices.
By using such a $k$-closure, we can bound the size of $\Tt^X$ by $\Oh[k]{|\Tpref|}$.
As each node in $\Tt^X$ has potential at most $\Oh[k]{\height(\Tt) + \log n}$ in the resulting decomposition, we get that the refinement operation increases the potential function by at most $\Oh[k]{|\Tpref| \cdot (\height(\Tt) + \log n)}$.
In particular, in the refinement operation applied directly after edge insertion, we have that $|\Tpref| \le 2 |\height(\Tt)|$, so the potential function increases by at most $\Oh[k]{\height(\Tt)^2}$ (note that $\height(\Tt) \ge \log n$).

Let us now turn to two issues that we have delayed for some time: How to guarantee that the $k$-closure $X$ is $d$-linked, and how to actually find such an $X$.
Let us say that a closure is \emph{$c$-small} if it satisfies the condition of \Cref{lem:overview:smallclosure} for some specific bound $c \in \Oh{k^4}$ that can be obtained from the proof of \Cref{lem:overview:smallclosure}.
The following lemma, which is proved similarly to proofs of Korhonen and Lokshtanov~\cite[Section 5]{DBLP:journals/corr/abs-2211-07154}, gives a simple condition that guarantees $d$-linkedness.

\begin{lemma}
\label{lem:overview:linkedoptimization}
Let $\Tpref$ be a prefix of a tree decomposition.
If $X$ is a $c$-small $k$-closure of $\Tpref$ that among all $c$-small $k$-closures of $\Tpref$ primarily minimizes $|X|$, and secondarily minimizes $d(X)$, then $X$ is $d$-linked into $\Tpref$.
\end{lemma}

With \Cref{lem:overview:linkedoptimization}, we can use dynamic programming for finding $c$-small $k$-closures that are $d$-linked.
In particular, we adapt the dynamic programming of Bodlaender and Kloks~\cite{DBLP:journals/jal/BodlaenderK96} for computing treewidth into computing $c$-small $k$-closures that optimize for the conditions in the lemma.
This adaptation uses quite standard techniques, but let us note that one complication is that even a small change to a tree decomposition can change the depths of all vertices, so we cannot just store the value $d(X)$ in the dynamic programming tables. Instead, we have to make use of the definition of the function $d$ and the fact that we are primarily minimizing $|X|$.
By maintaining these dynamic programming tables throughout the algorithm, we get that given $\Tpref$, we can in time $\Oh[k]{|\Tpref|}$ find an $\Oh{k^4}$-small $k$-closure $X$ of $\Tpref$ that is $d$-linked, and also the graph $\torso[G]{X}$.

\subsection{Height reduction}
\label{subsec:overview_heightredu}
In this subsection we sketch the height reduction scheme, in particular, the following lemma.

\begin{lemma}[Height reduction]
\label{lem:overview:heightredu}
Let $\Tt$ be the tree decomposition we are maintaining.
There is a function $f(n, k) \in 2^{\Oh[k]{\sqrt{\log n \log\log n}}}$ so that if $\height(\Tt) > f(n,k)$, then there exists a prefix $\Tpref$ of $\Tt$ so that the refinement operation on $\Tpref$ results in a tree decomposition $\Tt'$ with $\Phi(\Tt') < \Phi(\Tt)$ and runs in time $\Oh[k]{\Phi(\Tt) - \Phi(\Tt')}$.
\end{lemma}

To sketch the proof of \Cref{lem:overview:heightredu}, let us build a certain model of accounting how the potential changes in a refinement operation with a prefix $\Tpref$.
First, recall that \Cref{lem:overview:linkedpotential} takes care of the potential in the tree decompositions $\Tt^C$ for components $C$ of $G-X$, and the only place we need to worry about increasing the potential are the nodes of $\Tt^X$.
In the previous section we bounded this potential increase by $\Oh[k]{|\Tpref| \cdot (\height(\Tt) + \log n)}$, where in particular the factor $\height(\Tt) + \log n$ comes from the fact that after attaching the tree decomposition $\Tt^C$, the height of a node in $\Tt^X$ could be as large as $\height(\Tt) + \log n$.
This upper bound was sufficient for the refinement operation performed after an edge addition, but to prove \Cref{lem:overview:heightredu} we need a more fine-grained view.

Consider the following model: We start with the tree decomposition of $\Tt^X$ of height $\Oh{\log n}$, and attach the tree decompositions of the components $\Tt^C$ to it one by one.
Each time we attach a tree decomposition $\Tt^C$, we increase the height of at most $\Oh{\log n}$ nodes in $\Tt^X$ (because $\height(\Tt^X) \le \Oh{\log n}$), and this height is increased by at most $\height(\Tt^C)$, which is at most $\height(a)$, where $a$ is the appendix of $\Tpref$ whose subtree contains $C$.
While there can be many components $C$ contained in $\Tt_a$, observe that \Cref{lem:overview:smallclosure} implies that in fact such components can have only at most $k^{\Oh{k}}$ different neighborhoods $N(C)$, and therefore at most $k^{\Oh{k}}$ different attachment points in $\Tt^X$.
In particular, after handling technical details about the binary trees used to flatten high-degree attachment points, we can assume that an appendix $a$ of $\Tpref$ is responsible for increasing the height of at most $\Oh[k]{\log n}$ nodes in $\Tt^X$.
Moreover, it increases the height of those nodes by at most $\height(a)$ each, so in total, it is responsible for increasing the potential by $\Oh[k]{\height(a) \cdot \log n}$.

Denote $\Phi(\Tpref) = \sum_{t \in \Tpref} (\gamma \cdot k)^{|\bag(t)|} \cdot \height(t)$, i.e., the value of the potential on the nodes in $\Tpref$.
The above discussion, combined with \Cref{lem:overview:linkedpotential}, leads to the following lemma.

\begin{lemma}
Let $\Tpref$ be a prefix of a tree decomposition $\Tt$, $A \subseteq V(T)$ the appendices of $\Tpref$, and $\Tt'$ the tree decomposition resulting from refining $\Tt$ with $\Tpref$.
It holds that \[\Phi(\Tt') \le \Phi(\Tt) - \Phi(\Tpref) + \Oh[k]{|\Tpref| \cdot \log n} + \sum_{a \in A} \Oh[k]{\height(a) \cdot \log n}.\]
\end{lemma}

Now, in order to prove \Cref{lem:overview:heightredu}, it is sufficient to prove that if $\Tt$ has too large height, then there exists a prefix $\Tpref$ with appendices $A$ so that $\Phi(\Tpref) > c_k \log n \left(|\Tpref| + \sum_{a \in A} \height(a)\right)$, where $c_k$ is some large enough number depending on $k$. (Here, the required $c_k$ comes from the number of attachment points and the potential function, so some $c_k = k^{\Oh{k}}$ is sufficient.)
In our proof, the value $\Phi(\Tpref)$ in fact will be larger than $c_k \log n \left(|\Tpref| + \sum_{a \in A} \height(a)\right)$ by some arbitrary constant factor, which gives also the required property that the refinement operation on such a prefix $\Tpref$ will run in time $\Oh[k]{\Phi(\Tt) - \Phi(\Tt')}$.

\begin{figure}[htb]
\begin{center}
\includegraphics[width=0.6\textwidth]{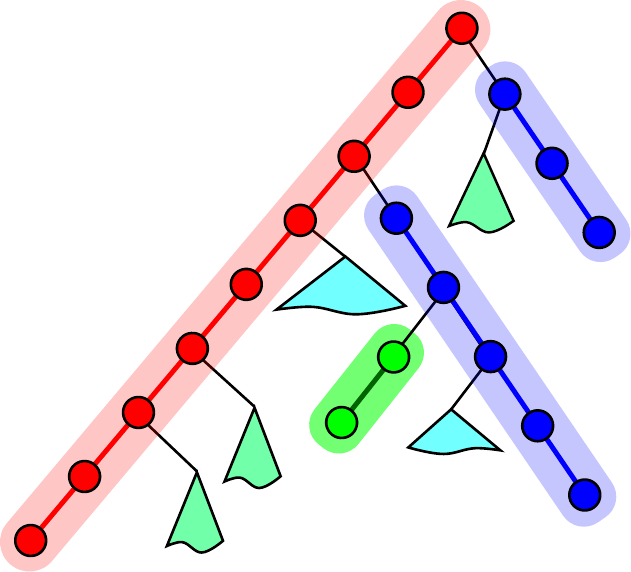}
\caption{Construction of $\Tpref$ in height reduction. The consecutive paths extracted by the construction procedure are depicted in red, blue, and green. Their union constitutes $\Tpref$. Big trees are depicted in sea-green, shallow trees are depicted in cyan.}\label{fig:overview:height}
\end{center}
\end{figure}

We first sketch how to select $\Tpref$ when $\height(T) > n^{\varepsilon}$ for some $\varepsilon > 0$ and $k$ is small compared to $n$; see \cref{fig:overview:height}.
Assume for simplicity that $|V(T)| = n$.
The natural strategy is to start by setting $\Tpref$ to be the path from the root to the deepest leaf in $T$.
Then we have $\Phi(\Tpref) \ge \Omega(n^{2 \varepsilon})$. 
However, $\Tpref$ may have $n^{\varepsilon}/2$ appendices that each have height $n^{\varepsilon}/2$, so it is possible that $\sum_{a \in A} \height(a) \cdot \log n \ge n^{2 \varepsilon} \cdot \log n/4$.
The key observation is that in this case, many subtrees of appendices $a \in A$ must be even more unbalanced than $T$ is, having height at least $n^{\varepsilon}/2$ while containing at most $4 \cdot n^{1-\varepsilon}$ nodes (simply by a counting argument).
In particular, let us say that a subtree rooted on an appendix $a$ is \emph{big} if it contains more than $c_k \cdot n^{1-\varepsilon} \cdot \log n$ nodes, and \emph{shallow} if $\height(a) \le n^{\varepsilon} / (c_k \cdot \log n)$.
Now, there can be at most $n/(c_k \cdot n^{1-\varepsilon} \cdot \log n) = n^{\varepsilon} / (c_k \cdot \log n)$
big subtrees, so they contribute at most $\Oh[k]{n^{2 \varepsilon} / c_k}$ to the sum $\sum_{a \in A} \Oh[k]{\height(a) \cdot \log n}$.
Similarly, the shallow subtrees also contribute at most $\Oh[k]{n^{2 \varepsilon} / c_k}$ to the sum, so by making the constant $c_k$ large enough the sum coming from subtrees that are either big or shallow (or both) is only a tiny fraction of $\Phi(\Tpref)$.


Then, there are subtrees of appendices that are neither big or shallow; these subtrees are both small and deep, hence they seem even more unbalanced than $\Tt$. We apply the same strategy to those trees recursively.
For each appendix $a$ whose subtree is small and deep, we insert to $\Tpref$ the path $P_a$ from $a$ to its deepest descendant.
As the subtree is deep, we have that $\Phi(\Tpref)$ increases by $\Phi(P_a) = \Omega(n^{2 \varepsilon} / (c_k \cdot \log n)^2)$.
Now, when analyzing the appendices of $P_a$, we again apply the strategy to handle subtrees that are big or shallow by charging them from $\Phi(P_a)$, and then handling subtrees that are both small and deep recursively.
This time, the right definition of big will be to have at least $n^{1-2 \varepsilon} (\log n \cdot c_k)^3$ nodes, and the right definition of shallow will be to have height at most $n^{\varepsilon} / (c_k \cdot \log n)^2$.
More generally, on the $i$-th level of such recursion we can call a subtree big if it contains more than $n^{1-i \varepsilon} (\log n \cdot c_k)^{i \cdot (i+1)/2}$ nodes, and shallow if its height is at most $n^{\varepsilon} / (c_k \cdot \log n)^{i}$.
When $\varepsilon$ is a constant, this recursion can continue only for a constant number of levels before no subtree can be both small and deep, simply because it would require the subtree to have larger height than the number of nodes.
Therefore, in the end we are able to find a prefix $\Tpref$ that satisfies the requirements of \Cref{lem:overview:heightredu}.

It is not surprising that selecting the height limit to be $n^{\varepsilon}$ is not optimal.
In particular, the same strategy as outlined above will work if we select the initial height to be of the form $2^{\Oh[k]{\sqrt{\log n \log\log n}}}$, resulting in showing that we can obtain the amortized running time of $2^{\Oh[k]{\sqrt{\log n \log\log n}}}$ per query.

\section{Preliminaries}
\label{sec:preliminaries}

For any positive integer $n$, we define $[n] \coloneqq \{1, 2, \dots, n\}$.
Given a~tuple of parameters $\bar{c}$, the notation $\Oh[\bar{c}]{\cdot}$ hides a multiplicative factor upper bounded by a~computable function of~$\bar{c}$.

\paragraph{Graphs.}
We use standard graph notation.
In this work, we only consider undirected simple graphs.
For a~graph $G$, by $V(G)$ and $E(G)$ we denote the vertex set and the edge set of $G$, respectively.
The set of connected components of $G$ is denoted $\cc{G}$.
For a~set $A \subseteq V(G)$, we denote the subgraph of $G$ induced by $A$ by $G[A]$, and the subgraph of $G$ induced by $V(G) \setminus A$ by $G - A$.

Given a~vertex $v \in V(G)$, we denote its (open) neighborhood by $\neighopen[G]{v}$ and its closed neighborhood by $\neighclosed[G]{v} \coloneqq \neighopen[G]{v} \cup \{v\}$.
This notation extends to the (open and closed) neighborhoods of subsets of vertices of $G$: given $A \subseteq V(G)$, we set $\neighclosed[G]{A} \coloneqq \bigcup_{v \in A} \neighclosed[G]{v}$ and $\neighopen[G]{A} \coloneqq \neighclosed[G]{A} \setminus A$.
When the graph $G$ is clear from the context, we may omit $G$ from this notation.

For a~set $A \subseteq V(G)$, we define the \emph{torso} of $A$ in $G$, denoted $\torso[G]{A}$, as the graph $H$ on the vertex set $A$ in which $uv \in E(H)$ if and only if $u, v \in A$ and there exists a~path connecting $u$ and $v$ that is internally disjoint with $A$.
Equivalently, $u, v \in A$ and we have that $uv \in E(G)$ or there exists a~connected component $C \in \cc{G - A}$ for which $\{u, v\} \subseteq \neighopen[G]{C}$.

Given two sets $A, B \subseteq V(G)$, we say that a~set $S \subseteq V(G)$ is an~$(A, B)$-\emph{separator} if every path connecting a~vertex of $A$ and a~vertex of $B$ intersects $S$.
If sets $A$, $S$, $B$ form a~partition of $V(G)$, then $(A, S, B)$ is a~\emph{separation} of $G$.
The \emph{order} of this separation is $|S|$.
Next, if every $(A, B)$-separator $S$ has size at least $|A|$, then we say that $A$ is \emph{linked into} $B$.
Equivalently (by Menger's theorem), there exist $|A|$ vertex-disjoint paths connecting $A$ and $B$.
Note that if $B \subseteq B'$ and $A$ is linked into $B$, then $A$ is also linked into $B'$.

\paragraph{Trees.}
We often call vertices of trees \emph{nodes} to distinguish them from vertices of graphs.
Unless explicitly stated otherwise, all trees in this work will be {\em{rooted}}: exactly one node of the tree is designated as the root of the tree. This naturally imposes parent/child and ancestor/descendant relations in the tree. 
Here, we assume that each node is both a~descendant and an~ancestor of itself.
We say that the tree is {\em{binary}} if each node has at most two children.
In this work, we will not distinguish between the left and the right child of a~node of a~binary tree.

The \emph{depth} of a~node $v$ is the distance from $v$ to the root, meanwhile the \emph{height} of a~node $v$ is the maximum distance from $v$ to a descendant of $v$ plus one.
In particular, the depth of the root is $0$, while the height of any leaf is $1$.

A~\emph{(rooted) forest} is a collection of disjoint (rooted) trees.
Two different nodes are~\emph{siblings} if they have a~common parent or are both roots. (The latter case may happen only in rooted forests.)

In a~tree $T$, we say that a~set $W \subseteq V(T)$ is a~\emph{prefix} of $T$ if for each non-root node $t \in W$, its parent $\parent{t}$ is also in $W$.
If a~node $t$ does not belong to $W$, but its parent does, we say that $t$ is an~\emph{appendix} of $W$.
Equivalently, $\neighopen[T]{W}$ is the set of appendices of $W$, which we will also denote by $\App(W)$.
The following fact is immediate:

\begin{fact}
  \label{fact:number-of-appendices}
  If $T$ is a~binary tree and $W$ is a~prefix of $T$, then there are at most $|W|+1$ appendices of $W$.
\end{fact}

The notion of \emph{lowest common ancestor} of two nodes $s, t$, denoted $\lca{s}{t}$, is defined in the standard way.
We say that a~subset $X \subseteq V(T)$ of the nodes of $T$ is \emph{lca-closed} if it satisfies the following property: for any two nodes $s, t \in X$, the node $\lca{s}{t}$ also belongs to $X$.
Given a~set $A \subseteq V(T)$, we define the \emph{lca-closure} of $A$ as the unique inclusion-wise minimal lca-closed set $X \supseteq A$.
Equivalently, $X = \{\lca{s}{t}\,\colon\,s, t \in A\}$.
The following facts are standard:

\begin{fact}
  The lca-closure $X$ of $A \subseteq V(T)$ satisfies $|X| \leq 2|A|-1$.
\end{fact}

\begin{fact}
  If $X$ is the lca-closure of $A \subseteq V(T)$, then each node of $X$ is an~ancestor of some node in $A$.
\end{fact}

\paragraph{Tree decompositions.}
A~tree decomposition of a~graph $G$ is a~pair $(T, \bag_T)$ comprised of a~tree~$T$ and a~function $\bag_T\,\colon\, V(T) \to 2^{V(G)}$, where:

\begin{itemize}
  \item for each vertex $v \in V(G)$, the subset of nodes $\{t \in V(T)\,\colon\,v \in \bag_T(t)\}$ induces a~nonempty connected subtree of $T$ (\emph{vertex condition}); and
  \item for each edge $uv \in E(G)$, there exists a~node $t \in V(T)$ such that $\{u, v\} \subseteq \bag_T(t)$ (\emph{edge condition}).
\end{itemize}
The \emph{width} of the decomposition $(T, \bag_T)$ is $\max_{t \in V(T)} |\bag_T(t)| - 1$.
The \emph{treewidth} of $G$, denoted $\tw{G}$, is the minimum possible width of any tree decomposition of $G$. 

Throughout this paper, all tree decompositions are rooted; that is, $T$ is always a rooted tree. If $T$ is a~(rooted) binary tree, then we say that $(T, \bag_T)$ is a~{\em{binary tree decomposition}}.

We also introduce the following syntactic sugar: given a~subset $W \subseteq V(T)$ of nodes of $T$, we set $\bags_T(W) \coloneqq \bigcup_{t \in W} \bag_T(t)$.
Usually the tree $T$ will be known from the context, so we will omit the subscript and write $\bag$ and $\bags$ instead of $\bag_T$ and $\bags_T$, respectively.
Also, by abusing the notation slightly, when a tree decomposition $(T,\bag)$ is clear from the context, we may simply write $\Tc$ as a shorthand for $(T, \bag)$.

The \emph{adhesion} of a~non-root node $t \in V(T)$ as $\adhesion{t} \coloneqq \bag(t) \cap \bag(\parent{t})$; if $t$ is the root of $T$, we set $\adhesion{t} \coloneqq \emptyset$.
Then, the \emph{component} of $t$ is defined as follows:
\[ \component{t} \coloneqq \left(\bigcup \{\bag(s)\,\colon\, s\text{ is a descendant of }t\text{ in }T\} \right) \setminus \adhesion{t}. \] 
We note that $(\component{t}, \adhesion{t}, V(G) \setminus (\component{t} \cup \adhesion{t}))$ is a~separation of $G$.


We remark the following well-known fact:

\begin{fact}
  \label{fact:tw-clique-bag}
  If $A$ is a~clique in $G$ and $\Tc$ is a~tree decomposition of $G$, then some bag of $\Tc$ must contain $A$.
\end{fact}

Finally, for a~tree decomposition $\Tc = (T, \bag)$ of a~graph $G$ we define the \emph{depth function}
\[ d_\Tc\,\colon\, V(G) \to \Z_{\geq 0} \]
that maps each vertex $v \in V(G)$ to the depth of the shallowest node $t \in V(T)$ such that $v \in \bag(t)$.
If the tree decomposition $\Tc$ is known from the context, we will write $d$ instead of $d_\Tc$.

\paragraph{Dynamic tree decompositions.}
We now present a~general design of a~data structure that will operate on a~dynamically changing binary tree decomposition of a~dynamically changing graph.

Define an~\emph{annotated tree decomposition} of a~graph $G$ as a~triple $(T, \bag, \edges)$ where:
\begin{itemize}
  \item $T$ and $\bag$ are defined as in the~standard definition of the tree decomposition;
  \item $\edges\,\colon\,V(T) \to 2^{\binom{V(G)}{2}}$ is defined as follows: for a node $t \in V(T)$, $\edges(t)$ is a~subset of $\binom{\bag(t)}{2}$ consisting of all edges $uv\in E(G)\cap \binom{\bag(t)}{2}$ for which $t$ is the shallowest node containing both $u$ and $v$. Note that thus, every edge of $G$ belongs to exactly one set $\edges(t)$.
\end{itemize} 
Note that $\edges$ is uniquely determined from $G$ and $(T, \bag)$; and conversely, $G$ is uniquely determined from $(T, \bag, \edges)$.
Given a~set $A \subseteq V(T)$, the restriction of $(T, \bag, \edges)$ to $A$, denoted $\funrestriction{(T, \bag, \edges)}{A}$, is the tuple $(T[A], \funrestriction{\bag}{A}, \funrestriction{\edges}{A})$ where $\funrestriction{\bag}{A}$ and $\funrestriction{\edges}{A}$ are restrictions of the functions $\bag$, $\edges$ to $A$, respectively.

Next, consider an~update changing an~annotated binary tree decomposition $(T, \bag, \edges)$ to another annotated binary tree decomposition $(T', \bag', \edges')$.
This update can also change the underlying graph $G$, in particular, it changes $G$ to be the graph uniquely determined from $(T', \bag', \edges')$.
We say that the update is \emph{prefix-rebuilding} if $(T', \bag', \edges')$ is created from $(T, \bag, \edges)$ by replacing a~prefix \Tpref of $T$ with a~new rooted tree $\Tpref'$ and then ``reattaching'' some subtrees of $T$ rooted at the appendices of $\Tpref$ below the nodes of $\Tpref'$.
Formally, a prefix-rebuilding update is described by a~tuple $\overline{u} \coloneqq (\Tpref, \Tpref', T^\star, \bag^\star, \edges^\star, \pi)$ where:
\begin{itemize}
  \item $\Tpref \subseteq V(T)$ is a~prefix of $T$;
  \item $\Tpref' \subseteq V(T')$ is a~prefix of $T'$ satisfying $$\funrestriction{(T, \bag, \edges)}{V(T) \setminus \Tpref} = \funrestriction{(T', \bag', \edges')}{V(T') \setminus \Tpref'};$$
  \item $(T^\star, \bag^\star, \edges^\star) = \funrestriction{(T', \bag', \edges')}{\Tpref'}$;
  \item $\pi \colon \App(\Tpref) \rightharpoonup \Tpref'$ is the partial function that maps appendices of $\Tpref$ to nodes of $\Tpref'$ such that for each appendix $t$ of \Tpref for which $\pi(t)$ is defined, the parent of $t$ in $T'$ is $\pi(t)$.
\end{itemize}
It is straightforward that $(T', \bag', \edges')$ can be uniquely determined from $(T, \bag, \edges)$ and the tuple $\overline{u}$ as above.
The \emph{size} of $\overline{u}$, denoted $|\overline{u}|$, is defined as $|\Tpref| + |\Tpref'|$.
It is also straightforward that given $\overline{u}$, a representation of $(T,\bag,\edges)$ can be turned into a representation of $(T',\bag',\edges')$ in time $\ell^{\Oh{1}} \cdot |\overline{u}|$, where $\ell$ is the maximum of the widths of $(T, \bag, \edges)$ and $(T',\bag',\edges')$.

Finally, we say that a~dynamic data structure is \emph{$\ell$-prefix-rebuilding with overhead $\tau$} if it stores an~annotated binary tree decomposition $(T, \bag, \edges)$ of width at most $\ell$ and supports the following operations:
\begin{itemize}
  \item $\mathsf{init}(T, \bag, \edges)$: initializes the annotated binary tree decomposition with $(T, \bag, \edges)$. Runs in worst-case time $\Oh{\tau \cdot |V(T)| \cdot \ell^{\Oh{1}}}$;
  \item $\mathsf{update}(\overline{u})$: applies a prefix-rebuilding update $\overline{u}$ to the decomposition $(T, \bag, \edges)$. It can be assumed that the resulting tree decomposition is binary and has width at most $\ell$. Runs in worst-case time $\Oh{\tau \cdot |\overline{u}| \cdot \ell^{\Oh{1}}}$.
\end{itemize}

Usually, the overhead $\tau$ will correspond to the time necessary to recompute any auxiliary information associated with each node of the decomposition undergoing the update.
For example, the height of a~node $s$ in the tree decomposition can be inferred in $\Oh{1}$ time from the heights of the (at most two) children of $s$, so the overhead required to recompute the heights of the nodes after the update is $\tau = \Oh{1}$ per affected node.

Prefix-rebuilding data structures will usually implement an~additional operation allowing to efficiently query the current state of the data structure.
For example, next we state a~data structure that allows us to access various auxiliary information about the tree decomposition:
\begin{lemma}\label{lem:height-maintenance}
  For every $\ell \in \N$, there exists an~$\ell$-prefix-rebuilding data structure with overhead $\Oh{1}$ that additionally implements the following operations:
  \begin{itemize}
    \item $\mathsf{height}(s)$: given a~node $s \in V(T)$, returns the height of $s$ in $T$. Runs in worst-case time $\Oh{1}$.
    \item $\size(s)$: given a~node $s \in V(T)$, returns the number of nodes in the subtree of $T$ rooted at $s$. Runs in worst-case time $\Oh{1}$.
    \item $\cmpsize(s)$: given a~node $s \in V(T)$, returns the size $|\component{s}|$. Runs in worst-case time $\Oh{1}$.
    \item $\Top(v)$: given a~vertex $v \in V(G)$, returns the unique highest node $t$ of $T$ so that $v \in \bag(t)$. Runs in worst-case time $\Oh{1}$.
  \end{itemize}
\end{lemma}
The proof of \cref{lem:height-maintenance} uses standard arguments on dynamic programming on tree decompositions.
It will be proved in \cref{sec:dynamic-dynamic-programming}.

More generally, any~typical dynamic programming scheme on tree decompositions can be turned into a~prefix-rebuilding data structure. Here is a statement that we present informally at the moment.

\begin{lemma}[informal]\label{lem:automaton-maintenance-informal}
  Fix $\ell \in \N$.
  Assume that there exists a~dynamic programming scheme operating on binary tree decompositions of width at most $\ell$, where the state of a~node $t$ of the tree decomposition depends only on $\bag(t)$, $\edges(t)$, and the states of the children of $t$ in the tree; and that this state can be computed in time $\tau$ from these information.
  Then, there exists an~$\ell$-prefix-rebuilding data structure with overhead $\tau$ that additionally implements the following operation:
  \begin{itemize}
    \item $\mathsf{state}(s)$: given a~node $s \in V(T)$, returns the state of the node $s$. Runs in worst-case time $\Oh{1}$.
  \end{itemize}
\end{lemma}

In \cref{sec:dynamic-dynamic-programming} we formalize what we mean by a ``dynamic programming scheme'' on tree decompositions through a suitable automaton model. Then \cref{lem:automaton-maintenance-informal} is formalized by a statement (\cref{lem:automaton-maintenance}) saying that the run of an automaton on a tree decomposition can be maintained under prefix-rebuilding updates, while the first three bullet points of \cref{lem:height-maintenance} are formally proved by applying this statement to specific (very simple) automata. In several places in the sequel, we will need to maintain more complicated dynamic programming schemes on tree decompositions under prefix-rebuilding updates. In every case, we state a suitable lemma about the existence of a prefix-rebuilding data structure, and this lemma is then proved in \cref{sec:dynamic-dynamic-programming} using a suitable automaton construction.

Finally, we show that the assumption that the function $\edges^\star$ is given in the description of a prefix-rebuilding update can be lifted in prefix-rebuilding updates that do not change the underlying graph $G$.
Consider a prefix-rebuilding update that does not change the graph $G$, and let us say that a \emph{weak description} of the update is a tuple $\widehat{u} \coloneqq (\Tpref, \Tpref', T^\star, \bag^\star, \pi)$ that is required to satisfy the same properties as a description of a prefix-rebuilding update except for the $\edges$ function.
Because the graph $G$ is not changed, the new annotated binary tree decomposition $(T',\bag',\edges')$ can be determined uniquely from $(T,\bag,\edges)$ and $\widehat{u}$.
We again denote $|\widehat{u}| = |\Tpref| + |\Tpref'|$.

We show that a weak description $\widehat{u}$ of a prefix-rebuilding update can be turned into a description $\overline{u}$ of a prefix-rebuilding update such that $|\overline{u}| = \Oh{|\widehat{u}|}$ and the annotated binary tree decomposition $(T',\bag',\edges')$ resulting from applying $\overline{u}$ is the same as the one resulting from applying $\widehat{u}$.
We note that this operation can make the sets $\Tpref$ and $\Tpref'$ larger, but this is bounded by $\Oh{|\widehat{u}|}$.

\begin{lemma}
\label{lem:prds-strengthen}
For every $\ell \in \N$, there exists an~$\ell$-prefix-rebuilding data structure with overhead $\Oh{1}$ that additionally implements the following operations:
\begin{itemize}
\item $\mathsf{strengthen}(\widehat{u})$: Given a weak description $\widehat{u}$ of a prefix-rebuilding operation, returns a description $\overline{u}$ of a prefix-rebuilding operation such that $|\overline{u}| = \Oh{\widehat{u}}$ and applying $\widehat{u}$ and $\overline{u}$ result in the same annotated tree decomposition $(T',\bag',\edges')$. Runs in worst-case time $|\widehat{u}| \cdot \ell^{\Oh{1}}$.
\end{itemize}
\end{lemma}
\begin{proof}
Let $\widehat{u} = (\widehat{T}_{\mathrm{pref}}, \widehat{T'}_{\mathrm{pref}}, \widehat{T}^\star, \widehat{\bag}^\star, \widehat{\pi})$ and $(T', \bag', \edges')$ be the resulting annotated tree decomposition.
We observe that the topmost bag containing an edge $uv \in E(G)$ can change only if both $u,v \in \bags_{T'}(\widehat{T'}_{\mathrm{pref}})$.
However, if $u,v \in \bags_{T'}(\widehat{T'}_{\mathrm{pref}})$, then because of the vertex condition of $(T',\bag')$ it must hold that $u,v \in \bags_{T}(\widehat{T}_{\mathrm{pref}} \cup \App(\widehat{T}_{\mathrm{pref}}))$, and in particular, the vertex condition in $(T,\bag)$ implies that $uv$ must be stored in $\edges(\widehat{T}_{\mathrm{pref}} \cup \App(\widehat{T}_{\mathrm{pref}}))$.
Therefore, the changes to the $\edges$ function are limited to the subtree of $T$ consisting of $\widehat{T}_{\mathrm{pref}} \cup \App(\widehat{T}_{\mathrm{pref}})$, and therefore for constructing $\overline{u} = (\Tpref, \Tpref', T^\star, \bag^\star, \edges^\star, \pi)$ it suffices to take $\Tpref = \widehat{T}_{\mathrm{pref}} \cup \App(\widehat{T}_{\mathrm{pref}})$ and analogously construct $\Tpref'$, $T^\star$, $\bag^\star$, and $\pi$ from $\widehat{T'}_{\mathrm{pref}}$, $\widehat{T}^\star$, $\widehat{\bag}^\star$, and $\widehat{\pi}$.
Then, the $\edges^\star$ function can be determined from $(T^\star, \bag^\star)$ and the $\edges$ function restricted to $\widehat{T}_{\mathrm{pref}} \cup \App(\widehat{T}_{\mathrm{pref}})$.
The running time and the bound on $|\overline{u}|$ follow from the fact that the tree decompositions are binary.
\end{proof}

Now, by using the data structure from \Cref{lem:prds-strengthen}, we can assume when constructing prefix-rebuilding operations that do not change $G$ that it is sufficient to construct a weak description, but when implementing prefix-rebuilding data structures that the $\mathsf{update}$ method receives a (not weak) description.
In the rest of this paper, we assume that we are always maintaining the data structure from \cref{lem:prds-strengthen}, in particular, usually first using it to turn a weak description $\widehat{u}$ into a description $\overline{u}$, and then immediately applying $\mathsf{update}(\overline{u})$ to it.

%
%
%

\newcommand{\inc}{\mathsf{inc}}

\paragraph{Logic.} We use $\CMSO_2$ --- monadic second-order logic on graphs with quantification over edge subsets and modular counting predicates --- which is typically associated with graphs of bounded treewidth; see~\cite[Section~7.4]{platypus} for an introduction suited for an algorithm designer. Formulas of $\CMSO_2$ are evaluated in graphs and there are variables of four different sorts: for single vertices, for single edges, for vertex subsets, and for edge subsets. The latter two sorts are called {\em{monadic}}. The atomic formulas of $\CMSO_2$ are of the following forms:
\begin{itemize}[nosep]
 \item {\em{Equality:}} $x=y$, where $x,y$ are both either single vertex/edge variables.
 \item {\em{Membership:}} $x\in X$, where $x$ is a single vertex/edge variable and $X$ is a monadic vertex/edge variable.
 \item {\em{Incidence:}} $\inc(x,e)$, where $x$ is a single vertex variable and $e$ is a single edge variable.
 \item {\em{Modular counting:}} $|X|\equiv a\pmod m$, where $a$ and $m$ are integers, $m>0$. 
\end{itemize}
The semantics of the above is as expected. Then $\CMSO_2$ consists of all formulas that can be obtained from atomic formulas using the following constructs: standard boolean connectives, negation, and quantification over all sorts of variables, both existential and universal. Thus, a formula of $\CMSO_2$ may contain variables that are not bound by any quantifier; these are called {\em{free variables} }. A formula without free variables is a {\em{sentence}}. For a sentence $\varphi$ and a graph $G$, we write $G\models \varphi$ to signify that $\varphi$ is satisfied in $G$ (read {\em{$G$ is a model of $\varphi$}}).

Courcelle's Theorem states that given a graph $G$ of treewidth $k$ and a $\CMSO_2$ formula $\varphi$, it can be decided whether $G\models \varphi$ in time $f(k,\varphi)\cdot n$, where $n$ is the vertex count of $G$ and $f$ is a computable function. In the proof of Courcelle's Theorem, one typically first computes a tree decomposition of $G$ of width at most $k$, for instance using the algorithm of Bodlaender~\cite{bodlaender-tw-opt}, and then applies a dynamic programming procedure ({\em{aka}} automaton) suitably constructed from $\varphi$ to verify the satisfaction of~$\varphi$. We show that this dynamic programming procedure can be maintained under prefix-rebuilding updates. More formally, in \cref{sec:dynamic-dynamic-programming} we prove the following statement.

\begin{lemma} \label{cor:dynamic-cmso}
 Fix $\ell\in \N$ and a $\CMSO_2$ sentence $\varphi$.
 Then there exists an~$\ell$-prefix-rebuilding data structure with overhead $\Oh[\ell,\varphi]{1}$ that additionally implements the following operation:
  \begin{itemize}
    \item $\mathsf{query}()$: returns whether $G\models \varphi$. Runs in worst-case time $\Oh[\ell,\varphi]{1}$.
  \end{itemize}
\end{lemma}

\paragraph{Dynamic tree decompositions under the promise of small treewidth.}
With all the definitions in place, we can finally state the core result that will be leveraged to prove \cref{thm:main}.

\begin{restatable}{lemma}{torestateWeakTreewidthDs}
  \label{lem:weak-treewidth-ds}
  There is a~data structure that for an~integer $k \in \N$, fixed upon initialization, and a~dynamic graph $G$, updated by edge insertions and deletions, maintains an~annotated tree decomposition $(T, \bag, \edges)$ of $G$ of width at most $6k + 5$ using prefix-rebuilding updates under the promise that $\tw{G}\leq k$ at all times.
  More precisely, at every point in time the graph is guaranteed to have treewidth at most $k$ and the data structure contains an~annotated tree decomposition of $G$ of width at most $6k + 5$.
  The data structure can be initialized on $k$ and an~edgeless $n$-vertex graph $G$ in time $2^{\Oh{k^8}} \cdot n$, and then every update:
  \begin{itemize}
    \item returns the sequence of prefix-rebuilding updates used to modify the tree decomposition; and
    \item takes amortized time $2^{\Oh{k^9 + k \log k \cdot \sqrt{\log n \log \log n}}}$.
  \end{itemize}
\end{restatable}

\ms{above: probably replace $\Oh{1}$ with specific constants later}

Note that as a~direct consequence of \cref{lem:weak-treewidth-ds}, the total size of all prefix-rebuilding updates returned by the data structure over first $q$ edge insertions/deletions is bounded by
\[ 2^{\Oh{k^8}} \cdot n\, +\, 2^{\Oh{k^9 + k \log k \cdot \sqrt{\log n \log \log n}}} \cdot q. \]
\wn{The above conclusion is quite weak as we get better bounds for the updates sizes than for running time, but who in their right mind would track that down...}

\cref{lem:weak-treewidth-ds} is proved in \cref{sec:wrap-up} using the results of \cref{sec:closures,sec:refinement,sec:height}.
We remark that the statement of the lemma is essentially a~weaker version of \cref{thm:main}: first, we assume that no update increasing $\tw{G}$ above $k$ may ever arrive to the data structure; next, we do not support dynamic $\CMSO_2$ model checking.
We fix these issues in \cref{sec:ultimate-proof} by means of, respectively: a~straightforward application of the technique of postponing invariant-breaking insertions of Eppstein et al.~\cite{EppsteinGIS96}, and \cref{cor:dynamic-cmso}.

\section{Closures}
\label{sec:closures}
\tk{A general todo for this section, if there is nothing better to do, would be to tidy up the formulas that are cut by line breaks...}
In this section, we introduce a~graph-theoretical notion of a~\emph{closure}, which will be used in the presentation of our algorithm later in the paper.
For the rest of the section, fix an~integer $k \in \N$ and let $G$ be a~graph of treewidth at most $k$.

Intuitively, when one tries to maintain a~tree decomposition of a~graph dynamically, one inevitably reaches a~situation where some of the bags of the maintained tree decomposition are too large.
Consider the following naive approach of improving such a~tree decomposition: let $W$ be the union of the bags that are deemed too large.
Construct a~(rooted) tree decomposition $T_W$ of $\torso[G]{W}$.
Then, for each connected component $C \in \cc{G - W}$, the set $\neighopen{C}$ is a~clique in $\torso[G]{W}$; hence, the entire set $\neighopen{C}$ resides in a~single bag $t_C$ of $T_W$.
Therefore, $C$ can be incorporated into $T_W$ by constructing a~tree decomposition $T_C$ of $G[\neighclosed{C}]$ whose root bag contains $\neighopen{C}$ entirely, and then attaching the root of $T_C$ to the bag $t_C \in V(T_W)$.
It can be straightforwardly verified that this is a~valid construction of a~tree decomposition of $G$.

However, it is not clear why this construction would improve the width of the maintained decomposition.
This owes to the fact that the treewidth of $\torso[G]{W}$ might be in principle much larger than $k$.
One might, however, hope that the set $W$ can be covered by an~only slightly larger set $X \supseteq W$ such that the treewidth of $\torso[G]{X}$ is small.
This is, indeed, the case, leading to the definition of a~\emph{closure} of $W$:

\begin{definition}
  Let $k \in \N$ and $G$ be such that $\tw{G} \leq k$.
  Let also $W \subseteq V(G)$.
  Then, the set $X \subseteq V(G)$ is called the \emph{$k$-closure} of $W$ in $G$ if
  \[ X \supseteq W \qquad\text{and}\qquad \tw{\torso[G]{X}} \leq 2k + 1. \]
\end{definition}

Note that each set $W$ admits a~trivial closure $X = V(G)$, as $\torso[G]{V(G)} = G$.
Obviously, such a~closure might be much larger than $W$.
Fortunately, the following lemma shows how to construct closures of more manageable size:

\begin{lemma}
  \label{lem:tw-torso-for-sum-of-bags}
  Let $(T, \bag)$ be a~tree decomposition of $G$ of width at most $k$.
  Let also $S$ be an~lca-closed set of nodes of $T$.
  Then,
  \[ \tw{\torso[G]{\bags(S)}} \leq 2k + 1. \]
  
  \begin{proof}
    We construct a~rooted tree $U$ in the following way: let $V(U) = S$ and let $t_1$ be a~parent of $t_2$ in $U$ if and only if $t_1$ is a~strict ancestor of $t_2$ and the simple path between $t_1$ and $t_2$ in $T$ does not contain any other vertices of $S$.
    Since $S$ is lca-closed, it can be easily verified that $U$ is indeed a~rooted tree.
    We also construct a~tree decomposition $(U, \bag')$ as follows:
    \[
      \bag'(t) = \begin{cases}
        \bag(t) & \text{if }t\text{ is the root of }U\text{,} \\
        \bag(t) \cup \bag(\parent[U]{t}) & \text{otherwise.}
      \end{cases}
    \]
    We claim that $(U, \bag')$ is a~tree decomposition of $\torso[G]{\bags(S)}$ of width at most $2k + 1$.
    The vertex condition is straightforward to verify.
    For the edge condition, consider an~edge $uv$ of $\torso[G]{\bags(S)}$.
    We have that $u, v \in \bags(S)$ and there exists a~simple path $P$ between $u$ and $v$ in $G$ that is internally disjoint with $\bags(S)$.
    Pick two nodes $t_u$, $t_v$ of $S$ such that $u \in \bag(t_u)$, $v \in \bag(t_v)$.
    For each node $t$ on the path between $t_u$ and $t_v$ in $U$, the set $\bag(t)$ must contain either $u$ or $v$; otherwise, $\bag(t)$ would be a~separator between $u$ and $v$ in $G$ disjoint with $\{u, v\}$, contradicting the existence of $P$.
    Hence, one of the following must hold:
    \begin{itemize}
      \item Both $u$ and $v$ belong to $\bag(t)$ for some $t \in S$. Then $u, v \in \bag'(t)$, so the edge condition is satisfied for the edge $uv$.
      \item We have $u \in \bag(t_1)$, $v \in \bag(t_2)$ for some nodes $t_1, t_2 \in S$ that are adjacent in $U$.
      Without loss of generality, assume that $t_1$ is the parent of $t_2$.
      Then, since $\bag'(t_2) = \bag(t_1) \cup \bag(t_2)$, we infer that $u, v \in \bag'(t_2)$.
    \end{itemize}
    We conclude that $(U, \bag')$ is indeed a~tree decomposition of $\torso[G]{\bags(S)}$.
    Since each bag of $(U, \bag')$ has size at most $2k + 2$, the proof is finished.
  \end{proof}
\end{lemma}

\cref{lem:tw-torso-for-sum-of-bags} already shows that each set $W \subseteq V(G)$ admits a~closure $X$ of cardinality at most $\Oh{k} \cdot |W|$. Indeed, consider a~tree decomposition $(T_{\mathrm{opt}}, \bag)$ of $G$ of minimum width.
For each vertex $v \in W$, select into $S$ a~node $t \in V(T_{\mathrm{opt}})$ such that $v \in \bag(t)$.
Then, take the lca-closure of $S$ (which increases $|S|$ by a~factor of at most 2) and apply \cref{lem:tw-torso-for-sum-of-bags}.

Unfortunately, this will not be sufficient in our setting.
In our algorithm, as we maintain a~tree decomposition $(T, \bag)$, the set $W$ will be chosen as the union of bags in a~prefix $\Tpref$ of $T$.
In this setup, another condition on the closure $X \supseteq W$ will be required: for every appendix $t$ of $\Tpref$, we require that the entire component $\component{t}$ contains only a~bounded number of vertices of $X$. Such closures will be called \emph{small}.
The existence of such a~closure will be proved as \cref{lem:small-closure-lemma} (Small Closure Lemma) in \cref{ssec:small-closure-lemma}.
This is followed in \cref{ssec:closure-linkedness-lemma} by proving a~structural result about small closures (\cref{lem:closure-linkedness-lemma}, Closure Linkedness Lemma).
Next, in \cref{ssec:closure-exploration}, we will define objects related to closures that will be central to the tree decomposition improvement algorithm --- \emph{blockages}, \emph{explorations} and \emph{collected components} --- as well as prove several structural properties of these notions.
Finally, \cref{ssec:computing-closures} sketches how to find small closures efficiently in a~dynamically changing tree decomposition.


\subsection{Small Closure Lemma}
\label{ssec:small-closure-lemma}

We now formally define the notion of small closures.
\begin{definition}
  Let $(T,\bag)$ be a~tree decomposition of $G$ and $\Tpref$ be a~prefix of $T$.
  Let also $c \in \N$ be an~integer.
  Then we say that a set $X \supseteq \bags(\Tpref)$ is \emph{$c$-small} with respect to $(T,\bag)$ if for every appendix $t$ of $\Tpref$, it holds that $|X \cap \component{t}| \leq c$.
\end{definition}

With this definition in place, we are ready to state the Small Closure Lemma, asserting the existence of $c$-small closures for $c$ large enough:

\begin{lemma}[Small Closure Lemma]
\label{lem:small-closure-lemma}
  There exists a~function $g(\ell) \in \Oh{\ell^4}$ such that the following holds.
  Let $k, \ell \in \N$ with $k \leq \ell$ and $G$ be a~graph with $\tw{G} \leq k$.
  Let $(T, \bag)$ be a~tree decomposition of $G$ of width at most $\ell$ and $\Tpref \subseteq V(T)$ be a~prefix of~$T$.
  Then there exists a~$g(\ell)$-small $k$-closure of $\bags(\Tpref)$ with respect to $(T,\bag)$.
\end{lemma}

The proof of \cref{lem:small-closure-lemma} uses the machinery of Bojańczyk and Pilipczuk~\cite{DBLP:journals/lmcs/BojanczykP22} in the form of the Dealternation Lemma: intuitively, since $T$ is a~bounded-width decomposition of $G$, there exists a~well-structured tree decomposition $U$ of $G$ such that for each appendix $t$ of \Tpref, the component $\component{t}$ can be partitioned into a~bounded number of well-structured ``chunks'' of $U$.
The closure $X$ will be constructed so that each chunk of $U$ contains only a~bounded number of vertices from $X$. Thus, $X$ will include a~bounded number of vertices from each $\component{t}$.

We now present a~formal version of the Dealternation Lemma.
The description follows the exposition in \cite{DBLP:journals/lmcs/BojanczykP22}, with some details irrelevant to us omitted.

\paragraph{Elimination forests.}
The output of the Dealternation Lemma is a~tree decomposition of $G$ presented as the so-called \emph{elimination forest}:

\begin{definition}
  An~\emph{elimination forest} of $G$ is a~rooted forest $F$ on vertex set $V(G)$ with the following property: if $uv \in E(G)$, then $u$ and $v$ are in the ancestor-descendant relationship in~$F$.
\end{definition}

The following definition shows how to turn an~elimination forest of $G$ into a~tree decomposition:


\begin{definition}
  \label{def:td-of-elim-forest}
  Assume that $F$ is an~elimination forest of $G$ (so $V(F) = V(G)$).
  A~tree decomposition \emph{induced} by $F$ is the tree decomposition $(F, \bag)$, where for each $v \in V(G)$, we set $\bag(v)$ to contain $v$ and each ancestor of $v$ connected by an~edge of $G$ to any descendant of $v$.
\end{definition}

In the definition above, we slightly abuse the notation and allow the shape of a tree decomposition to be a rooted forest, rather than a rooted tree; all other conditions remain the same. Note that such a forest decomposition can be always turned into a tree decomposition of same width by selecting one root $r$ and making all other roots children of $r$.

That $(F, \bag)$ constructed as in \cref{def:td-of-elim-forest} is indeed a~tree decomposition of $G$ is argued in \cite[Section~3]{DBLP:journals/lmcs/BojanczykP22}.
It is now natural to define the \emph{width} of an~elimination forest $F$ as the width of the tree decomposition induced by $F$.
Clearly, each elimination forest has width lower-bounded by~$\tw{G}$.
On the other hand, every graph $G$ has an elimination forest of width exactly $\tw{G}$~\cite[Lemma~3.6]{DBLP:journals/lmcs/BojanczykP22}.

\paragraph{Factors.}
Intuitively, \emph{factors} are well-structured ``chunks'' of a~forest $F$.
Formally, a~factor is a~subset of $V(F)$ that is either:
\begin{itemize}
  \item a~\emph{forest factor}: a~union of a~nonempty set of rooted subtrees of $F$, whose roots are all siblings to each other; or
  \item a~\emph{context factor}: a~nonempty set of the form $\Phi_1 \setminus \Phi_2$, where $\Phi_1$ is a~rooted subtree and $\Phi_2 \subseteq \Phi_1$ is a~forest factor; the root of a~context factor is the root of $\Phi_1$, while the roots of the tree factors in $\Phi_2$ are called the \emph{appendices}.
\end{itemize}

\ms{picture}

\paragraph{Dealternation Lemma.}
We can now state the Dealternation Lemma.

\begin{lemma}[Dealternation Lemma, \cite{DBLP:journals/lmcs/BojanczykP22}]
\label{lem:dealternation-lemma}
  There exists a~function $f(\ell) \in \Oh{\ell^3}$ such that the following holds.
  Let $(T, \bag)$ be a~tree decomposition of $G$ of width at most $\ell$.
  Then there exists an~elimination forest $F$ of $G$ of width $\tw{G}$ such that for every node $t \in V(T)$, the set $\component{t}$ is a~disjoint union of at most $f(\ell)$ factors of $F$.
\end{lemma}

\paragraph{Proof of the Small Closure Lemma.}
We now show how the Dealternation Lemma implies the Small Closure Lemma (\cref{lem:small-closure-lemma}).

\begin{proof}[Proof of \cref{lem:small-closure-lemma}]
  Recall that we are given: a~graph $G$ of treewidth at most $k$; a~tree decomposition $(T, \bag)$ of $G$ of width at most $\ell$, where $\ell \geq k$; and a~prefix $\Tpref$ of $T$. We are supposed to find a~small $k$-closure $X$ of $\bags(\Tpref)$.
  We stress that this requires that $\tw{\torso[G]{X}} \leq 2k + 1$.

  We begin by applying~\cref{lem:dealternation-lemma} to $(T, \bag)$ and getting an~elimination forest $F$ of $G$ of width $\tw{G}$, for which each $\component{t}$ for $t \in V(T)$ can be decomposed into at most $f(\ell) \in \Oh{\ell^3}$ factors of $F$.
  We remark that $V(F) = V(G)$.
  
  As argued in the paragraph following \cref{def:td-of-elim-forest}, there exists a~tree decomposition $(F, \bag')$ of $G$ of width $\tw{G}$, where $\bag'(v)$ for $v \in V(G)$ is defined as the set containing $v$ and every ancestor of $v$ in $F$ incident to an~edge whose other endpoint is a~descendant of $v$.
  Let $W \coloneqq \bags_T(\Tpref)$ and $W'$ be the lca-closure of $W$ in $F$.
  We claim that the set
    \[ X \coloneqq \bigcup_{v \in W'} \bag'(v) \]
  is an $((\ell+1) \cdot f(\ell))$-small $k$-closure of $W$ with respect to $T$.
  This will conclude the proof.

  \begin{claim}
    \label{clm:scl-closure-verification}
    $X$ is a~$k$-closure of $\bags_T(\Tpref)$.
    \begin{claimproof}
      Since $\bags_T(\Tpref) \subseteq W'$ and $v \in \bag'(v)$ for each $v \in W'$, we have that $\bags_T(\Tpref) \subseteq X$, as required.
      It remains to show that $\tw{\torso[G]{X}} \leq 2k + 1$.
      However, as $W'$ is lca-closed in $F$, \cref{lem:tw-torso-for-sum-of-bags} applies to $W'$ and (every component of) tree decomposition $(F, \bag')$, finishing the proof.
    \end{claimproof}
  \end{claim}
  
  \begin{claim}
    \label{clm:scl-factor-small-intersection}
    Let $t$ be an~appendix of $\Tpref$ and let $\Phi$ be a~factor of $F$ with $\Phi \subseteq \component{t}$.
    Then
    \[ |\Phi \cap X| \leq \ell+1. \]
    \begin{claimproof}
      Since $t$ is an~appendix of $\Tpref$, it follows from the definition of $\component{t}$ that $\component{t}$ is disjoint with $\bags_T(\Tpref)$.
      Thus, $\Phi$ is also disjoint with $\bags_T(\Tpref)$.
      
      First, assume that $\Phi$ is a~forest factor.
      Then $\Phi$ is downwards closed: if a~vertex belongs to $\Phi$, then all its descendants also belong to $\Phi$.
      Since $W'$ consists of $\bags_T(\Tpref)$ and a~subset of ancestors of $\bags_T(\Tpref)$, it follows that $\Phi$ is disjoint with $W'$.
      Now, for each $v \in W'$, the set $\bag'(v)$ comprises $v$ and some ancestors of $v$.
      So again, $\Phi$ is disjoint with each set $\bag'(v)$ for $v \in W'$ and thus disjoint with $X$.
      
      Now assume that $\Phi$ is a~context factor.
      Recall that $\Phi = \Phi_1 \setminus \Phi_2$, where $\Phi_1$ is a~subtree of $F$ rooted at some vertex $r$, and $\Phi_2 \subseteq \Phi_1$ is a~forest factor.
      The appendices of $\Phi$ have a~common parent, which we call $s$.
      By the disjointness of $\bags_T(\Tpref)$ with $\Phi$, we see that each vertex of $\bags_T(\Tpref)$ is either outside of $\Phi_1$ or inside some subtree of $\Phi_2$.

      Since $W'$ is the~lca-closure of $\bags_T(\Tpref)$, we have that $W' = \{\lca{u}{v}\,\colon\,u, v \in \bags_T(\Tpref)\}$.
      Consider $u, v \in \bags_T(\Tpref)$ such that $\lca{u}{v} \in \Phi_1$.
      If either $u$ or $v$ is outside of $\Phi_1$, then $\lca{u}{v}$ is also outside of $\Phi_1$.
      Therefore, both $u$ and $v$ belong to $\Phi_2$.
      In this case, it can be easily seen that $\lca{u}{v}$ either belongs to $\Phi_2$ (if both $u$ and $v$ are from the same rooted tree of $\Phi_2$) or is equal to $s$ (otherwise).
      Thus,
      \[ \Phi_1 \cap W' \subseteq \Phi_2 \cup \{s\}. \]
      Note that $\Phi_2 \cup \{s\}$ is a~connected subgraph of $F$ containing $s$ and some rooted subtrees attached to~$s$.

      Again, for each $v \in W'$, the set $\bag'(v)$ comprises $v$ and some ancestors of $v$.
      Hence, for $v \in W' \setminus \Phi_1$, the set $\bag'(v)$ is disjoint with $\Phi$.
      Therefore,
      \[ \Phi \cap X \subseteq \bigcup \{ \Phi \cap \bag'(v) \,\colon\, v \in \Phi_2 \cup \{s\} \}. \]

      Now let $v \in \Phi_2 \cup \{s\}$ and $x \in \Phi \cap \bag'(v)$.
      By the definition of $\bag'$, $x$ is either:
      \begin{itemize}
        \item equal to $v$. Then we have that $x \in \Phi_2 \cup \{s\}$ and $x \in \Phi$, so necessarily $x = s$; or
        \item an~ancestor of $v$ that is connected by an~edge to a~descendant $y$ of $v$.
          But since also $x \in \Phi$, we get that $x$ is also an~ancestor of $s$.
          Also, $y$ is a~descendant of $s$ (as $y$ is a~descendant of $v$ and $v$ is a~descendant of $s$).
          We conclude that $x \in \bag'(s)$.
      \end{itemize}
      
      In both cases we have $x \in \bag'(s)$ and therefore
      \[ \Phi \cap X \subseteq \bag'(s). \]
      \ms{consider a picture somewhere in the proof}
      As $(F, \bag')$ is a~decomposition of width $\tw{G} \leq k \leq \ell$, the statement of the claim follows immediately.
    \end{claimproof}
  \end{claim}
  
  \begin{claim}
    \label{clm:scl-closure-intersection}
    Let $t$ be an~appendix of $\Tpref$. Then
    \[ |\component{t} \cap X| \leq (\ell + 1)f(\ell). \]
    \begin{claimproof}
      By the Dealternation Lemma (\cref{lem:dealternation-lemma}), $\component{t}$ can be partitioned into at most $f(\ell)$ disjoint factors of $F$.
      By \cref{clm:scl-factor-small-intersection}, each such factor intersects $X$ in at most $\ell + 1$ elements.
    \end{claimproof}
  \end{claim}
  
  The proof of the Small Closure Lemma follows immediately from \cref{clm:scl-closure-verification,clm:scl-closure-intersection}.
\end{proof}

\subsection{Minimum-weight closures and Closure Linkedness Lemma}
\label{ssec:closure-linkedness-lemma}

Having established the existence of $c$-small closures for sufficiently large $c$, we now show a~structural result about such closures: the Closure Linkedness Lemma.
Intuitively, we prove that if $X \supseteq W$ is a~$c$-small closure of $W$ optimal with respect to some measure, then each connected component $C$ of $G - X$ is well-connected to $W$.
This property can be thought of as an~analog of a~similar result in a~work of Korhonen and Lokshtanov \cite[Lemma~5.1]{DBLP:journals/corr/abs-2211-07154}: the difference is that we work with optimal $c$-small closures, compared to just optimal closures in~\cite{DBLP:journals/corr/abs-2211-07154}.

From now on, let $\omega \colon V(G) \to \Z$ be an~arbitrary weight function.
For any subset $A \subseteq V(G)$, let $\omega(A) \coloneqq \sum_{v \in A} \omega(v)$.
First, let us define a~few notions involving weight functions:

\begin{definition}
  \label{def:d-minimal-closure}
  Fix $k, c \in \N$ and a~weight function $\omega\colon V(G) \to \Z$.
  Let $G$ be a~graph of treewidth at most $k$, $W \subseteq V(G)$, and $X$ be a~$c$-small $k$-closure of $W$.
  We say that $X$ is \emph{$\omega$-minimal} if for every $c$-small $k$-closure $X'$ of $W$, one of the following conditions holds:
  \begin{itemize}
    \item $|X'| > |X|$, or
    \item $|X'| = |X|$ and $\omega(X') \geq \omega(X)$.
  \end{itemize}
\end{definition}

\begin{definition}[\cite{DBLP:journals/corr/abs-2211-07154}]
  \label{def:d-linked}
  Let $G$ be a~graph, $A, B \subseteq V(G)$, and $\omega\,\colon\,V(G) \to \Z$ be a~weight function.
  We say that a~set $A$ is \emph{$\omega$-linked} into $B$ if each $(A,B)$-separator $S$ satisfies either of the following conditions:
  \begin{itemize}
    \item $|S| > |A|$,
    \item $|S| = |A|$ and $\omega(S) \geq \omega(A)$.
  \end{itemize}
\end{definition}

By definition, if $A$ is $\omega$-linked into $B$, then $A$ is also linked into $B$.

We can now state and prove the Closure Linkedness Lemma:

\begin{lemma}[Closure Linkedness Lemma]
  \label{lem:closure-linkedness-lemma}
  Fix $k, c \in \N$ and a~weight function $\omega\,\colon\,V(G) \to \Z$.
  Let $G$ be a~graph of treewidth at most $k$, $(T, \bag)$ be a~tree decomposition of $G$ (of any width), $\Tpref$ be a~prefix of $T$, and $X$ be an~$\omega$-minimal $c$-small $k$-closure of $\bags(\Tpref)$.
  Then for each $C \in \cc{G - X}$, the set $\neighopen{C}$ is $\omega$-linked into $\bags(\Tpref)$.
\end{lemma}

We remark that it follows from \cref{lem:closure-linkedness-lemma} and the Small Closure Lemma that if the width of the decomposition $(T, \bag)$ is $\ell \geq k$, then for some large enough constant $c_\ell \in \Oh{\ell^4}$ there exists a~$c_\ell$-small $k$-closure $X$ of $\bags(\Tpref)$ such that the neighborhood of each connected component of $G - X$ is $\omega$-linked into $\bags(\Tpref)$.
In fact, any such $\omega$-minimal closure will have this property.

The proof of the Closure Linkedness Lemma proceeds by assuming that some set $\neighopen{C}$ is not $\omega$-linked into $\bags(\Tpref)$; then, a~small $(\neighopen{C}, \bags(\Tpref))$-separator will exist.
This separator will be used to construct a~new $k$-closure $X'$ of $\bags(\Tpref)$ with smaller weight, thus contradicting the $\omega$-minimality of $X$.
However, in order to prove that $X'$ is a~$k$-closure of $\bags(\Tpref)$, one needs to show that $\torso[G]{X'}$ has sufficiently small treewidth.
In order to facilitate this argument, we will use a~useful technical tool from the work of Korhonen and Lokshtanov~\cite{DBLP:journals/corr/abs-2211-07154}:

\begin{lemma}[Pulling Lemma, {\cite[Lemma 4.8]{DBLP:journals/corr/abs-2211-07154}}]
  \label{lem:pulling-lemma}
  Let $G$ be a~graph, $X \subseteq V(G)$ and $(T, \bag)$ be a~tree decomposition of $\torso[G]{X}$.
  Let $(A, S, B)$ be a~separation of $G$ satisfying the following: there exists a~node $r \in V(T)$ such that $S$ is linked into $\bag(r) \cap (S \cup B)$.
  Let also $X' \coloneqq (X \cap A) \cup S$.
  Then, there exists a~tree decomposition $(T, \bag')$ of $\torso[G]{X'}$ of width not exceeding the width of $(T, \bag)$.
\end{lemma}

We also need a~simple helper lemma:

\begin{lemma}
  \label{lem:neighborhood-clique-lemma}
  Let $G$ be a~graph, $X \subseteq V(G)$ and $C \in \cc{G - X}$. Then $\neighopen{C}$ is a~clique in $\torso[G]{X}$.
  \begin{proof}
    For every $u, v \in \neighopen{C}$, there exists a~path connecting $u$ and $v$ whose all internal vertices are contained in $C$.
  \end{proof}
\end{lemma}

We are now ready to prove the Closure Linkedness Lemma.

\begin{proof}[Proof of \cref{lem:closure-linkedness-lemma}]
  Let $W \coloneqq \bags(\Tpref)$.
  For the sake of contradiction, assume we have a~component $C \in \cc{G - X}$ such that its neighborhood $\neighopen{C}$ is \emph{not} $\omega$-linked into $W$.
  By \cref{def:d-linked}, there exists an~$(\neighopen{C}, W)$-separator $S$ such that either $|S| < |\neighopen{C}|$, or $|S| = |\neighopen{C}|$ and $\omega(S) < \omega(\neighopen{C})$.
  Without loss of generality, assume that $S$ is such a~separator with minimum possible size.
  Naturally, $S$ induces a~separation $(A, S, B)$ such that $W \subseteq A \cup S$ and $\neighopen{C} \subseteq S \cup B$.
  
  Now, construct a~new set $X'$ from $X$ as follows:
  \[ X' \coloneqq (X \cap A) \cup S. \]
  We claim that $X'$ is also a~$c$-small $k$-closure of $W$.
  Note that since $\neighopen{C}$ is contained in $X$ and disjoint from $X \cap A$, we have that $X' \subseteq (X \setminus \neighopen{C}) \cup S$.
  Therefore, $|X'| \leq |X| - |\neighopen{C}| + |S|$ and $\omega(X') \leq \omega(X) - \omega(\neighopen{C}) + \omega(S)$. So if $|S|<|\neighopen{C}|$, we have $|X'|<|X|$, and if $|S|=|\neighopen{C}|$ and $\omega(S)<\omega(\neighopen{C})$, then $|X'|\leq |X|$ and $\omega(X')<\omega(X)$. So provided $X'$ is indeed a~$c$-small $k$-closure of $W$, $X$ cannot be $\omega$-minimal, contradicting our assumption.
  
  \begin{claim}
    \label{clm:cll-constructed-is-closure}
    $X'$ is a~$k$-closure of $W$.
    \begin{claimproof}
      Since $W$ is disjoint with $B$ and $X$ is a~$k$-closure of $W$, we get that $X' \supseteq W$.    
    
      Aiming to use the Pulling Lemma, we let $(T_X, \bag_X)$ to be a~tree decomposition of $\torso[G]{X}$ of width at most $2k + 1$.
      By \cref{lem:neighborhood-clique-lemma}, $\neighopen{C}$ is a~clique in $\torso[G]{X}$, so there exists a~node $r \in V(T_X)$ with $\neighopen{C} \subseteq \bag_X(r)$.
      It remains to verify that $S$ is linked into $\bag_X(r) \cap (S \cup B)$.

      Observe that $\neighopen{C} \subseteq \bag_X(r) \cap (S \cup B)$.
      Hence, it is enough to check that $S$ is linked into $\neighopen{C}$.
      However, if it was not the case, then there would be an~$(\neighopen{C}, S)$-separator $S'$ of size $|S'| < |S|$.
      But then $S'$ would be also an~$(\neighopen{C}, W)$-separator of size smaller than $|S|$, contradicting the minimality of $S$.
      Hence, \cref{lem:pulling-lemma} applies to the tree decomposition $(T_X, \bag_X)$ and the separation $(A, S, B)$, producing a~tree decomposition of $\torso[G]{X'}$ of width not exceeding $2k+1$.
    \end{claimproof}
  \end{claim}
  
  \begin{claim}
    \label{clm:cll-constructed-is-small}
    $X'$ is $c$-small.
    \begin{claimproof}
      Recall that $X \supseteq W$, where $W = \bags(\Tpref)$.
      Therefore, as $C$ is a~connected component of $G - X$, the entire connected component $C$ must be contained within $\component{t}$ for some appendix $t$ of $\Tpref$.
      Hence, $\neighopen{C} \subseteq \component{t} \cup \adhesion{t}$.
      Also, $\adhesion{t} \subseteq W$.
      As $(\component{t}, \adhesion{t}, V(G) \setminus (\component{t} \cup \adhesion{t}))$ is a~separation of $G$, and $S$ is a~minimum-size $(\neighopen{C}, W)$-separator, it follows that $S \subseteq \component{t} \cup \adhesion{t}$; in other words, vertices outside of $\component{t} \cup \adhesion{t}$ are not useful towards the separation of $\neighopen{C}$ from $W$.
      
      Now, for an~appendix $t' \neq t$ of $\Tpref$, $\component{t'}$ is disjoint from $\component{t} \cup \adhesion{t}$: this is because $\component{t'}$ is disjoint from $W$ (containing $\adhesion{t}$ in its entirety) and from $\component{t}$ (as $t$, $t'$ do not remain in the ancestor-descendant relationship in $T$).
      Therefore, $S$ is disjoint with $\component{t'}$ and thus, by the definition of $X'$,
      \[ |X' \cap \component{t'}| = |(X \cap A) \cap \component{t'}| \leq |X \cap \component{t'}| \leq c. \]
      Hence, the smallness condition is not violated for the appendix $t'$.
      
      In order to prove that the same condition is satisfied for the appendix $t$, we observe that $\neighopen{C} \cap W \subseteq S$ (as $S$ separates $\neighopen{C}$ from $W$), so also $\neighopen{C} \cap \adhesion{t} \subseteq S$.
      Therefore, $|\neighopen{C} \cap \adhesion{t}| \leq |S \cap \adhesion{t}|$, and we get that
      \[ |\neighopen{C} \cap \component{t}| = |\neighopen{C}| - |\neighopen{C} \cap \adhesion{t}| \geq |S| - |S \cap \adhesion{t}| = |S \cap \component{t}|. \]
      
      Now, as $X' \subseteq (X \setminus \neighopen{C}) \cup S$ and $\neighopen{C} \subseteq X$, we infer that
      \[ |X' \cap \component{t}| \leq |X \cap \component{t}| - |\neighopen{C} \cap \component{t}| + |S \cap \component{t}| \leq |X \cap \component{t}| \leq c.\qedhere \]
    \end{claimproof}
  \end{claim}
  By \cref{clm:cll-constructed-is-closure,clm:cll-constructed-is-small} we infer that $X'$ is a~$c$-small $k$-closure of $W$. This is a contradiction to the $\omega$-minimality of~$X$.
\end{proof}

\subsection{Closure exploration, blockages and collected components}
\label{ssec:closure-exploration}

In this section, we define several auxiliary objects that will be constructed in the tree decomposition improvement algorithm after a~suitable closure is found: explorations, blockages and collected components.
For the remainder of the section, let us fix a~tree decomposition $(T, \bag)$ of $G$ of width at most $\ell$, a~nonempty prefix \Tpref of $T$ and a~$\omega$-minimal $c$-small $k$-closure $X \supseteq \bags(\Tpref)$ for some weight function $\omega\colon \N \to \N$.

\paragraph{Blockages.}
We start with the~definition of a~blockage:

\begin{definition}
  \label{def:blockage}
  We say that a~node $t \in V(T) \setminus \Tpref$ is a~\emph{blockage} in $T$ with respect to $\Tpref$ and $X$ if one of the following cases holds:
  \begin{itemize}
    \item $\bag(t) \subseteq \neighclosed[G]{C}$ for some component $C \in \cc{G - X}$ that intersects $\bag(t)$ \emph{(component blockage)};
    \item $\bag(t) \subseteq X$ and $\bag(t)$ is a~clique in $\torso[G]{X}$ \emph{(clique blockage)};
  \end{itemize}
  and no strict ancestor of $t$ is a~blockage.
\end{definition}

Note that a~component blockage intersects with exactly one  component $C \in \cc{G - X}$.

Since $T$, $\Tpref$ and $X$ will usually be known from the context, we will usually just say that $t$ is a~blockage.
Then, we introduce the following notation: let $\blockages{\Tpref}{X}$ be the set of blockages.

\paragraph{Properties of blockages.}
We now prove a~few lemmas relating blockages to $X$.

\begin{lemma}
  \label{lem:blockage-all-internally-disjoint-paths}
  If $t \in \blockages{\Tpref}{X}$, then for every pair of vertices $u, v \in \bag(t)$ there exists a~$uv$-path in $G$ internally disjoint with $X$.
  \begin{proof}
    If $t$ is a~component blockage for component $C \in \cc{G - X}$, then the statement of the lemma immediately follows from $\bag(t) \subseteq \neighclosed[G]{C}$: in fact, for all $u, v \in \bag(t)$, there exists a~$uv$-path whose all internal vertices belong to $C$.
    On the other hand, if $t$ is a~clique blockage, then the statement is equivalent to the assertion that $\bag(t)$ is a~clique in $\torso[G]{X}$.
  \end{proof}
\end{lemma}

Note that it immediately follows from \cref{lem:blockage-all-internally-disjoint-paths} that if $t$ is a~blockage, then $\bag(t) \cap X$ is a~clique in $\torso[G]{X}$.

\begin{lemma}
  \label{lem:no-closure-under-blockage}
  If $t \in \blockages{\Tpref}{X}$, then $\component{t} \cap X = \emptyset$.
  
  \begin{proof}
    Assume otherwise.
    Then, we claim that the set $X' \coloneqq X \setminus \component{t}$ is also a~$c$-small $k$-closure of $\bags(\Tpref)$, contradicting the minimality of $X$.
    In fact, we will show that $\torso[G]{X'}$ is a~subgraph of $\torso[G]{X}$.
    Towards this goal, choose vertices $u, v \in X'$ and assume that $uv \in E(\torso[G]{X'})$, i.e., there exists a~path $P = v_1, v_2, \dots, v_p$ from $u=v_1$ to $v=v_p$ internally disjoint with $X'$.
    If $P$ is internally disjoint with $X$, then also $uv \in E(\torso[G]{X})$, finishing the proof.
    In the opposite case, $P$ must intersect $X \setminus X'$ and thus intersect $\component{t}$.
    Note that by the definition of $X'$, we necessarily have $u, v \notin \component{t}$.
    Let $v_a, v_b$ ($1 \leq a \leq b \leq p$) be the first and the last intersection of $P$ with $\adhesion{t}$, respectively; such vertices exist since $(\component{t}, \adhesion{t}, (V(G) \setminus (\component{t} \cup \adhesion{t}))$ is a~separation of~$G$.
    The subpaths $v_1, v_2, \dots, v_a$ and $v_b, v_{b+1}, \dots, v_p$ are disjoint with $\component{t}$ and thus are internally disjoint with $X$.
    Construct a~new path $P'$ from $u$ to $v$ by concatenating:
    \begin{itemize}
      \item the subpath $v_1, v_2, \dots, v_a$,
      \item a~path from $v_a$ to $v_b$ internally disjoint with $X$ (its existence is asserted by \cref{lem:blockage-all-internally-disjoint-paths}),
      \item the subpath $v_b, v_{b+1}, \dots, v_p$.
    \end{itemize}
    Naturally, $P'$ is again internally disjoint with $X'$.
    We claim that $P'$ is internally disjoint with $X$, thus witnessing that $uv \in E(\torso[G]{X})$.
    Indeed, each of the three segments of $P'$ is internally disjoint with $X$, so $P'$ can internally intersect $X$ only if $v_a \in X$ (or $v_b \in X$).
    However, in this case, we have that $v_a \in X'$ (respectively, $v_b \in X'$), as $v_a, v_b \notin \component{t}$.
    But $P'$ is internally disjoint with $X'$, hence $v_a$ (resp., $v_b$) must be the first (resp., the last) vertex of $P'$ and thus not an~internal vertex of $P'$.
    Hence, $P'$ is internally disjoint with $X$.
  \end{proof}
\end{lemma}

We remark that if $t$ is a~clique blockage, then from \cref{lem:no-closure-under-blockage} it follows that $\bag(t) = \adhesion{t}$, so in particular $\bag(t) \subseteq \bag(\parent{t})$.



\newcommand{\compression}{\xi}
\paragraph{Exploration and exploration graph.}
For a~given weight function $\omega\colon V(G) \to \Z$, prefix $\Tpref$ of $T$ and a~$\omega$-minimal $c$-small $k$-closure $X$ of $\bags(\Tpref)$, we define the \emph{exploration} as the prefix $\exploration{\Tpref}{X}$ of $T$ whose set of appendices is given by $\blockages{\Tpref}{X}$.
A~node $t \in V(T)$ is deemed \emph{explored} if $t \in \exploration{\Tpref}{X}$, otherwise it is \emph{unexplored}.
Then, a~vertex $v \in V(G)$ is explored if it belongs to $\bag(t)$ for some explored node $t$, and unexplored otherwise.
Observe that $v \in V(G)$ is unexplored if and only if it belongs to $\component{t}$ for some blockage $t$.
In particular, by \cref{lem:no-closure-under-blockage}, every vertex of $X$ is explored.

Clearly, we have $\Tpref \subseteq \exploration{\Tpref}{X}$.
Also, if $T$ is a~binary tree decomposition, then $|\blockages{\Tpref}{X}| \leq |\exploration{\Tpref}{X}| + 1$ (this follows immediately from \cref{fact:number-of-appendices}).

Next, we define the~\emph{exploration graph} $H \coloneqq \explorationgraph{\Tpref}{X}$ by compressing the components $\component{t}$ of blockages $t$ to single vertices.
The purpose of this definition will be to present the connected components of $G - X$ in a way that can be bounded by $|\exploration{\Tpref}{X}|$, even if the number of such components can be much larger.
Formally:
\begin{itemize}
  \item $V(H)$ comprises: \emph{explored vertices}, that is, the set of vertices $\bigcup_{t \in \exploration{\Tpref}{X}} \bag(t)$; and \emph{blockage vertices}, that is, the set $\blockages{\Tpref}{X}$ of nodes of $T$.
  \item The set of edges $E(H)$ is constructed by taking the~subgraph of $G$ induced by the explored vertices and adding edges $tv$ for each $t \in \blockages{\Tpref}{X}$ and $v \in \adhesion{t}$.
  \item The \emph{compression mapping} $\compression\colon V(G) \to V(H)$ is an~identity mapping on $V(G) \cap V(H)$; and for $t \in \blockages{\Tpref}{X}$, we set $\compression^{-1}(t) \coloneqq \component{t}$.
\end{itemize}

Note that \cref{lem:no-closure-under-blockage} implies that $X \subseteq V(H)$.
Also, if $T$ is a~binary tree decomposition, then $$|V(H)| \leq (\ell+1) \cdot |\exploration{\Tpref}{X}| + |\exploration{\Tpref}{X}| + 1.$$

We also observe the following fact:

\begin{lemma}
  \label{lem:compression-edge-mapping}
  If $uv \in E(G)$, then $\compression(u) = \compression(v)$ or $\compression(u)\compression(v) \in E(H)$.
  \begin{proof}
    If both $u$ and $v$ are explored, then $\compression(u) = u$, $\compression(v) = v$ and $\compression(u)\compression(v) \in E(H)$.
    Otherwise, assume without loss of generality that $\compression(v) = t$ is a~blockage vertex.
    Then $v \in \component{t}$, so $u \in \component{t} \cup \adhesion{t}$.
    If $u \in \component{t}$, then $\compression(u) = \compression(v) = t$.
    Otherwise $u \in \adhesion{t}$, so $u$ is an~explored vertex.
    Then $\compression(u) = u$ and $\compression(v) = t$, so $\compression(u)\compression(v) \in E(H)$ by the construction.
  \end{proof}
\end{lemma}

\paragraph{Collected components.}
Using the notation above, we partition $V(G) \setminus X$ into subsets called \emph{collected components} as follows. For each $C \in \cc{H - X}$, let $\collectedsym \coloneqq \compression^{-1}(C)$ (so $\collectedsym \subseteq V(G) \setminus X$) and if $\collectedsym$ is nonempty, then include \collectedsym in the set $\collectedcomps{\Tpref}{X}$ of collected components of $G$.
In other words, we uncompress each connected component of $H - X$ to form a~collected component --- a~subset of $V(G) \setminus X$.
It can be easily verified that the set $\collectedcomps{\Tpref}{X}$ is, indeed, a~partitioning of $V(G) \setminus X$.

Let $\collectedsym \in \collectedcomps{\Tpref}{X}$.
We say that \collectedsym is:
\begin{itemize}
  \item \emph{unblocked} if $\compression(\collectedsym)$ contains at least one~explored vertex of $V(H)$; or
  \item \emph{blocked} if $\compression(\collectedsym)$ only contains blockage vertices of $V(H)$.
\end{itemize}
Note that if $t$ is a~clique blockage, then $t$ is an~isolated vertex of $H - X$; hence, the collected component $\compression^{-1}(\{t\})$ (if it is nonempty) is necessarily a~blocked component.
However, the converse implication is not true in general: it could be that a~component blockage $t'$ satisfies $\adhesion{t'} \subseteq X$, again causing $t'$ to be an~isolated vertex of $H - X$.
Then, the collected component $\compression^{-1}(\{t'\})$ (if it is nonempty) will, too, be a~blocked component.

It turns out that $\collectedcomps{\Tpref}{X}$ is a~coarser partitioning of $V(G) \setminus X$ than $\cc{G - X}$:

\begin{lemma}
  \label{lem:collected-is-union-of-connected}
  Each collected component $\collectedsym \in \collectedcomps{\Tpref}{X}$ is the~union of a collection of connected components of $G - X$.
\end{lemma}
\begin{proof}
    It is enough to prove the following: if $u,v \in V(G) \setminus X$ are adjacent in $G$, then $\compression(u)$ and $\compression(v)$ belong to the same connected component of $H$. But by \cref{lem:compression-edge-mapping}, $\compression(u)$ and $\compression(v)$ are either equal or adjacent in $H$; this proves the claim.
  \end{proof}
  
Next, we prove a~characterization of each type of collected component.

\begin{lemma}
  \label{lem:collected-types-description}
  If $\collectedsym \in \collectedcomps{\Tpref}{X}$ is:
  \begin{itemize}
    \item unblocked, then the set of explored vertices belonging to $\compression(\collectedsym)$ is contained in a~single connected component $C \in \cc{G - X}$.
    Moreover, $C \subseteq \collectedsym$ and $\neighopen{\collectedsym} = \neighopen{C}$;
    \item blocked, then $\compression(\collectedsym)$ consists of a~single blockage vertex $t$.
    Moreover, $\neighopen{\collectedsym} \subseteq \adhesion{t} \cap X$.
  \end{itemize}
  \begin{proof}
    We start by proving the following claim:
    \begin{claim}
      \label{clm:explored-vertices-in-same-component}
      If $u, v \in \collectedsym$ are explored, then $u$ and $v$ belong to the same connected component of $G - X$.
    \end{claim}      \begin{claimproof}
        Let $P = w_1w_2\dots{}w_p$ ($w_1 = \compression(u)$, $w_p = \compression(v)$) be a~simple path from $\compression(u)$ to $\compression(v)$ in $H - X$ (such a~path exists since $u$, $v$ are in the same collected component).
        We will prove inductively that each explored vertex of $P$ is in the same connected component of $G - X$ as $w_1$.
        This is trivially true for $w_1$.
        For an~inductive step, choose an~explored vertex $w_i$ ($i \geq 2$).
        If $w_{i-1}$ is also explored, then $w_{i-1}w_i$ is an~edge of $G - X$, so $w_i$ is in the same connected component of $G - X$ as $w_{i-1}$ (and thus in the same connected component of $G - X$ as $w_1$).

        If $w_{i-1}$ is, on the other hand, a~blockage vertex (i.e., $w_{i-1} \in V(T)$), then note that $i \geq 3$ (since $w_1$ is an~explored vertex).
        Moreover, by the definition of $E(H)$, $w_{i-1}$ is adjacent in $H$ only to explored vertices in $\adhesion{w_{i-1}}$.
        Therefore, $w_{i-2}$ is also an~explored vertex and $w_{i-2}, w_i \in \adhesion{w_{i-1}}$.
        Since $w_{i-2}, w_i \notin X$, it follows that $w_{i-1}$ is a~component blockage, i.e., $\bag(w_{i-1}) \subseteq \neighclosed[G]{C}$ for some component $C \in \cc{G - X}$ intersecting $\bag(t)$.
        Therefore, both $w_{i-2}$ and $w_i$ belong to $C$; and by the inductive assumption, so does $w_1$.
      \end{claimproof}

    Also, the following simple fact will be useful:
    
    \begin{claim}
      \label{clm:small-neighborhood-of-blocked-component}
      If $D$ is a~connected component of $G - X$ such that $D \subseteq \component{t}$ for a~blockage $t$, then $\neighopen{D} \subseteq \adhesion{t} \cap X$.
    \end{claim}
      \begin{claimproof}
        Obviously, $\neighopen{D} \subseteq X$.
        However, by \cref{lem:no-closure-under-blockage}, $\component{t} \cap X = \emptyset$.
        As $\neighclosed{D} \subseteq \component{t} \cup \adhesion{t}$, it follows that $\neighopen{D} \subseteq \adhesion{t}$.
      \end{claimproof}

    Now, assume that \collectedsym is unblocked.
    By \cref{clm:explored-vertices-in-same-component}, the set of explored vertices belonging to $\compression(\collectedsym)$ is contained in a single~connected component $C \in \cc{G - X}$.
    From \cref{lem:collected-is-union-of-connected}, we have that $C \subseteq \collectedsym$.
    It remains to show that $\neighopen{\collectedsym} = \neighopen{C}$.
    To this end, we will argue the following:
    \begin{claim}
      \label{clm:neighborhoods-contained-in-unblocked}
      Let $D \in \cc{G - X}$ be a~connected component with $D \subseteq \collectedsym$. Then $\neighopen{D} \subseteq \neighopen{C}$.
      \begin{claimproof}
        We can assume that $D \neq C$.
        Then the set $\compression(D)$ only contains blockage vertices; and since blockage vertices form an~independent set in $H$, it follows from \cref{lem:compression-edge-mapping} that $\compression(D)$ comprises just one blockage vertex, say $t$.
        Hence $D \subseteq \component{t}$, and it also follows from \cref{clm:small-neighborhood-of-blocked-component} that $\neighopen{D} \subseteq \adhesion{t} \cap X$.
        
        Let $u \in C$ be any explored vertex (so $\compression(u) = u$).
        Now, since $C, D \subseteq \collectedsym$ and $\compression(D) = \{t\}$, there must exist a~path $P$ from $t$ to $u$ in $H - X$.
        Let $v$ be the (explored) vertex immediately after $t$ in $P$.
        By \cref{clm:explored-vertices-in-same-component}, $v \in C$; and $v \in \adhesion{t}$ by construction.
        It follows that $t$ is a~component blockage intersecting $C \in \cc{G - X}$, so $\bag(t) \subseteq \neighclosed{C}$; in particular, $\adhesion{t} \cap X \subseteq \neighopen{C}$.
      \end{claimproof}
    \end{claim}

    As $\neighopen{\collectedsym} = \bigcup_{D \in \cc{G - X},\, D \subseteq \collectedsym} \neighopen{D}$, we get from \cref{clm:neighborhoods-contained-in-unblocked} that $\neighopen{\collectedsym} = \neighopen{C}$, as required.
    
    Finally, assume that \collectedsym is blocked.
    Since blockage vertices of $H$ form an~independent set, it follows from \cref{lem:compression-edge-mapping} that $|\compression(\collectedsym)| = 1$.
    Let $t$ be the blockage vertex that is the only element of $\compression(\collectedsym)$.
    Then \cref{clm:small-neighborhood-of-blocked-component} applies to every $D \in \cc{G - X}$ with $D \subseteq \collectedsym$ and implies that $\neighopen{D} \subseteq \adhesion{t} \cap X$.
    Therefore, $\neighopen{\collectedsym} \subseteq \adhesion{t} \cap X$.
  \end{proof}
\end{lemma}

The next lemma follows easily from the previous lemma:

\begin{lemma}
  \label{lem:collected-neigh-clique}
  For each $\collectedsym \in \collectedcomps{\Tpref}{X}$, the set $\neighopen{\collectedsym}$ is a~clique in $\torso[G]{X}$.
  \begin{proof}
    If \collectedsym is unblocked, then $\neighopen{\collectedsym} = \neighopen{C}$ for some $C \in \cc{G - X}$.
    By \cref{lem:neighborhood-clique-lemma}, $\neighopen{C}$ is a~clique in $\torso[G]{X}$.
    On the other hand, if \collectedsym is blocked, then $\neighopen{\collectedsym} \subseteq \adhesion{t} \cap X$ for some blockage $t$.
    But it follows immediately from the definition of a~blockage that $\adhesion{t} \cap X$ is a~clique in $\torso[G]{X}$.
  \end{proof}
\end{lemma}

Finally, we define the \emph{home bag} of a~collected component \collectedsym to be a~node of $T$ selected as follows:
\begin{itemize}
  \item If \collectedsym is unblocked, then the home bag of \collectedsym is the shallowest node $t$ of $T$ such that $\bag(t)$ contains an~explored vertex in \collectedsym.
  Note that this bag is uniquely defined, as the explored vertices in \collectedsym induce a~connected subgraph of $G$ and hence their occurrences in $(T, \bag)$ induce a~connected subtree of $T$.
  \item If \collectedsym is blocked, then the home bag of \collectedsym is the unique blockage $t$ such that $\compression(\collectedsym) = \{t\}$.
\end{itemize}
We conclude by proving some properties of the home bags of collected components.
\begin{lemma}
  \label{lem:unblocked-home-bags}
  Let $\collectedsym \in \collectedcomps{\Tpref}{X}$ be unblocked and let $t$ be the home bag of \collectedsym.
  Then 
  $$t \in \exploration{\Tpref}{X} \setminus \Tpref\qquad\textrm{and}\qquad|\bag(t)| > |\neighopen{\collectedsym}|.$$
  \begin{proof}
    First note that $t$ is an~explored node: each explored vertex of \collectedsym is contained in some bag $\bag(u)$ for an~explored node $u$, and the set of explored nodes is ancestor-closed.
    Also, $t \notin \Tpref$ is immediate: we have $\bags(\Tpref) \subseteq X$ by the definition, but $\collectedsym$ is disjoint from $X$.
    We conclude that $t \in \exploration{\Tpref}{X} \setminus \Tpref$.
    
    Next we show that $|\bag(t)| > |\neighopen{\collectedsym}|$.
    Recall from \cref{lem:collected-types-description} that there is some $D \in \cc{G - X}$ that contains all explored vertices of \collectedsym and such that $\neighopen{\collectedsym} = \neighopen{D}$.
    Since $t$ is an~explored node of~$T$, $\bag(t) \setminus X$ comprises only of explored vertices.
    Hence, $t$ is the shallowest bag of $T$ intersecting $D$.
    Because $G[D]$ is connected, all bags of $T$ intersecting $D$ must be in the subtree of $T$ rooted at $t$.
    Additionally, by the Closure Linkedness Lemma (\cref{lem:closure-linkedness-lemma}), $\neighopen{D}$ is $\omega$-linked into $\bags(\Tpref)$.
    As $t \notin \Tpref$, we find that $\adhesion{t}$ is an~$(\neighopen{D}, \bags(\Tpref))$-separator, from which it follows that $|\adhesion{t}| \geq |\neighopen{D}|$.
    However, it cannot be that $\adhesion{t} = \bag(t)$ since $\adhesion{t} \cap D = \emptyset$ and $\bag(t) \cap D \neq \emptyset$.
    Thus, $\adhesion{t} \subsetneq \bag(t)$ and so $|\bag(t)| > |\neighopen{D}|$.
  \end{proof}
\end{lemma}

\begin{lemma}
  \label{lem:blocked-home-bags}
  Let $\collectedsym \in \collectedcomps{\Tpref}{X}$ be blocked and let $t$ be the home bag of \collectedsym.
  Assume that $\parent{t} \notin \Tpref$.
  Then $|\bag(t) \cap X| < |\bag(\parent{t})|$.
  \begin{proof}
    Since \collectedsym is blocked, $t$ is a~blockage.
    From \cref{lem:no-closure-under-blockage} we have $\component{t} \cap X = \emptyset$, so $\bag(t) \cap X = \adhesion{t} \cap X$.
    But it cannot be that $\adhesion{t} \cap X = \bag(\parent{t})$: otherwise, as $\adhesion{t} \cap X$ is a~clique in $\torso[G]{X}$ (\cref{lem:blockage-all-internally-disjoint-paths}), $\parent{t}$ would be a~clique blockage and then $t$ could not be a~blockage itself.
    Since $\adhesion{t} \subseteq \bag(\parent{t})$, we conclude that $|\bag(t) \cap X| < |\bag(\parent{t})|$.
  \end{proof}
\end{lemma}


\subsection{Computing closures and auxiliary objects}
\label{ssec:computing-closures}


We finally show that the objects defined in the previous sections can be computed efficiently in a~tree decomposition that is changing dynamically under prefix-rebuilding updates.
Recall that such decompositions are modeled by the annotated tree decompositions, and data structures maintaining such dynamic decompositions are called prefix-rebuilding data structures.
In this section, we fix a~specific weight function: for an annotated tree decomposition $\Tc=(T,\bag,\edges)$ of a graph $G$, as the weight function we use the depth function $d_\Tc$ of $\Tc$.

We are now ready to state the promised result.

\begin{lemma}\label{lem:closure-maintenance}
  For every $c, \ell, k \in \N$ with $\ell \geq k$, there exists an~$\ell$-prefix-rebuilding data structure with overhead $2^{\Oh{(c + \ell)^2}}$, additionally supporting the following operation:
  \begin{itemize}
    \item $\mathsf{query}(\Tpref)$: given a~prefix $\Tpref$ of $T$, returns either \noclosuremsg if there is no $c$-small $k$-closure of $\bags(\Tpref)$, or
      \begin{itemize}[nosep]
        \item a~$d_\Tc$-minimal $c$-small $k$-closure $X$ of $\bags(\Tpref)$,
        \item the graph $\torso[G]{X}$, and
        \item the set $\blockages{\Tpref}{X} \subseteq V(T)$.
      \end{itemize}
      This operation runs in worst-case time $2^{\Oh{k(c + \ell)^2}} \cdot (|\Tpref| + |\exploration{\Tpref}{X}|)$.
  \end{itemize}
\end{lemma}

The proof of \cref{lem:closure-maintenance} is provided in \cref{ssec:closures-and-blockages}.


\section{Refinement data structure}
\label{sec:refinement}
In this section we define a data structure, called the \emph{refinement data structure}, which will be used in \Cref{sec:height,sec:wrap-up} to improve a tree decomposition of a given dynamic graph.


Let $G$ be a graph of treewidth at most $k$, which is changing over time by edge insertions and deletions.
The ultimate goal of this work is to maintain an annotated binary tree decomposition $\Tc = (T, \bag, \edges)$ of $G$ of width at most $6k + 5$ (we usually write $\Tc = (T,\bag)$ when we are not using the $\edges$ function). The aim of the refinement operation is twofold.
First, it is used to improve the width of the decomposition, that is, when given a prefix $\Tpref$ of $(T,\bag)$ that contains all the bags of size more than $6k + 6$, the refinement operation outputs a tree decomposition $(T', \bag')$ of width at most $6k + 5$.
Second, it is called when the tree of the decomposition becomes \emph{too unbalanced}.
This shall be clarified later in \cref{sec:height}.

Before proceeding with the description of the refinement data structure, let us introduce the potential function which plays a major role in both the analysis of the amortized complexity of the refinement and in the height-reduction scheme of \Cref{sec:height}.

\subsection{Potential function}\label{sec:refine-potential}
Let $k \in \N$ be the upper bound on the treewidth of the dynamic graph, and $\ell = 6k+5$, which is the desired width of the tree decomposition we are maintaining (in particular, the actual width can be at most $6k+6$ when the refinement operation is called).
For a node $t$ of a tree decomposition $\Tc = (T, \bag)$, we define its potential by the formula:
\begin{equation}
\Phi_{\ell,\Tc}(t) \coloneqq g_\ell(|\bag(t)|) \cdot \height_T(t),
\end{equation}
where $$g_\ell(x) \coloneqq (53 (\ell + 4))^x\qquad \textrm{for every }x\in \N.$$
Observe that since we maintain a decomposition of width~$\Oh{k}$, we have $$\Phi_{\ell, \Tc}(t) \leq 2^{\Oh{k \log k}} \cdot \height_T(t)\qquad\textrm{for every node }t.$$
Intuitively, the term $g_\ell(|\bag(t)|)$ in the potential function allows us to update the tree decomposition by replacing the node $t$ with $\Oh{\ell}$ copies of $t$, where each copy $t'$ has the same height as $t$ but the bag of $t'$ is strictly smaller than that of $t$.
Then, we can pay for such a transformation from the decrease in the potential function.
The second term $\height_T(t)$ will turn out to be essential in \cref{sec:height}: it will allow to argue that the potential function can be decreased significantly if the current tree decomposition is too unbalanced.

For a subset $W \subseteq V(T)$, we denote
\[
\Phi_{\ell, \Tc}(W) \coloneqq \sum_{t \in W} \Phi_{\ell, \Tc}(t),
\]
and similarly, for the whole tree decomposition $\Tc = (T, \bag)$, we set
\[
\Phi_\ell(\Tc) \coloneqq \sum_{t \in V(T)} \Phi_{\ell, \Tc}(t).
\]

\subsection{Data structure}
\label{sec:refine-operation} 
The next lemma describes the interface of our data structure.

\begin{lemma}[Refinement data structure]\label{lem:refinement-data-structure}
  Fix $k \in \N$ and let $\ell = 6k + 5$.
  There exists an~$(\ell+1)$-prefix-rebuilding data structure with overhead $2^{\Oh{k^8}}$, that maintains a tree decomposition $\Tc = (T,\bag)$ with $N \coloneqq |V(T)|$ nodes, and additionally supports the following operation:
  \begin{itemize}
    \item $\mathsf{refine}(\Tpref)$: given a~prefix $\Tpref$ of $T$ so that $\Tpref$ contains all nodes of $\Tc$ with bags of size $\ell+2$, returns a description $\overline{u}$ of a prefix-rebuilding update, so that the tree decomposition $\Tc'$ obtained by applying $\overline{u}$ has the following properties:
  \end{itemize}
    \begin{itemize}
        \witem $\Tc'$ has width at most $\ell$ and
        \amitem{potential-difference} the following inequality holds:
        \begin{align*}
    \Phi_\ell(\Tc) - \Phi_\ell(\Tc') \geq \sum_{t \in \Tpref} \height_T(t) & - 2^{\Oh{k \log k}} \cdot \left( |\Tpref| + \sum_{t \in \App(\Tpref)} \height_T(t) \right) \cdot \log N.
        \end{align*}
    \end{itemize}
    Moreover, it holds that
    \begin{itemize}
        \rtitem{running-time} the worst-case running time of $\mathsf{refine}(\Tpref)$, and therefore also $|\overline{u}|$, is upper-bounded by
        \[
        2^{\Oh{k^{9}}} \cdot \left( \left( |\Tpref| + \sum_{t \in \App(\Tpref)} \height_T(t)\right) \cdot \log N + \max(\Phi_\ell(\Tc) - \Phi_\ell(\Tc'), 0) \right).
        \]
    \end{itemize}
\end{lemma}

In the remainder we present the data structure of \cref{lem:refinement-data-structure} and prove some key properties used later.
This section ends with a proof of the correctness of the refinement operation, including the proof of property \wref.
The amortized analysis of the data structure, in particular the proofs of properties \rtref{running-time} and \amref{potential-difference} will follow in \cref{sec:refine-analysis}.


Fix positive integers $k, \ell$ such that $\ell = 6k + 5$.
Let $G$ be a dynamic $n$-vertex graph of treewidth at most $k$.
In our data structure we store:
\begin{itemize}
\item a binary annotated tree decomposition $\Tc = (T, \bag, \edges)$ of $G$ of width at most $\ell + 1$;
\item an instance $\Daux$ of a data structure from \cref{lem:height-maintenance} for maintaining various auxiliary information about $\Tc$; and
\item an instance $\Dexplore$ of a data structure from \cref{lem:closure-maintenance} used to compute necessary objects for the refinement operation.
\end{itemize}

Implementation of the update operations on our data structure is simple. Upon receiving a prefix-rebuilding update $\tup u$ with a prefix $\Tpref$, we recompute the decomposition $(T,\bag,\edges)$ according to $\tup u$, and pass $\tup u$ to  the inner data structures $\Daux$ and~$\Dexplore$. The initialization of the data structure is similarly easy.

From now on, we focus on the refinement operation.
Let $\Tpref \subseteq V(T)$ be the given prefix of $T$ which includes all bags of $\Tc$ of size $\ell + 2$.

For the ease of presentation, instead of constructing a description $\overline{u}$ of a prefix-rebuilding update, we shall show how to construct a tree decomposition $(T', \bag')$ obtained from $(T, \bag)$ after applying $\overline{u}$.
The suitable description $\overline{u}$ can be then easily extracted from $(T',\bag')$ and from the objects computed along the way, in particular, \cref{lem:prds-strengthen} will be implicitly used here.



The refinement operation proceeds in five steps.


\myparagraph{\stepitem{aux} (Compute auxiliary objects)}
Given a prefix $\Tpref$ of the decomposition $\Tc=(T, \bag)$ we use the data structure $\Dexplore$ from \Cref{lem:closure-maintenance} to compute the following objects:

\begin{itemize}
\item a $d_\Tc$-minimal $c$-small $k$-closure $X$ of $\bags(\Tpref)$, where $d_\Tc$ is the depth function of $\Tc$ defined in \cref{sec:preliminaries} and $c \in \Oh{k^4}$ is the bound given \cref{lem:small-closure-lemma}, 
\item the graph $\torso[G]{X}$, and
\item the set $\blockages{\Tpref}{X} \subseteq V(T)$.
\end{itemize}
Then, we immediately apply the following \Cref{lem:explorationgraphcomputation}, provided below, to compute also
\begin{itemize}
\item the explored prefix $F \coloneqq \exploration{\Tpref}{X}$ and
\item the exploration graph $H \coloneqq \explorationgraph{\Tpref}{X}$.
\end{itemize}

\begin{lemma}
\label{lem:explorationgraphcomputation}
Let $\Tpref$ be a prefix of $T$ and $X$ a $d_\Tc$-minimal $c$-small $k$-closure of $\bags(\Tpref)$.
Given $\Tpref$, $X$, and $\blockages{\Tpref}{X}$, the explored prefix $F$ and the exploration graph $H$ can be computed in time $k^{\Oh{1}} \cdot |F|$.
\end{lemma}
\begin{proof}
First, $F$ can be computed in time $\Oh{|F|}$ by simply a depth-first search on $T$ that stops on blockages.
Then, the vertices $V(H)$ of the exploration graph can be computed in time $k^{\Oh{1}} \cdot |F|$ by taking the union of $\bags(F)$ and $\blockages{\Tpref}{X}$.
Having access to $(T,\bag,\edges)$, the induced subgraph $G[\bags(F)]$ can be computed in time $k^{\Oh{1}} \cdot |F|$, because all of its edges are stored in $\edges(F)$.
Also, the edges between blockage vertices and vertices in $\bags(F)$ can be directly computed from the definition.
\end{proof}

The running time of this step is clearly dominated by the call to $\Dexplore$, and by substituting $c = \Oh{k^4}$ in the bound on the running time of $\mathsf{query}$, we conclude the following.

\begin{itemize}
\rtitem{aux} The running time of \stepref{aux} is $2^{\Oh{k^9}} \cdot |F|$.
\end{itemize}

In further analysis of the refinement data structure, the following definitions will be useful.
\begin{definition}
\label{def:interface-weight-height}
For a collected component $\collectedsym \in \collectedcomps{\Tpref}{X}$, we define its \emph{interface}, \emph{weight}, and \emph{height}, denoted respectively by $\Cinterface$, $\weight(\collectedsym)$ and $\height(\collectedsym)$, as follows.
\begin{itemize}
  \item If $\collectedsym$ is unblocked, then set $\Cinterface \coloneqq \neighopen{\collectedsym}$.
  \item If $\collectedsym$ is blocked, then set $\Cinterface \coloneqq \bag(t) \cap X$, where $t \in \blockages{\Tpref}{X}$ is the only element of $\compression(\collectedsym)$ (i.e., $t$ is the home bag of $\collectedsym$).
\end{itemize}
In both cases, we set 
$$\weight(\collectedsym) \coloneqq |\Cinterface|\qquad\textrm{and}\qquad\height(\collectedsym) \coloneqq \height(t),$$ where $t$ is the home bag of $\collectedsym$.
\end{definition}

Also, we introduce the notation mapping each set $B \subseteq V(G)$ to the set of collected components with interface $B$:
$$ \invinterface{B} \coloneqq \{\collectedsym \in \collectedcomps{\Tpref}{X} \,\colon\, \Cinterface = B\}. $$

We remark that it follows from \cref{lem:collected-types-description} that $\neighopen{\collectedsym} \subseteq \Cinterface$.

\myparagraph{\stepitem{prefix} (Build a new prefix)}
At this step, we intend to construct a prefix of a new decomposition $\Tc' = (T', \bag')$.
As it turns out, the prefix will be a tree decomposition $\Tdnew{X}$ of $\torso[G]{X}$ with some additional properties.
The reason for using here $\torso[G]{X}$ is simple: for every connected component $C$ of $G - X$, the neighborhood $N(C) \subseteq X$ is a clique in $\torso[G]{X}$.
Therefore, by \cref{fact:tw-clique-bag}, for every tree decomposition $(\Tnew{X}, \bagnew{X})$ of $\torso[G]{X}$, there is a node $t \in \Tnew{X}$ whose bag $\bag(t)$ contains $N(C)$.
This will allow us to combine the tree decompositions of $\torso[G]{X}$ with tree decompositions of the components of $G - X$.


Let us give a detailed description of the constructed prefix $\Tdnew{X}$.
We start with stating the following two known results, due to Bodlaender, and Bodlaender and Hagerup, respectively.

\begin{theorem}[{\cite[Theorem 1]{bodlaender-tw-opt}}]
\label{thm:tw-opt}
Given a graph $G$, where $k \coloneqq \tw{G}$,
one can compute in time $2^{\Oh{k^3}} \cdot |V(G)|$ a~tree decomposition ${(T, \bag)}$ of~$G$ of width $k$ and with $|V(T)| = \Oh{|V(G)|}$.  
\end{theorem}

\begin{theorem}[{\cite[Lemma 2.2]{bodlaender-hagerup}}]
\label{thm:tw-shallow}
Given a graph $G$ and its tree decomposition $(T, \bag)$ of width $k'$,
one can compute in time $\Oh{k' \cdot |V(T)|}$ a binary tree decomposition $(T', \bag')$ of $G$ of height $\Oh{\log |V(T)|}$, width at most $3 k' + 2$, and with $|V(T')| = \Oh{|V(T)|}$.
\end{theorem}

Moreover, we will use another auxiliary operation modifying a binary tree decomposition.

\begin{lemma}
\label{lem:tw-all-subsets}
Given a graph $G$ and its binary tree decomposition $(T, \bag)$ of height $h$ and width $k'$,
one can compute in time $\Oh{2^{k'} \cdot |V(T)|}$ a binary tree decomposition $(T', \bag')$ of $G$ of height $\Oh{h + k'}$, width $k'$, and with $|V(T')| = \Oh{2^{k'} \cdot |V(T)|}$,
such that for every node $t \in V(T)$ and every subset $B \subseteq \bag(t)$, there is a leaf node $t_B$ of $T'$ such that $\bag'(t_B) = B$.
\end{lemma}

\begin{proof}
Let us iterate through all the nodes $t \in V(T)$.
Let $\bag(t) = \{ v_1, v_2, \ldots, v_{p} \}$, where $p \leq k' + 1$.

We construct a binary tree decomposition $(T_t, \bag_t)$ of the induced subgraph $\bag(t)$ as follows.
The root of $T_t$ is a new node $p$ with $\bag_t(p) = \bag(t)$.
The height of $T_t$ is $p$.
Each vertex $u \in V(T_t)$ at depth $i$, where $i = 0, 1, \ldots, p-1$, has two children $u_{\text{yes}}$ and $u_{\text{no}}$
with $\bag_t(u_{\text{yes}}) = \bag_t(u)$ and $\bag_t(u_{\text{no}}) = \bag_t(u) \setminus \{ v_{i+1} \}$.
These properties define the whole decomposition $(T_t, \bag_t)$.
This is a valid tree decomposition of $\bag(t)$ since whenever we exclude a vertex $v$ from a bag of $u \in V(T_t)$, then $v$ does not appear in any of the bags of descendants of $u$ in $T_t$.
Moreover, observe that for every subset $B \subseteq \bag(t)$, there is a leaf node $t_B \in V(T)$ such that $\bag_t(t_B) = B$.

It remains to attach such trees $T_t$ to the original tree $T$ --- this will form the desired decomposition $(T', \bag')$, where $\bag'|_T = \bag$ and $\bag'|_{T_t} = \bag_t$ for every $t$.
If a vertex $t$ has at most one child, then we can attach $T_t$ to $T$ by simply making the root of $T_t$ a child of $t$.
If $t$ has two children, then we can subdivide one of its edges to its children with a vertex $t_{\text{copy}}$, where $\bag'(t_{\text{copy}}) = \bag(t)$, and apply the previous case to $t_{\text{copy}}$.

It is easy to see that the width of $(T', \bag')$ is still $k'$, and the height of $T'$ is at most $\Oh{h + k'}$ as every attached tree $T_t$ has height at most $k' + 1$.
\end{proof}

We are ready to define the desired decomposition $\Tdnew{X}$ of $\torso[G]{X}$.
We perform the following operations.
\begin{enumerate}
    \item Apply \cref{thm:tw-opt} to obtain a tree decomposition $\myopt{\Tdnew{X}}$ of $\torso[G]{X}$.
    \item Use the algorithm from \cref{thm:tw-shallow} on the decomposition $\myopt{\Tdnew{X}}$ and call the resulting decomposition $\myshallow{\Tdnew{X}}$.
    \item Run the algorithm from \cref{lem:tw-all-subsets} to transform the decomposition $\myshallow{\Tdnew{X}}$ into the desired tree decomposition $\Tdnew{X} = (\Tnew{X}, \bagnew{X})$.
\end{enumerate}





Let us list all the properties of this step that will be required in further parts of the algorithm.
The purpose of property \pref{prefix-subsets} might be unclear at this point, but it will play an important role in the analysis of the amortized running time.

\begin{lemma}
The following properties hold.
\label{lem:refine-step1}
    \begin{itemize}
    \pitem{prefix-tree} $\Tdnew{X}$ is a binary tree decomposition of $\torso[G]{X}$ of width at most $6k+5$ and height at most $\Oh{k + \log N}$.
    \pitem{prefix-size} $V(\Tnew{X}) \leq 2^{\Oh{k}} \cdot |\Tpref|$.
    \pitem{prefix-subsets} For every collected component $\collectedsym \in \collectedcomps{\Tpref}{X}$, there exists a leaf $t_{\collectedsym} \in V(\Tnew{X})$ such that $\bagnew{X}(t_{\collectedsym}) = \Cinterface$.
    \rtitem{prefix} The running time of \stepref{prefix} is $2^{\Oh{k^3}} \cdot |\Tpref|$.
\end{itemize}
\end{lemma}

\begin{proof} We prove the consecutive points of the lemma.

\begin{claim}
    $\Tdnew{X}$ is a binary tree decomposition of $\torso[G]{X}$ of width at most $6k+5$ and height $\Oh{k + \log N}$.
    Additionally, $V(\Tnew{X}) \leq 2^{\Oh{k}} \cdot |\Tpref|$.
\end{claim}

\begin{claimproof}
    Recall that since $X$ is a~$k$-closure, we have $\tw{\torso[G]{X}} \leq 2k + 1$.
    Hence, we can compute the decomposition $\myopt{\Tdnew{X}}$ of width $2k + 1$.
    Consequently, the width of the decomposition $\myshallow{\Tdnew{X}}$ is bounded by $3 \cdot (2k + 1) + 2 = 6k + 5$.
    This is also the bound on the width of $\Tdnew{X}$, as applying \cref{lem:tw-all-subsets} does not increase the width of a decomposition.

    Next, observe that $|X| \le \Oh{k^4 \cdot |\Tpref|}$ because $X$ is $\Oh{k^4}$-small, implying that $|V(\Topt{X})| \le \Oh{k^4 \cdot |\Tpref|}$.
    It follows that the height of $\myshallow{\Tdnew{X}}$ is bounded by $\Oh{\log |X|} = \Oh{k + \log N}$.
    Therefore, by \cref{lem:tw-all-subsets}, the height of $\Tdnew{X}$ is indeed bounded by $\Oh{k + \log N}$.

    Finally, for the size of $\Tdnew{X}$, we have that $ |V(\Tnew{X})| \leq 2^{\Oh{k}} \cdot |V(\Tshallow{X})| \leq 2^{\Oh{k}} \cdot |\myopt{\Tnew{X}}| \leq 2^{\Oh{k}} \cdot |\bags(\Tpref)|$.
\end{claimproof}

\begin{claim}
For every collected component $\collectedsym \in \collectedcomps{\Tpref}{X}$, there is a leaf $t_{\collectedsym} \in \Tnew{X}$ such that $\bagnew{X}(t_{\collectedsym}) = \Cinterface$.
\end{claim}

\begin{claimproof}
Consider a collected component $\collectedsym \in \collectedcomps{\Tpref}{X}$.
First, observe that $\Cinterface$ is a clique in $\torso[G]{X}$.
Indeed, if $\collectedsym$ is unblocked, then
by \cref{lem:collected-neigh-clique}, $\Cinterface = \neighopen{\collectedsym}$ is a clique in $\torso[G]{X}$.
If $\collectedsym$ is blocked, then $\Cinterface = \bag(t) \cap X$, where $t$ is the home bag of~$\collectedsym$.
However, $\bag(t) \cap X$ is a clique in $\torso[G]{X}$ as well, since $t$ is a blockage.

Hence, by \cref{fact:tw-clique-bag}, there is a node $t \in V(\Tshallow{X})$ such that $\Cinterface \subseteq \bagnew{X}(t)$.
Then, by \cref{lem:tw-all-subsets}, we know that there is a leaf $t_{\collectedsym} \in V(\Tnew{X})$ such that $\bagnew{X}(t_{\collectedsym}) = \Cinterface$.
\end{claimproof}

\begin{claim}
    The running time of \stepref{prefix} is $2^{\Oh{k^3}} \cdot |\Tpref|$.
\end{claim}

\begin{claimproof}
    According to previous observations, computing $\myopt{\Tdnew{X}}$ takes $2^{\Oh{k^3}} \cdot |X|$ time, and since $|X| \leq k^{\Oh{1}} |\Tpref|$,
    we can write this bound as $2^{\Oh{k^3}} \cdot |\Tpref|$.
    The running time of the other two steps can be upper-bounded by $2^{\Oh{k}} \cdot|X|$, and thus the claim follows.
\end{claimproof}

The claims above verify all the required properties, so the proof is complete.
\end{proof}

\myparagraph{\stepitem{split} (Split the old appendices)}
So far, we have constructed a tree decomposition $\Tdnew{X}$ of $\torso[G]{X}$.
Now, we are going to construct for each collected component $\collectedsym \in \collectedcomps{\Tpref}{X}$ a tree decomposition $\Tdnew{\collectedsym}$ that will be attached to the decomposition $\Tdnew{X}$.
In particular, $\Tdnew{\collectedsym}$ will be a tree decomposition of the graph $G[\collectedsym \cup \Cinterface]$ whose root bag contains $\Cinterface$.
We will first describe the construction, then prove its required properties, and then argue that the relevant objects for constructing $\Tdnew{\collectedsym}$ for all $\collectedsym \in \collectedcomps{\Tpref}{X}$ can be computed in time $k^{\Oh{1}} \cdot |F|$.

We now describe the construction of $\Tdnew{\collectedsym}$.
Let $\collectedsym$ be a collected component, and recall the definition of home bag from \Cref{ssec:closure-exploration}.
We say that a blockage $t \in \blockages{\Tpref}{X}$ is \emph{associated} to $\collectedsym$ if either (1) $\collectedsym$ is a blocked component whose home bag $t$ is, or (2) $\collectedsym$ is an unblocked component such that $\collectedsym \cap \bag(t)$ is non-empty.
Note that the definitions of blockage and collected component imply that each blockage is associated with at most one collected component.

Now, if $\collectedsym$ is a blocked component, we define $\Tdnew{\collectedsym} = (\Tnew{\collectedsym}, \bag^{\collectedsym})$ so that $\Tnew{\collectedsym}$ is a copy of the subtree of $T$ rooted at the home bag of $\collectedsym$, and $\bag^{\collectedsym}$ is similarly a copy of the $\bag$ function in this subtree.
If $\collectedsym$ is an unblocked component, $\Tdnew{\collectedsym} = (\Tnew{\collectedsym}, \bag^{\collectedsym})$ is constructed as follows.
For a blockage $b \in \blockages{\Tpref}{X}$ we denote by $T_b$ the subtree of $T$ rooted at $b$.
We denote by $\mathsf{AB}(\collectedsym)$ the blockages associated with $\collectedsym$.
Then, the tree $\Tnew{\collectedsym}$ is defined as
\[\Tnew{\collectedsym} = T[\{t \in F \mid \collectedsym \cap \bag(t) \neq \emptyset\} \cup \bigcup_{b \in \mathsf{AB}(\collectedsym)} V(T_b)].\]
That is, $\Tnew{\collectedsym}$ is the subtree of $T$ consisting of (1) the explored nodes whose bags contain vertices in $\collectedsym$ and (2) the subtrees rooted at blockages that are associated with $\collectedsym$.
We remark that $\Tnew{\collectedsym}$ is connected by the definition of a collected component, and also that the unique highest node in $\Tnew{\collectedsym}$ (in particular, the root of $\Tnew{\collectedsym}$ since $\Tnew{\collectedsym}$ is connected) corresponds to the home bag of $\collectedsym$.

We observe that the trees $\Tnew{\collectedsym}$ across all (blocked and unblocked) collected components $\collectedsym$ satisfy the following properties:
\begin{itemize}[nosep]
 \item 
 each node in $V(T) \setminus F$ is in at most one tree $\Tnew{\collectedsym}$,
 \item each node in $F \setminus \Tpref$ is in at most $\ell+1$ trees $\Tnew{\collectedsym}$ (because all nodes in $F \setminus \Tpref$ have bags of size at most $\ell+1$), and
 \item no node in $\Tpref$ is in any tree $\Tnew{\collectedsym}$.
\end{itemize} 
In the rest of this section, for a node $t \in V(\Tnew{\collectedsym})$, we denote by $\origin(t)$ the corresponding node in~$T$.
Note that $\origin$ is a mapping from the union $\bigcup_{\collectedsym \in \collectedcomps{\Tpref}{X}} V(\Tnew{\collectedsym})$ to $V(T)$.

The bags of $\Tdnew{\collectedsym} = (\Tnew{\collectedsym}, \bag^{\collectedsym})$ for unblocked components $\collectedsym$ are defined as follows.
For a node $t \in V(\Tnew{\collectedsym})$ that has $\origin(t) \notin F$ (i.e., $\origin(t) \in V(T_b)$ for some blockage $b$), we set $\bag^{\collectedsym}(t) = \bag(\origin(t))$.
In other words, the subtrees rooted at blockages associated with $\collectedsym$ are just copied from $\Tc$ to $\Tdnew{\collectedsym}$ without any change.
Then, for a node $t \in V(\Tnew{\collectedsym})$ with $\origin(t) \in F$, we define 
\[\pullNei(t,\collectedsym) = N(\collectedsym) \cap (\component{\origin(t)} \setminus \bag(\origin(t))).\]
That is, $\pullNei(t,\collectedsym)$ consists of the vertices in $N(\collectedsym)$ that occur in the bags of the subtree of $(T,\bag)$ below $\origin(t)$ but not in $\bag(\origin(t))$.
Finally, for nodes $t \in \Tnew{\collectedsym}$ with $\origin(t) \in F$ we define
\[\bag^{\collectedsym}(t) = (\bag(\origin(t)) \cap N[\collectedsym]) \cup \pullNei(t,\collectedsym).\]
The purpose of having $\pullNei(t,\collectedsym)$ in $\bag^{\collectedsym}(t)$ is to ensure that $N(\collectedsym)$ is in the root of $\Tdnew{\collectedsym}$: for every $v \in \neighopen{\collectedsym}$ not in the root bag of $\Tdnew{\collectedsym}$, we add $v$ to every bag on the path between the root and the shallowest node of $\Tdnew{\collectedsym}$ containing $v$.

This concludes the definition of $\Tdnew{\collectedsym}$.
Next, we prove some elementary properties of this construction.

\begin{lemma}
Let \collectedsym be a~collected component.
\begin{itemize}
\pitem{component-tree} $\Tdnew{\collectedsym}$ is a rooted binary tree decomposition of $G[\collectedsym \cup \Cinterface]$.
\pitem{component-root} The root bag of $\Tdnew{\collectedsym}$ contains $\Cinterface$ as a~subset.
\pitem{component-height} The height of $\Tnew{\collectedsym}$ is at most the height of the home bag of \collectedsym.
\end{itemize}
\end{lemma}
\begin{proof}
First, when $\collectedsym$ is a blocked component, let $t \in V(T)$ be the home bag of $\collectedsym$.
All of the properties hold directly by the facts that $\collectedsym \subseteq \component{t}$, $\Cinterface \subseteq \bag(t)$, and $\bag(t) \subseteq \collectedsym \cup \Cinterface$.

Then, suppose that $\collectedsym$ is unblocked.
Recall that in this case, $\Cinterface = N(\collectedsym)$.
First, we have that $\Tnew{\collectedsym}$ with bags restricted to $\collectedsym$ is a tree decomposition of $G[\collectedsym]$ because all occurrences of the vertices in $\collectedsym$ in $(T,\bag)$ are in $T^{\collectedsym}$, and we never remove occurrences of the vertices in $\collectedsym$ when constructing the bags.
By the same argument, $\Tdnew{\collectedsym}$ also covers all edges between $\collectedsym$ and $N(\collectedsym)$.
Then, recall that by \Cref{lem:no-closure-under-blockage}, $\component{b} \cap X = \emptyset$ for each blockage $b$, implying that each $v \in N(\collectedsym)$ must be in some bag $\bag(t)$ of a node $t \in V(\Tnew{\collectedsym})$ with $\origin(t) \in F$.
This implies that $\Tnew{\collectedsym}$ satisfies the vertex condition also for vertices $v \in N(\collectedsym)$, even after the insertions of the sets $\pullNei(t,\collectedsym)$.
Note that these insertions ensure that $N(\collectedsym) = \Cinterface$ is contained in the root of $\Tdnew{\collectedsym}$, which shows that $\Tdnew{\collectedsym}$ with bags restricted to $\collectedsym \cup \Cinterface$ is a tree decomposition of $G[\collectedsym \cup \Cinterface]$.

To show that the bags of $\Tdnew{\collectedsym}$ do not contain vertices outside of $\collectedsym \cup \Cinterface$, observe that by the definition of an unblocked collected component, each blockage $b$ associated with $\collectedsym$ satisfies $\bag(b) \subseteq N[\collectedsym]$ and $\component{b} \subseteq \collectedsym$.
Hence, the bags below the blockages, which were copied verbatim, do not contain any elements of $V(G) \setminus \neighclosed{\collectedsym}$; and the explored bags of $\Tdnew{\collectedsym}$ have been explicitly truncated to $\neighclosed{\collectedsym}$.
This concludes the proof that $\Tdnew{\collectedsym}$ is a tree decomposition of $G[\collectedsym \cup \Cinterface]$.
The facts that $\Tdnew{\collectedsym}$ is a binary tree and that the height of $\Tnew{\collectedsym}$ is at most the height of the home bag of $\collectedsym$ follow directly from the construction.
\end{proof}

Then, we prove that $\Tdnew{\collectedsym}$ has width at most $\ell$. Moreover, an even stronger bound on the sizes of the bags holds, which will be crucial for bounding the potential function.
This is finally a proof where we use the assumption that $X$ is $d_\Tc$-minimal.

\begin{lemma}
\label{lem:exchange-argument}
Let $t \in V(\Tnew{\collectedsym})$.
\begin{itemize}
\pitem{collected-explored} If $\origin(t)$ is explored, then $|\bagnew{\collectedsym}(t)| < |\bag(\origin(t))|$ and $|\origininv(\origin(t))| \leq \ell + 1$.
\pitem{collected-unexplored} If $\origin(t)$ is unexplored, then $|\bagnew{\collectedsym}(t)| = |\bag(\origin(t))|$ and $|\origininv(\origin(t))| = 1$.
\pitem{collected-height} It holds that $\height_{\Tnew{\collectedsym}}(t) \le \height_T(\origin(t))$.
\end{itemize}
\end{lemma}
\begin{proof}
The height property \pref{collected-height} is obvious from the construction.
The property \pref{collected-unexplored} for unexplored nodes follows from the construction and the fact that each blockage is associated with at most one collected component.
The property that $|\origininv(t)| \le \ell+1$ for explored nodes $t \in F \setminus \Tpref$ follows from the fact that each bag in $F \setminus \Tpref$ has size at most $\ell+1$, and therefore can intersect at most $\ell+1$ different components $\collectedsym$.

It remains to prove that $|\bagnew{\collectedsym}(t)| < |\bag(\origin(t))|$ when $\origin(t) \in F$, for which we need to use the $d_\Tc$-minimality of $X$.
Note that since $\origin(t) \in F$, we necessarily have that \collectedsym is unblocked.
By the definition of $\bagnew{\collectedsym}(t)$, it suffices to prove that $|\pullNei(t,\collectedsym)| < |\bag(\origin(t)) \setminus N[\collectedsym]|$.

First, consider the case when $\pullNei(t,\collectedsym) = \emptyset$, in which case we need to prove that $\bag(\origin(t))$ is not a subset of $N[\collectedsym]$.
In this case, we observe that if $\bag(\origin(t))$ were be a subset of $N[\collectedsym]$, then by \Cref{lem:collected-types-description}, it would be a subset of $N[C]$ for $C \in \cc{G - X}$, in which case it would either be a component blockage if $C$ intersects $\bag(\origin(t))$, or a subset of $N(C)$ and therefore a clique blockage if $C$ does not intersect $\bag(\origin(t))$.

It remains to prove that if $\pullNei(t,\collectedsym)$ is non-empty, then $|\pullNei(t,\collectedsym)| < |\bag(\origin(t)) \setminus N[\collectedsym]|$.
For this proof, recall the definition of the function $d_\Tc$ in \cref{sec:preliminaries} as the depth function of $\Tc$.
For the sake of contradiction, assume that $\pullNei(t,\collectedsym)$ is non-empty and $|\pullNei(t,\collectedsym)| \ge |\bag(\origin(t)) \setminus N[\collectedsym]|$.
Let $C$ be a component of $G-X$ such that $N(\collectedsym) = N(C)$ (its existence follows from \Cref{lem:collected-types-description}).
We claim that now, \[S \coloneqq (N(C) \setminus \pullNei(t,\collectedsym)) \cup (\bag(\origin(t)) \setminus N[\collectedsym])\]
is an $(N(C), \bags(\Tpref))$-separator that contradicts the fact that $N(C)$ is $d_\Tc$-linked into $\bags(\Tpref)$, which by \Cref{lem:closure-linkedness-lemma} contradicts that $X$ is $d_\Tc$-minimal.
First, note that because $\pullNei(t,\collectedsym) \subseteq N(C)$, we indeed have that $|S| \le |N(C)|$.
Moreover, because for each $v \in \pullNei(t,\collectedsym)$ the highest bag containing $v$ is a strict descendant of $\origin(t)$, we have for all $v \in \pullNei(t,\collectedsym)$ and $u \in \bag(\origin(t))$ that $d_\Tc(v) > d_\Tc(u)$, implying that $d_\Tc(S) < d_\Tc(N(C))$.
It remains to prove that $S$ indeed separates $N(C)$ from $\bags(\Tpref)$.
For this, it suffices to prove that it separates $\pullNei(t,\collectedsym)$ from $\bags(\Tpref)$, because $N(C) \setminus S = \pullNei(t,\collectedsym)$.
Suppose $P$ is a shortest path from $\pullNei(t,\collectedsym)$ to $\bags(\Tpref)$ in $G - S$.
Because $\collectedsym$ is disjoint from $\bags(\Tpref)$ and $N(\collectedsym) \setminus S = \pullNei(t,\collectedsym)$, we have that $P$ intersects $N[\collectedsym]$ only on its first vertex.
However, observe that because the nodes of $\Tc$ whose bags contain vertices from $\pullNei(t,\collectedsym)$ are strict descendants of $\origin(t)$, and $\origin(t)$ is a descendant of an appendix of $\Tpref$, it holds that $\bag(\origin(t))$ separates $\pullNei(t,\collectedsym)$ from $\bags(\Tpref)$, and therefore $P$ must intersect $\bag(\origin(t))$.
However, as $\bag(\origin(t))$ is disjoint from $\pullNei(t,\collectedsym)$, the intersection of $P$ and $\bag(\origin(t))$ must be in $\bag(\origin(t)) \setminus N[\collectedsym]$, but $\bag(\origin(t)) \setminus N[\collectedsym] \subseteq S$, so no such path $P$ can exist, implying that $S$ is an $(N(C), \bags(\Tpref))$-separator.
\end{proof}

As the next step, we show that constructing the decompositions $\Tdnew{\collectedsym}$ can be done efficiently.
Note that this construction via a prefix-rebuilding update amounts to giving explicit constructions of $(\Tnew{\collectedsym}, \bag^{\collectedsym})$ for the subtrees consisting of nodes in $\origininv(F)$, and then pointers how the subtrees rooted at blockages should be attached to such nodes.
In particular, slightly abusing notation, let us denote by $\funrestriction{\Tdnew{\collectedsym}}{F}$ the restriction of $\Tdnew{\collectedsym}$ to nodes $t$ with $\origin(t) \in F$.

\begin{lemma}
\label{lem:step3:building_tc_decompositions}
There is an algorithm that given $\Tpref$, $X$, $F$, $H$, and $\blockages{\Tpref}{X}$, computes a list of length $|\collectedcomps{\Tpref}{X}|$ of tuples: for every collected component $\collectedsym \in \collectedcomps{\Tpref}{X}$, the list contains a tuple consisting of $\Cinterface$, $\height(\Tnew{\collectedsym})$, and in addition,
\begin{itemize}
\item if $\collectedsym$ is unblocked, the restriction $\funrestriction{\Tdnew{\collectedsym}}{F}$ of $\Tdnew{\collectedsym}$ to the nodes in $\origininv(F)$ and the mapping from blockages associated with $\collectedsym$ to the leaves of $\funrestriction{\Tdnew{\collectedsym}}{F}$ to which they should be attached, and
\item if $\collectedsym$ is blocked, a pointer to the home bag of $\collectedsym$.
\end{itemize}
\rtitem{comps} The running time is $k^{\Oh{1}} \cdot |F|$.
\end{lemma}
\begin{proof}
First, the blocked components $\collectedsym$ correspond to isolated blockage vertices $t$ of $H$ that have $\component{t} \neq \emptyset$.
Such blockage vertices can be recognized by using the $\cmpsize$ method of $\Daux$.
We can use the $\cmpsize$ method of $\Daux$ to test if $\component{t} \neq \emptyset$, and therefore recognize the blockages associated to blocked components and output them.
Their height $\height(\Tnew{\collectedsym})$ can be obtained from $\Daux$, and for them $\Cinterface = \bag(t) \cap X$.

Then, for constructing the objects for unblocked components, we first do precomputation step that computes for every node $t \in F \setminus \Tpref$ the set $X \cap \component{t}$.
These sets can be computed in total time $k^{\Oh{1}} \cdot |F|$ in a bottom-up manner, because of the facts that (1) $|X \cap \component{t}| \le \Oh{k^4}$ by $c$-smallness of $X$, and (2) by \Cref{lem:no-closure-under-blockage}, for all $t \in \blockages{\Tpref}{X}$, it holds that $X \cap \component{t} = \emptyset$.
We also precompute for every vertex $v \in \bags(F \setminus \Tpref)$ a pointer to some bag in $F \setminus \Tpref$ that contains $v$.

We then iterate over the connected components $C \in \cc{H - X}$ of the exploration graph with $C \cap \bags(F) \neq \emptyset$, that is, corresponding to unblocked collected components.
Let us now fix a component $C \in \cc{H - X}$.
Note that $C$ uniquely identifies an unblocked collected component $\compression^{-1}(C) = \collectedsym$ and note that $\neighopen[H]{C} = \neighopen{\collectedsym}$.
Now, the nodes in $\funrestriction{\Tdnew{\collectedsym}}{F}$ can be identified as the set $\{t \in F \mid \bag(t) \cap C \neq \emptyset\}$, and they can be computed in time $k^{\Oh{1}} \cdot |\funrestriction{\Tdnew{\collectedsym}}{F}|$ by using the previously computed pointers, because they correspond to a connected subtree of $T$.
The set $N(\collectedsym)$ can be computed in time $k^{\Oh{1}} \cdot |C|$ because $|N(\collectedsym)| \le \Oh{k}$.
Then, the sets $\pullNei(t,\collectedsym)$ for all $t \in \funrestriction{\Tdnew{\collectedsym}}{F}$ can be computed in time $k^{\Oh{1}} \cdot |\funrestriction{\Tdnew{\collectedsym}}{F}|$ by using the precomputed sets $X \cap \component{t}$ as we have $\pullNei(t,\collectedsym) = N(\collectedsym) \cap (X \cap \component{\origin(t)}) \setminus \bag(\origin(t))$.
The sets $\bag(\origin(t)) \cap N[\collectedsym]$ can be trivially computed in time $k^{\Oh{1}} \cdot |\funrestriction{\Tdnew{\collectedsym}}{F}|$.
Hence, all bags $\bag^{\collectedsym}(t)$ can also be computed in time $k^{\Oh{1}} \cdot |\funrestriction{\Tdnew{\collectedsym}}{F}|$.
This concludes the construction of $\funrestriction{\Tdnew{\collectedsym}}{F}$.
Then, the pointers from the blockages associated with $\collectedsym$ to the leaves of $\funrestriction{\Tdnew{\collectedsym}}{F}$ can be computed in time $k^{\Oh{1}} \cdot |F|$, because these blockages are exactly $C \cap \blockages{\Tpref}{X}$, and their parents are in $\funrestriction{\Tdnew{\collectedsym}}{F}$.
The height $\height(\Tnew{\collectedsym})$ can be computed in time $\Oh{|\funrestriction{\Tdnew{\collectedsym}}{F}|}$ by using $\Daux$ for computing the heights of the blockages associated with $\collectedsym$ and induction.

Summing up, for each collected component $\collectedsym$ we spend time $k^{\Oh{1}} \cdot (|C| + |V(\funrestriction{\Tnew{\collectedsym}}{F})|)$ constructing $\funrestriction{\Tdnew{\collectedsym}}{F}$ and the pointers.
The sizes of $C$ sum up to at most $|V(H)| \le \Oh{k \cdot  |F|}$, and the sizes of $V(\funrestriction{\Tnew{\collectedsym}}{F})$ sum up to at most $\Oh{k \cdot |F|}$, so the total running time is $k^{\Oh{1}} \cdot |F|$.
\end{proof}

\myparagraph{\stepitem{green} (Append auxiliary subtrees)}
At this point, we have constructed a tree decomposition $\Tdnew{X}$ of $\torso[G]{X}$ and tree decompositions $\Tdnew{\collectedsym}$ of $G[\collectedsym \cup \Cinterface]$ for every collected component $\collectedsym$.
We know that the root of $\Tnew{\collectedsym}$ contains $\Cinterface$ (property \pref{component-root}) and that there is a leaf node $t_\collectedsym \in V(\Tnew{X})$ such that $\bagnew{X}(t_\collectedsym) = \Cinterface$ (property \pref{prefix-subsets}).
Hence, a natural idea is to create the final decomposition $\Tc'$ by connecting the root of each tree $\Tnew{\collectedsym}$ with the corresponding vertex $t_{\collectedsym} \in V(\Tnew{X})$ by an~edge.

Unfortunately, many collected components can have the same interface $B \subseteq X$, and connecting their roots with the same vertex of $\Tnew{X}$ would break the invariant that the data structure maintains a~binary tree decomposition.
It is tempting to work around this problem by building a sufficiently large complete binary tree $\Tnew{\mathrm{bin}}$ rooted at $t_\collectedsym$ and containing $\Cinterface$ in each of its bags.
Then, we could append each tree $\Tnew{\collectedsym}$ to a different leaf of $\Tnew{\mathrm{bin}}$.

However, this approach fails in a~subtle way: the construction of $\Tnew{\mathrm{bin}}$ increases the potential value of the resulting tree decomposition; perhaps quite significantly if a~lot of collected components share the same interface.
Hence, we must ensure that the amortized cost of the construction of $\Tnew{\mathrm{bin}}$ is upper-bounded by the decrease in the potential value resulting from the creation of the collected components.
Roughly speaking, this can only be achieved when the sum of the heights of the vertices of $\Tnew{\mathrm{bin}}$ in the final tree decomposition is at most proportional to the sum of heights of the trees $\Tnew{\collectedsym}$.
Complete binary trees unfortunately do not always satisfy this condition.

However, we will fix this using an idea resembling Huffman codes as follows.
Intuitively, the tree $\Tnew{\mathrm{bin}}$ should be skewed so as to satisfy the following property: given two tree decompositions $\Tnew{\collectedsym_1}, \Tnew{\collectedsym_2}$ of collected components $\collectedsym_1$ and $\collectedsym_2$ with the same interface, if $\Tnew{\collectedsym_1}$ has much larger height than $\Tnew{\collectedsym_2}$, then $\Tnew{\collectedsym_1}$ should be attached to a~leaf of $\Tnew{\mathrm{bin}}$ that is closer to the root of $\Tnew{\mathrm{bin}}$ than the leaf to which $\Tnew{\collectedsym_2}$ is attached.
The following lemma formalizes this intuition.


\newcommand{\Hb}{Q}

\begin{lemma} \label{lem:green-trees}
        Let $h_1, \ldots, h_m$ be positive integers and let $\Hb \coloneqq h_1 + \ldots + h_m$. Then, there exists a rooted binary tree $T$ of height $\Oh{\log \Hb}$ and leaves labeled $h_1, \ldots, h_m$ in some order such that $\sum_{v \in V(T)} \lheight(v) \leq 26\Hb$, where $\lheight$ is defined as follows:
\begin{gather*}
\lheight(v) \coloneqq
\begin{cases}
\leaflabel(v) & \text{if v is a leaf,}\\
1 + \max \{\lheight(c)\,\colon\, \text{$c$ is a child of $v$}\} & \text{otherwise.}
\end{cases}
\end{gather*}
Moreover, such a~tree can be computed in time $\Oh{m + \log \Hb}$.
\end{lemma}
\begin{proof}
        Let $C \colon \Z_{\ge 1} \to \Z_{\ge 0}$ be the function given by the formula $C(a) \coloneqq \ceil{\log_2{a}}$. In other words, $C(a)$ is the smallest nonnegative integer such that $2^{C(a)} \ge a$. In particular, we have $2^{C(a)} < 2a$.

Partition $h_1, \ldots, h_m$ into groups $G_0, \ldots, G_{C(\Hb)}$ as follows: we put $h_i$ into $G_j$ if $h_i \in (\frac{\Hb}{2^{j+1}}, \frac{\Hb}{2^j}]$. This is a~well-defined partition as these intervals for $j=0, \ldots, C(\Hb)$ are disjoint and their union covers the interval $[1, \Hb]$.

Now, for each nonempty group $G_i$ let us create a rooted binary tree $T_i$ as follows. Let the elements of $G_i$ be $g_1, \ldots, g_a$. Then, let $T_i$ be an arbitrary rooted binary tree with $a$ leaves labeled $g_1, \ldots, g_a$, where all leaves are at distance at most $C(a)$ from the root. As $2^{C(a)} \ge a$, such a tree exists. If $G_i$ is empty, we assume that $T_i$ is empty as well.

As the next step create a path $P = v_0, \ldots, v_{C(\Hb)}$ rooted at $v_0$. For each $i=0, \ldots, C(\Hb)$, if $G_i$ is nonempty, assign the root of $T_i$ as a child of $v_i$. Remove a suffix of $P$ that has no subtrees attached to it. This completes the description of desired $T$. One can readily see that the construction can be computed in time $\Oh{m + \log \Hb}$.

What remains is to prove the required properties of $T$. As each $T_i$ has height $
\Oh{\log \Hb}$ and the path $P$ has length $\Oh{\log \Hb}$, it is clear that the height of $T$ is $\Oh{\log \Hb}$ as well.

Next, we will prove a bound on the sum of $\lheight(v)$ for $v$ in a particular $T_i$. Let the elements of $G_i$ be $g_1, \ldots, g_a \in (\frac{\Hb}{2^{i+1}}, \frac{\Hb}{2^i}]$.
Let us group the vertices $v \in V(T_i)$ by their distance $j$ to the farthest leaf in the subtree of $T_i$ rooted at $v$.
For $j = 0$, the group comprises the leaves of $T_i$.
The leaves are labeled by $G_i$, so their values of $\lheight$ do not exceed $\frac{\Hb}{2^i}$.
Next, the vertices at distance $j \geq 1$ from the farthest leaf have the value of $\lheight$ at most $\frac{\Hb}{2^i} + j$ (which follows from the inspection of the definition of $\lheight$) and there are at most $2^{C(a) - j}$ of them (which follows from the fact that each such vertex is at depth at most $C(a) - j$ in $T_i$).
Hence, we have the bound
\[
\begin{split}
  \sum_{v \in V(T_i)} \lheight(v) & \le \sum_{j = 0}^{C(a)} 2^{C(a) - j} \left( \frac{\Hb}{2^i} + j \right) = 2^{C(a)} \left(\frac{\Hb}{2^i} \sum_{j=0}^{C(a)} 2^{-j} + \sum_{j=0}^{C(a)} j \cdot 2^{-j} \right) \le \\
 & \le
2^{C(a)} \left(\frac{\Hb}{2^i} \sum_{j=0}^{\infty} 2^{-j} + \sum_{j=0}^{\infty} j \cdot 2^{-j} \right) = 2^{C(a)} \left(2\cdot\frac{\Hb}{2^i} + 2 \right).
\end{split}
\]

As additionally $2^{C(a)} \cdot 2 \le 4a$ and $2^{C(a)} \cdot 2 \cdot \frac{\Hb}{2^i} \le 8a \cdot \frac{\Hb}{2^{i+1}} < 8 (g_1 + \ldots + g_a)$, we find that
\begin{equation}
  \label{eq:height-sum-over-subtrees}
  \sum_{v \in V(T_i)} \lheight(v) \le 8(g_1 + \ldots + g_a) + 4a \le 12(g_1 + \ldots + g_a).
\end{equation}

Summing \cref{eq:height-sum-over-subtrees} over all $i = 0, 1, \dots, C(\Hb)$, we get that
\[\sum_{i = 0}^{C(\Hb)} \sum_{v \in V(T_i)} \lheight(v) \le 12\Hb. \]

What remains is to bound the sum of $\lheight$ for vertices of $P$.
Define
  \[ l_i \coloneqq \frac{\Hb}{2^i} + 2C(\Hb) - i + 4. \]
We now prove that $\lheight(v_i) \le l_i$ by induction on $i = C(\Hb), C(\Hb) - 1, \dots, 0$.
Each vertex $v_i$ has at most two children:
\begin{itemize}
  \item the root $r_i$ of $T_i$ if $G_i$ is nonempty: since $|G_i| \le 2^{i+1}$, we get that $\lheight(r_i) \le \frac{\Hb}{2^i} + (i + 1) \le l_i - 1$;
  \item $v_{i+1}$ if $v_{i+1}$ exists: we have that $\lheight(v_{i+1}) \le l_{i+1} \le l_i - 1$.
\end{itemize}
Thus indeed $\lheight(v_i) \le l_i$.
Therefore,
%
\[
\sum_{v \in P} \lheight(v) \le \sum_{i=0}^{C(\Hb)} l_i \le \sum_{i=0}^{C(\Hb)} \frac{\Hb}{2^i} + (C(\Hb) + 1)(2C(\Hb) + 4) \le 2\Hb + (C(\Hb) + 1)(2C(\Hb) + 4).
\]
Now, since $C(\Hb) < \log_2 \Hb + 1$, we find that $(C(\Hb) + 1)(2C(\Hb) + 4) < (\log_2 \Hb + 2)(2\log_2 \Hb + 6)$.
Using the standard tools of the real analysis, it can be shown that for all $\Hb \in \Z_{\ge 1}$ we have that
\[ (\log_2 \Hb + 2)(2\log_2 \Hb + 6) \le 12\Hb. \]
Hence,
\[ \sum_{v \in P} \lheight(v) \le 14\Hb. \]

We conclude that
\[\sum_{v \in V(T)} \lheight(v) = \sum_{i=0}^{C(\Hb)} \sum_{v \in V(T_i)} \lheight(v) + \sum_{v \in P} \lheight(v) \leq 12 \Hb + 14\Hb = 26\Hb. \qedhere\]
\end{proof}

Having proved \cref{lem:green-trees}, we can define the tree decompositions that will be attached to the 	already constructed decomposition $\Tdnew{X}$.
Let $\neighall{\Tpref}{X} \coloneqq \{ \Cinterface \mid \collectedsym \in \collectedcomps{\Tpref}{X} \}$ denote the set of all interfaces of collected components.
After obtaining the tuples representing the collected components from \cref{lem:step3:building_tc_decompositions}, we group these tuples by their interfaces, that is, we iterate over all sets $B \in \neighall{\Tpref}{X}$ and list all components $\collectedsym_1, \collectedsym_2, \ldots, \collectedsym_m \in \invinterface{B}$ with interface $B$.
This can be implemented in $|\collectedcomps{\Tpref}{X}| \cdot k^{\Oh{1}}$ time by standard arguments using bucket sorting.

Now, let us fix $B \in \neighall{\Tpref}{X}$, consider components $\collectedsym_1, \collectedsym_2, \ldots, \collectedsym_m \in \invinterface{B}$, and let $h_i = \height(\Tnew{\collectedsym_i})+1$, for $i \in \{1, \ldots, m\}$.
We run the algorithm from \cref{lem:green-trees} for integers $h_1, \ldots, h_m$ to obtain a tree $\Tnewpref{B}$ with leaves labeled with integers $h_i$.
Then, our construction is to attach the trees $\Tnew{\collectedsym_i}$ to $\Tnewpref{B}$, for $i = 1, \ldots, m$, by appending the root of $\Tnew{\collectedsym_i}$ as a child of the leaf of $\Tnewpref{B}$ labeled with $h_i$, so that each leaf of $\Tnewpref{B}$ has exactly one child.
Let $\Tnew{B}$ be the obtained binary tree, and observe that for each node $t_i$ of $\Tnew{B}$ that is a leaf of $\Tnewpref{B}$ to which $\Tnew{\collectedsym_i}$ was attached, $\height_{\Tnew{B}}(t_i) = h_i$.

To define the tree decomposition $\Tdnew{B} = (\Tnew{B}, \bagnew{B})$ it remains to define the function $\bagnew{B}$.
We set $\funrestriction{\bagnew{B}}{V(\Tnew{\collectedsym_i})} = \bag^{\collectedsym_i}$, and $\bagnew{B}(t) = B$ for every $t \in V(\Tnewpref{B})$.

Now, let us analyze this procedure.
First, for each appendix $t$ of $\Tpref$, we show that the number of all possible interfaces of all collected components contained in $\component{t}$ is bounded by a~number depending only on $k$.
Note that this bound is precisely the reason why we require the closure $X$ to be $c$-small.

\begin{lemma}
\label{lem:everywhere-not-many-neighborhoods}
For every appendix $t \in \App(\Tpref)$, we have that
\[ |\{\Cinterface \,\mid\, \collectedsym \in \collectedcomps{\Tpref}{X},\, \collectedsym \subseteq \component{t}\}| \,\leq\, k^{\Oh{k}}. \]
\begin{proof}
Let $\collectedsym \in \collectedcomps{\Tpref}{X}$ be such that $\collectedsym \subseteq \component{t}$.
Recall that $\Cinterface \subseteq X$, implying that $\Cinterface \subseteq X \cap (\component{t} \cup \adhesion{t})$.
Since $X$ is $c$-small, we have $|X \cap \component{t}| \leq c \in \Oh{k^4}$.
Moreover, we have that $|X \cap \adhesion{t}| \leq |\bag(t)| \leq \ell + 1$.
As $|\Cinterface| \leq 2k + 2$ ($\Cinterface$ is a~clique in $\torso[G]{X}$ and $\tw{\torso[G]{X}} \leq 2k + 1$), there are at most $\sum_{i=0}^{2k+2}\binom{c+\ell+1}{i} = (k^4)^{\Oh{k}} = k^{\Oh{k}}$ possible values for $\Cinterface$ such that $\collectedsym \subseteq \component{t}$.
\end{proof}
\end{lemma}

We immediately derive a~bound on the number of all interfaces of all collected components in the graph:

\begin{corollary}
\label{cor:not-many-neighborhoods}
The following inequality holds:
\[
|\neighall{\Tpref}{X}| \leq k^{\Oh{k}} \cdot |\Tpref|.
\]
\end{corollary}
\begin{proof}
For each collected component \collectedsym there is a~unique appendix $t$ such that $\collectedsym \subseteq \component{t}$.
Since $(T, \bag)$ is a binary tree decomposition, we know that $|\App(\Tpref)| \leq |\Tpref| + 1$.
Consequently, $|\neighall{\Tpref}{X}| \leq k^{\Oh{k}} \cdot |\Tpref|$, as desired.
\end{proof}

Let us summarize the properties of the construction from this step of the refinement operation.

\begin{lemma}
Let $B \in \neighall{\Tpref}{X}$.
\begin{itemize}
    \pitem{green-tree} $\Tdnew{B} = (\Tnew{B}, \bagnew{B})$ is a binary tree decomposition of $G[B \cup \bigcup_{\collectedsym \in \invinterface{B}} \collectedsym]$ and has width at most $\ell$.
    \pitem{green-height} The height of $\Tnew{B}$ is at most \[\Oh{\log N} + \max_{\collectedsym \in \invinterface{B}} \height(\Tdnew{\collectedsym}).\]
    \pitem{green-root-bag} The root bag of $\Tdnew{B}$ contains exactly the set $B$.
    \pitem{green-root-sum-height} The sum of the heights of the nodes in $V(\Tnewpref{B})$ is bounded as follows: \[\sum_{t \in V(\Tnewpref{B})} \height_{\Tnew{B}}(t) \le 52 \cdot \sum_{\collectedsym \in \invinterface{B}} \height(\Tnew{\collectedsym}).\]
\end{itemize}
\end{lemma}
\begin{proof}
First, we prove that $\Tdnew{B}$ is indeed a valid tree decomposition.
$\Tdnew{B}$ is a concatenation of valid tree decompositions $\Tdnew{\collectedsym}$ with a prefix $\Tnewpref{B}$ (which is, in fact, a valid tree decomposition of $G[B]$).
Recall that all collected components are vertex-disjoint.
Also, the edge condition is satisfied: for an~edge $uv$ with $u, v \in B$ the edge condition is trivial, and if some endpoint of an~edge $uv$, say $v$, belongs to a~collected component \collectedsym, then $u \in \neighclosed{\collectedsym} \subseteq \collectedsym \cup \Cinterface$.
Since $\Tdnew{\collectedsym}$ is a~tree decomposition of $G[\collectedsym \cup \Cinterface]$ (property \pref{component-tree}) we infer that $u$ and $v$ are together in some bag of $\Tdnew{\collectedsym}$.
Hence, it remains to argue that every vertex $v \in B$ appears in a connected subset of nodes of $\Tdnew{B}$.
This follows immediately from the fact that $\Tnewpref{B}$ contains $B$ in all of the bags, and for every collected component $\collectedsym$, the root bag of $\Tdnew{\collectedsym}$ contains $\Cinterface = B$ as a~subset (property \pref{component-root}).

By \cref{lem:green-trees}, the height of $\Tnewpref{B}$ is at most $\Oh{\log \sum_{\collectedsym \in \invinterface{B}} (\height(\Tnew{\collectedsym})+1)}$, which is at most $\Oh{\log (N^2)} = \Oh{\log N}$ since $\height(\Tnew{\collectedsym}) \le N$ and $|\invinterface{B}| \le N$.
It follows from the construction that the height of $\Tnew{B}$ is at most $\Oh{\log N} + \max_{\collectedsym \in \invinterface{B}} \height(\Tnew{\collectedsym})$.

Since $\Tnewpref{B}$ is a non-empty tree, we have that its root bag (which is the root bag of $\Tdnew{B}$ as well) contains exactly the set $B$.

For bounding the sum of the heights of the nodes in $V(\Tnewpref{B})$, we observe that the height $\height_{\Tnew{B}}(t_i)$ of a node $t_i \in V(\Tnewpref{B})$ that is a leaf in $\Tnewpref{B}$ to which $\Tdnew{\collectedsym_i}$ was attached is exactly $\height(\Tnew{\collectedsym_i})+ 1 = h_i$, and therefore the heights in $V(\Tnewpref{B})$ correspond exactly to the $\lheight$ in the statement of \cref{lem:green-trees}.
Therefore, by \cref{lem:green-trees} we get that
\[\sum_{t \in V(\Tnewpref{B})} \height_{\Tnew{B}}(t) \le 26 \cdot \sum_{\collectedsym \in \invinterface{B}} (\height(\Tnew{\collectedsym})+1) \le 52 \cdot \sum_{\collectedsym \in \invinterface{B}} \height(\Tnew{\collectedsym}).\]
\end{proof}

Finally, we bound the running time of this step.

\begin{lemma}
\label{lem:step4:building_decompositions}
There is an algorithm that given the output of \cref{lem:step3:building_tc_decompositions}, computes the set $\neighall{\Tpref}{X}$ and for each $B \in \neighall{\Tpref}{X}$ the tree $\Tnewpref{B}$ and pointers from the representations of collected components $\collectedsym \in \invinterface{B}$ to the leaves of $\Tnewpref{B}$ to which the tree decompositions $\Tdnew{\collectedsym}$ are attached in the construction.
\rtitem{green} The running time is $|F| \cdot k^{\Oh{1}} + |\Tpref| \cdot k^{\Oh{k}} \cdot \log N$.
\end{lemma}
\begin{proof}
Note that the output of \cref{lem:step3:building_tc_decompositions} has size at most $|F| \cdot k^{\Oh{1}}$.
We first group the interfaces $B = \Cinterface$ by using bucket sorting, which takes $\collectedcomps{\Tpref}{X} \cdot k^{\Oh{1}} = k^{\Oh{1}} \cdot |F|$ time.
Then, the remaining running time is clearly dominated by the total running time of the calls to \cref{lem:green-trees}.
A single application of this lemma for a set $B \in \neighall{\Tpref}{X}$ takes time
\[
\Oh{|\invinterface{B}| + \log N}
\]
as the sum of heights of all collected components can be bounded $\Oh{N^2}$.

Hence, as the sum of $|\invinterface{B}|$ over all $B$ can be bounded by $|F| k^{\Oh{1}}$ simply by the size of the output of \cref{lem:step3:building_tc_decompositions}, and by $|\Tpref| \cdot k^{\Oh{k}}$ by \cref{cor:not-many-neighborhoods}, the total running time can be bounded by $|F| \cdot k^{\Oh{1}} + |\Tpref| \cdot k^{\Oh{k}} \cdot \log N$.
\end{proof}

\myparagraph{\stepitem{join} (Join the pieces of the decomposition together)}
It is time to define the final decomposition $\Tc' = (T', \bag')$.
In the previous step, we have defined, for every $B \in \neighall{\Tpref}{X}$, a binary tree decomposition $\Tdnew{B}$ of all collected components $\collectedsym$ such that $\Cinterface = B$.
Moreover, the root bag of $\Tdnew{B}$ contains exactly the set $B$ (property \pref{green-root-bag}) and in the tree decomposition $\Tdnew{X}$ constructed in \stepref{prefix} there exists a leaf bag $t_B \in V(\Tnew{X})$ such that $\bagnew{X}(t_B) = B$ (property \pref{prefix-subsets}).

Hence, to construct $T'$ it is enough to connect via an edge, for every $B \in \neighall{\Tpref}{X}$, the root of $\Tnew{B}$ with the corresponding vertex $t_B \in V(\Tnew{X})$.
The function $\bag'$ is just the union of the functions $\bagnew{X}$ and $\bagnew{B}$, for $B \in \neighall{\Tpref}{X}$.
Clearly, given $\Tdnew{X}$ and the outputs of \cref{lem:step3:building_tc_decompositions,lem:step4:building_decompositions}, we can in time $|F| \cdot 2^{\Oh{k}}$ construct a description of a prefix-rebuilding update of size bounded by $|F| \cdot 2^{\Oh{k}}$ that turns $\Tc$ into $\Tc'$.

\begin{itemize}
    \rtitem{join} The running time of \stepref{join} is $|F| \cdot 2^{\Oh{k}}$.
\end{itemize}

At this point, we should prove the correctness of the given procedure.
The analysis of the amortized running time together with the exact formula for the potential function $\Phi$ will be given in the next section.

\begin{lemma}
$(T', \bag')$ is a valid tree decomposition of $G$ of width at most $\ell$.
\end{lemma}

\begin{proof}
First, recall that $(T', \bag')$ was obtained by gluing the tree decompositions $\Tdnew{X}$ and $\Tdnew{B}$, for every $B \in \neighall{\Tpref}{X}$.
All of these decompositions are of width at most $\ell$ (properties \pref{prefix-tree} and \pref{green-tree}), and thus $(T', \bag')$ is of width at most $\ell$ as well.

It remains to prove that $(T', \bag')$ is indeed a valid tree decomposition of $G$.
According to the definition of tree decomposition, we need to verify two facts.

\begin{claim}
\label{clm:valid-new-td-vertex}
For each vertex $v \in V(G)$, the subset of nodes $\{t \in V(T')\,\colon\,v \in \bag'(t)\}$ induces a~nonempty connected subtree of $T'$.
\end{claim}

\begin{claimproof}
Consider a vertex $v \in V(G)$.
\begin{itemize}
    \item If $v \not\in X$, then there is a~unique collected component $\collectedsym$ such that $v \in \collectedsym$.
    Since $\Tdnew{\collectedsym}$ is a valid tree decomposition of $G[\collectedsym \cup \Cinterface]$, the subset of nodes of $\Tnew{\collectedsym}$ containing $v$ must induce a connected subtree of $\Tnew{\collectedsym}$.
    Moreover, from the construction of $\Tc'$, we conclude that $v$ cannot appear in any bag of $\Tc'$ outside of $\Tnew{\collectedsym}$.
    \item Now, assume that $v \in X$.
    The vertex $v$ appears in a connected subset of nodes of $\Tdnew{X}$ and in connected subtrees of nodes in trees $\Tdnew{B}$ for each $B \in \neighall{\Tpref}{X}$ satisfying $x \in B$.
    Moreover, by the construction of ${\cal T'}$, the trees $\Tnew{B}$ are attached to the tree $\Tnew{X}$ by connecting two bags containing precisely the set $B$.
     Hence, $v$ must appear in a connected subset of nodes of~${\cal T'}$.\hfill\qedhere
\end{itemize}
\end{claimproof}

\begin{claim}
\label{clm:valid-new-td-edge}
For each edge $uv \in E(G)$, there exists a~node $t \in V(T')$ such that $\{u, v\} \subseteq \bag'(t)$.
\end{claim}

\begin{claimproof}
Consider an edge $uv \in E(G)$.
\begin{itemize}
    \item If $\{u, v\} \subseteq X$, then $uv \subseteq E(\torso[G]{X})$, and thus there must be a bag $t$ in the prefix $\Tnew{X}$ containing both $u$ and $v$, as $\Tc^X$ is a~tree decomposition of $\torso[G]{X}$ (property \pref{prefix-tree}).
    \item If $\{u, v\} \cap X = \emptyset$, there is a~unique component $\collectedsym \in \collectedcomps{\Tpref}{X}$ containing both $u$ and $v$ (since $uv \in E(G)$ and there are no edges between different collected components).
    Let $B = \Cinterface$.
    Then $\Tc$ contains the tree decomposition $\Tc^B$ as a~subtree.
    Since $\Tc^B$ is a~tree decomposition of $G[B \cup \invinterface{B}]$ (property \pref{green-tree}), $u$ and $v$ are together in some bag of $\Tc^B$.
    \item Finally, if exactly one of the vertices $u$ and $v$ belongs to $X$, say $u \in X$, then there is a~unique component $\collectedsym \in \collectedcomps{\Tpref}{X}$ containing $v$.
    Again let $B = \Cinterface$.
    By the definition of an~interface, we have that $u \in B$.
    Again, since $\Tc^B$ is a~tree decomposition of $G[B \cup \invinterface{B}]$ (property \pref{green-tree}), $u$ and $v$ are together in some bag of $\Tc^B$. \hfill\qedhere
\end{itemize}
\end{claimproof}
\cref{clm:valid-new-td-vertex,clm:valid-new-td-edge} conclude the proof of the lemma.
\end{proof}

\subsection{Analysis of the amortized running time}
\label{sec:refine-analysis}




In this section, we provide the proofs of properties \rtref{running-time} and \amref{potential-difference} of the defined refinement operation.
We preserve all notation from \cref{sec:refine-operation}.
Furthermore, denote by $\Fprim \coloneqq F \setminus \Tpref$ the set of all explored nodes excluding the nodes of $\Tpref$.

By summing the running times \rtref{aux} -- \rtref{join} of consecutive steps of the refinement operation, we immediately obtain the following statement.

\begin{fact}
\label{fact:refine-running-time}
The worst-case time complexity of $\mathsf{refine}(\Tpref)$ is upper-bounded by
\[
2^{\Oh{k^9}} (|F| + |\Tpref| \cdot \log N) \leq 2^{\Oh{k^9}}(|\Fprim| + |\Tpref| \cdot \log N).
\]
\end{fact}

Observe that both $|\Tpref|$ and $\log N$ are expressions we can easily control during the run of the data structure.
Unfortunately, the size $|\Fprim|$ of the explored region can be arbitrarily large; in the worst-case the exploration $F$ may contain all the nodes of $T$.
Hence, to prove the desired bounds on the running time of our algorithm, we need to bind the drop of the potential function to the size of exploration $\Fprim$.

Before we proceed with the analysis of the amortized running time, we need to prove one more auxiliary lemma.




\begin{lemma}
  \label{lem:collected-mapping}
  There exists a~mapping \[\collectedmap\,\colon\, \collectedcomps{\Tpref}{X}\, \to\, \Fprim \cup \App(\Tpref) \] satisfying the following properties:
  \begin{enumerate}[label=(\roman*)]
    \item \label{item:collected-mapping-height} for each $\collectedsym \in \collectedcomps{\Tpref}{X}$, we have that $\height(\collectedsym) \leq \height_T(\collectedmap(\collectedsym))$;
    \item \label{item:collected-mapping-weight} for each $\collectedsym \in \collectedcomps{\Tpref}{X}$, if $\collectedmap(\collectedsym) \notin \App(\Tpref)$, then $\weight(\collectedsym) < |\bag(\collectedmap(\collectedsym))|$;
    \item \label{item:collected-mapping-revmap} for each $t \in \Fprim \cup \App(\Tpref)$, we have that $|\collectedmap^{-1}(t)| \leq \ell + 4$.
  \end{enumerate}
\end{lemma}
  \begin{proof}
    We define $\collectedmap$ as follows.
    Take $\collectedsym \in \collectedcomps{\Tpref}{X}$.
    Let $t$ be the home bag of \collectedsym.
    Then:
    \begin{itemize}
      \item If \collectedsym is unblocked, then set $\collectedmap(t) \coloneqq t$.
      Recall here that $t$ is the shallowest node of $T$ such that $\bag(t)$ contains an~explored vertex of \collectedsym.
      \item If \collectedsym is blocked, then set $\collectedmap(t) \coloneqq t$ if $t \in \App(\Tpref)$; otherwise set $\collectedmap(t) \coloneqq \parent{t}$.
      Recall here that $t$ is the blockage that is the only element of $\compression(\collectedsym)$.
    \end{itemize}
    
    Recall that $\height(\collectedsym) = \height_T(t)$ and $\weight(\collectedsym) = |\Cinterface|$.

    We first prove that for each $\collectedsym \in \collectedcomps{\Tpref}{X}$ we indeed have that $\collectedmap(\collectedsym) \in \Fprim \cup \App(\Tpref)$.
    First, the case where \collectedsym is unblocked is immediately resolved by \cref{lem:unblocked-home-bags}: we always have that $\collectedmap(\collectedsym) \in \Fprim$.
    On the other hand, if \collectedsym is blocked, then $t$ is a~blockage and not the root of $T$ as $\Tpref \neq \emptyset$.
    Therefore, $\parent{t}$ exists and belongs to $F$.
    If $\parent{t} \in \Tpref$, then $\collectedmap(\collectedsym) = t \in \App(\Tpref)$.
    Otherwise, $\collectedmap(\collectedsym) = \parent{t} \in \Fprim$.
    
    That property \ref{item:collected-mapping-height} is satisfied is immediate.
    Next we show property \ref{item:collected-mapping-weight}.
    Consider $\collectedsym \in \collectedcomps{\Tpref}{X}$ with $\collectedmap(\collectedsym) \notin \App(\Tpref)$.
    Again, the case of unblocked collected components is resolved by \cref{lem:unblocked-home-bags}.
    Now assume \collectedsym is blocked.
    As $\collectedmap(\collectedsym) \notin \App(\Tpref)$, we find by examining the definition of $\collectedmap$ that $\parent{t} \notin \Tpref$ and $\collectedmap(\collectedsym) = \parent{t}$.	
    But then \cref{lem:blocked-home-bags} applies, yielding that $\weight(\collectedsym) = |\bag(t) \cap X| < |\bag(\parent{t})| = |\bag(\collectedmap(\collectedsym))|$.

    It remains to argue  property \ref{item:collected-mapping-revmap}.
    Consider a~node $t \in \Fprim \cup \App(\Tpref)$.
    Recall that since $t \notin \Tpref$, we have $|\bag(t)| \leq \ell + 1$.
    Note that each unblocked component \collectedsym is assigned by $\collectedmap$ to some explored node $t'$ such that $\bag(t')$ intersects \collectedsym.
    Since collected components form a partitioning of $V(G) \setminus X$, at most $\ell + 1$ different collected components may intersect $\bag(t)$ and thus at most $\ell + 1$ different unblocked collected components may be assigned $t$ by $\collectedmap$.
    Next, each blocked collected component \collectedsym is mapped by $\compression$ to a~different blockage $b \in \blockages{\Tpref}{X}$.
    Observe that, by the definition of $\collectedmap$, \collectedsym may be mapped by $\collectedmap$ to $t$ only if $b$ is equal to either $t$ or a~child of $t$.
    Since $T$ is binary, we conclude that $\collectedmap$ may map at most three different blocked components to~$t$.
  \end{proof}

Now, we prove a slight strengthening of property \amref{potential-difference}.
Note that since $\ell$ is fixed, for the rest of the section we omit it in the lower index of the potential function; that is, we use $\Phi \coloneqq \Phi_\ell$.

\begin{lemma}
\label{lem:phi-difference}
    The following inequality holds:
    \begin{align*}
    \Phi(\Tc) - \Phi(\Tc') \geq |\Fprim| + \sum_{t \in \Tpref} \height_T(t) & - k^{\Oh{k}} \cdot \left( |\Tpref| + \sum_{t \in \App(\Tpref)} \height_T(t) \right) \cdot \log N.
    \end{align*}
\end{lemma}

\begin{proof}
    Within the proof, for clarity we use $\Phi_{\Tc}(\cdot)$ to denote the potential function on nodes of $\Tc$ computed in $\Tc$, and $\Phi_{\Tc'}(\cdot)$ to denote the potential function on nodes of $\Tc'$ computed in $\Tc'$. Also, for a subtree $S$ of $T$, denote $\Phi_{\Tc}(S)\coloneqq\Phi_\Tc(V(S))=\sum_{t\in V(S)}\Phi_\Tc(t)$, and similarly for $\Phi_{\Tc'}(\cdot)$.

    For the tree decomposition $\Tc$, we consider the partitioning of its nodes given by: the prefix $\Tpref$, the set $\Fprim$ of all explored vertices excluding $\Tpref$, and the set of unexplored vertices.
    Hence, we can write $\Phi(\Tc)$ in the form:
    \[
    \label{eq:initial-phi-t}
    \Phi(\Tc) = \Phi_\Tc(\Tpref) + \Phi_\Tc(\Fprim) + \sum_{t \not\in F} \Phi_\Tc(t).
    \]
    Applying the inequality $\Phi_\Tc(t) \geq \height_T(t)$ for every $t \in \Tpref$, we get that
    \begin{equation}
    \label{eq:phi-t}
    \Phi(\Tc) \geq \sum_{t \in \Tpref} \height_T(t) + \Phi_\Tc(\Fprim) + \sum_{t \not\in F} \Phi_\Tc(t).
    \end{equation}
    
    Now, let us focus on the structure of the tree decomposition $\Tc'$.
    From \stepref{join} we know that $V(T')$ is the disjoint union of $V(\Tnew{X})$ and $V(\Tnew{B})$ for every $B \in \neighall{\Tpref}{X}$.
    Hence,
    \[
    \Phi(\Tc') = \Phi_{\Tc'}(\Tnew{X}) + \sum_{B \in \neighall{\Tpref}{X}} \Phi_{\Tc'} (\Tnew{B}).
    \]
    
    From \stepref{green}, we know that every tree $\Tnew{B}$ has a prefix $\Tnewpref{B}$ with attached decompositions $\Tdnew{\collectedsym}$ of collected components in such a way that for every component $\collectedsym \in \collectedcomps{\Tpref}{X}$ there is a unique tree $\Tnewpref{B}$ to which $\Tnew{\collectedsym}$ is attached. Therefore,
    \begin{equation}
    \label{eq:phi-tprim-form}
    \Phi(\Tc') = \Phi_{\Tc'}(\Tnew{X})
    + \sum_{B \in \neighall{\Tpref}{X}} \Phi_{\Tc'} (\Tnewpref{B})
    + \sum_{\collectedsym \in \collectedcomps{\Tpref}{X}} \Phi_{\Tc'}(\Tnew{\collectedsym}).
    \end{equation}
    
    Now, we are going to bound the three terms on the right-hand side of \cref{eq:phi-tprim-form} separately.
    
    \begin{claim}
    \label{clm:phi-bound1}
    The following inequality holds:
    \[
    \Phi_{\Tc'}(\Tnew{X}) \leq k^{\Oh{k}} \cdot \left( |\Tpref| + \sum_{t \in \App(\Tpref)} \height_T(t) \right) \cdot \log N.
    \]
    \end{claim}
 
\newcommand{\ancestor}{\preceq}
\newcommand{\maxheightmap}{\Lambda}
    \begin{claimproof}
    For brevity, we introduce the following notation: for two nodes $t, t' \in V(T')$, we write $t \ancestor t'$ if $t$ is an~ancestor of $t'$ in $T'$.
    
    Recall that $\Tc'$ is created by attaching the subtree $\Tnew{B}$ for each $B \in \neighall{\Tpref}{X}$ to a~different leaf $t_B$ of $\Tnew{X}$.
    Therefore, for each $t \in V(\Tnew{X})$, we have that
    \begin{equation}
    \label{eq:x-easy-height-bound}
    \height_{T'}(t) \leq \height(\Tnew{X}) + \max\{\height(\Tnew{B}) \,\colon\, B \in \neighall{\Tpref}{X}\text{ and } t \ancestor t_B\}.
    \end{equation}
    (We assume that the maximum value of the empty set is $0$.)
    Recall that $\height(\Tnew{X}) \leq \Oh{\log N + k}$ (property \pref{prefix-tree}) and that for each $B \in \neighall{\Tpref}{X}$, we have
    \[\height(\Tnew{B}) \leq \Oh{\log N} + \max_{\collectedsym \in \invinterface{B}} \height(\Tnew{\collectedsym})\]
    (property \pref{green-height}).
    Now, for each $B \in \neighall{\Tpref}{X}$, let us fix $\collectedsym_B$ to be an~arbitrary collected component $\collectedsym \in \invinterface{B}$ maximizing $\height(\Tnew{\collectedsym})$.
    
    Plugging the inequalities above into \cref{eq:x-easy-height-bound}, we get
    \[ \height_{T'}(t) \leq \Oh{\log N + k} + \max\{\height(\Tnew{\collectedsym_B}) \,\colon\, B \in \neighall{\Tpref}{X}\text{ and } t \ancestor t_B\}. \]
    
    Define a~function $\delta \,\colon\,V(\Tnew{X}) \to \N$ and a~partial function $\maxheightmap \,\colon\, V(\Tnew{X}) \rightharpoonup \collectedcomps{\Tpref}{X}$ as follows: for $t \in V(\Tnew{X})$, let $\maxheightmap(t)$ be a collected component $\collectedsym_B$ of maximum height such that $t \ancestor t_B$ (if any such component exists).
    Following that, set $\delta(t) \coloneqq \height(\maxheightmap(t))$; or set $\delta(t) \coloneqq 0$ if $\maxheightmap(t)$ is undefined.
	Note that we have that $\height_{T'}(t) \leq \Oh{\log N + k} + \delta(t)$; therefore,
	\begin{equation}
	\label{eq:phi-x-to-delta}
	\begin{split}
	\Phi_{\Tc'}(\Tnew{X}) & = \sum_{t \in V(\Tnew{X})} g(|\bag'(t)|) \cdot \height_{T'}(t)
	  \leq \sum_{t \in V(\Tnew{X})} g(\ell + 1) \cdot (\Oh{\log N + k} + \delta(t)) \\
	  &\leq \sum_{t \in V(\Tnew{X})} k^{\Oh{k}} \cdot (\log N + \delta(t)) \\
	  &\stackrel{(\star)}{\leq} k^{\Oh{k}} \cdot |\Tpref| \cdot \log N \, + \, k^{\Oh{k}} \cdot \sum_{t \in V(\Tnew{X})} \delta(t).
	\end{split}
	\end{equation}
	(In $(\star)$, we used the fact that $|V(\Tnew{X})| \leq 2^{\Oh{k}} \cdot |\Tpref|$ by property \pref{prefix-size}.)
	Thus, it remains to bound $\sum_{t \in V(\Tnew{X})} \delta(t)$.
	
    Let $B \in \neighall{\Tpref}{X}$ be an interface.
	Since $\Tnew{X}$ is a~tree of height $\Oh{\log N + k}$, there are at most $\Oh{\log N + k}$ ancestors of $t_B$ in $\Tnew{X}$.
	Hence, there are at most $\Oh{\log N + k}$ nodes $t \in V(\Tnew{X})$ for which $\maxheightmap(t) = \collectedsym_B$.
	Naturally, for each such $t$, we have $\delta(t) = \height(\Tnew{\collectedsym_B})$.
	As such,
	\begin{equation}
	\label{eq:sum-delta-to-interfaces}
	\sum_{t \in V(\Tnew{X})} \delta(t) \leq \sum_{B \in \neighall{\Tpref}{X}} \Oh{\log N + k} \cdot \height(\Tnew{\collectedsym_B}).
	\end{equation}
	
	For each interface $B$, let $a_B \in \App(\Tpref)$ be the unique appendix of $\Tpref$ for which $\collectedsym_B \subseteq \component{a_B}$.
	Now, by \cref{lem:everywhere-not-many-neighborhoods}, for each appendix $t \in \App(\Tpref)$, there are at most $k^{\Oh{k}}$ interfaces $B$ such that $a_B = t$.
	For each such interface $B$, we have that $\height(\Tnew{\collectedsym_B})$ is at most the height of the home bag of $\collectedsym_B$ (property \pref{component-height}) in $T$.
	Since $\collectedsym_B \subseteq \component{t}$, the home bag of $\collectedsym_B$ is a~descendant of $t$.
	Therefore, $\height(\Tnew{\collectedsym_B}) \leq \height_T(t)$.
	Combined with \cref{eq:sum-delta-to-interfaces}, this gives us that
	\[
	\sum_{t \in V(\Tnew{X})} \delta(t) \leq \Oh{\log N + k} \cdot \sum_{t \in \App(\Tpref)} k^{\Oh{k}} \cdot \height_T(t) \leq k^{\Oh{k}} \cdot \sum_{t \in \App(\Tpref)} \height_T(t) \cdot \log N.
	\]
	Together with \cref{eq:phi-x-to-delta}, this concludes the proof.
	\end{claimproof}

    \tk{(optional todo): The following formula is slightly too long :(}
    \begin{claim}
    \label{clm:phi-bound2}
    The following inequality holds:
    \[
    \sum_{B \in \neighall{\Tpref}{X}} \Phi_{\Tc'} (\Tnewpref{B}) \leq 52(\ell + 4) \cdot \left( \sum_{t \in \Fprim} g(|\bag(t)| - 1) \cdot \height_T(t) \right) + \sum_{t \in \App(\Tpref)} k^{\Oh{k}} \cdot \height_T(t).
    \]
    \end{claim}
    
    \begin{claimproof} %
    Fix a set $B \in \neighall{\Tpref}{X}$.
    Recall that
    \[\sum_{t \in V(\Tnewpref{B})} \height_{\Tnew{B}}(t) \le 52 \cdot \sum_{\collectedsym \in \invinterface{B}} \height(\Tnew{\collectedsym})\]
    (property \pref{green-root-sum-height}).
    Thus,
    \begin{align*}
    \Phi_{\Tc'} (\Tnewpref{B})
    & = \sum_{t \in \Tnewpref{B}} g(|B|) \cdot \height_{T'}(t) = \sum_{t \in \Tnewpref{B}} g(|B|) \cdot \height_{\Tnew{B}}(t) \\
    & \leq g(|B|) \cdot 52 \cdot \sum_{\substack{\collectedsym \in \invinterface{B}}} \height(\Tnew{\collectedsym}) \\
    & \stackrel{(\star)}{\le} 52 \cdot \sum_{\substack{\collectedsym \in \invinterface{B}}} g(\weight(\collectedsym)) \cdot \height(\collectedsym)
    .
    \end{align*}
    Note that $(\star)$ follows from the facts that $\weight(\collectedsym) = |\Cinterface| = |B|$ for $\collectedsym \in \invinterface{B}$ (\cref{def:interface-weight-height}); and $\height_T(\collectedsym)$, or by definition the height of the home bag of \collectedsym, is at least $\height(\Tnew{\collectedsym})$ (property \pref{component-height}).
    Hence, after summing the above over all sets $B \in \neighall{\Tpref}{X}$, we obtain that
    \[
    \sum_{B \in \neighall{\Tpref}{X}} \Phi_{\Tc'} (\Tnewpref{B}) \leq 52 \cdot \sum_{\collectedsym \in \collectedcomps{\Tpref}{X}} g(\weight(\collectedsym)) \cdot \height(\collectedsym).
    \]
    
    Now we use \cref{lem:collected-mapping}.
    Let $\collectedmap\,\colon\, \collectedcomps{\Tpref}{X}\, \to\, \Fprim \cup \App(\Tpref)$ be the mapping from this lemma.
    Let us split the sum above into two terms:
    \begin{align}
    \label{eq:collmap-use-split}
    \sum_{\collectedsym \in \collectedcomps{\Tpref}{X}} g(\weight(\collectedsym)) \cdot \height(\collectedsym)
    & \leq \sum_{\substack{\collectedsym \in \collectedcomps{\Tpref}{X} \\ \collectedmap(\collectedsym) \in \Fprim }} g(\weight(\collectedsym)) \cdot \height(\collectedsym) \nonumber \\
    & + \sum_{\substack{\collectedsym \in \collectedcomps{\Tpref}{X} \\ \collectedmap(\collectedsym) \in \App(\Tpref) }} g(\weight(\collectedsym)) \cdot \height(\collectedsym).
    \end{align}
    (The inequality comes from the fact that it is possible that $\App(\Tpref) \cap \Fprim \neq \emptyset$.)
    To bound the first sum on the right-hand side we use the fact that if $\collectedmap(\collectedsym) \in \Fprim$, then
    \begin{itemize}
        \item $\weight(\collectedsym) < |\bag(\collectedmap(\collectedsym))|$,
        \item $\height(\collectedsym) \leq \height_T(\collectedmap(\collectedsym))$, and
        \item $|\collectedmap^{-1}(t)| \leq \ell + 4$, for each $t \in \Fprim$.
    \end{itemize}
    Hence,
    \begin{align}
    \label{eq:collmap-use-1}
    \sum_{\substack{\collectedsym \in \collectedcomps{\Tpref}{X} \\ \collectedmap(\collectedsym) \in \Fprim }} g(\weight(\collectedsym)) \cdot \height(\collectedsym)
    & \leq \sum_{\substack{\collectedsym \in \collectedcomps{\Tpref}{X} \\ \collectedmap(\collectedsym) \in \Fprim }} g(|\bag(\collectedmap(\collectedsym))| - 1) \cdot \height_T(\collectedmap(\collectedsym)) \nonumber \\
    & \leq \sum_{t \in \Fprim} (\ell + 4) \cdot g(|\bag(t)| - 1) \cdot \height_T(t).
    \end{align}
    
    \tk{(optional todo): Here is some repetition with the argument for the first sum, but whatever...}
    To bound the second sum on the right-hand side of \cref{eq:collmap-use-split}, we use the fact that if $\collectedmap(\collectedsym) \in \App(\Tpref)$, then
    \begin{itemize}
        \item $\weight(\collectedsym) = |\Cinterface| \leq \ell + 1 = 6k + 6$ (this follows from the fact that $\Cinterface$ is a~subset of the home bag of $\Cinterface$ by property \pref{component-root}),
        \item $\height(\collectedsym) \leq \height_T(\collectedmap(\collectedsym))$, and
        \item $|\collectedmap^{-1}(t)| \leq \ell + 4$, for each $t \in \App(\Tpref)$.
    \end{itemize}
    Therefore,
    \begin{align}
    \label{eq:collmap-use-2}
    \sum_{\substack{\collectedsym \in \collectedcomps{\Tpref}{X} \\ \collectedmap(\collectedsym) \in \App(\Tpref) }} g(\weight(\collectedsym)) \cdot \height(\collectedsym)
    & \leq \sum_{\substack{\collectedsym \in \collectedcomps{\Tpref}{X} \\ \collectedmap(\collectedsym) \in \App(\Tpref) }} g(6k + 6) \cdot \height_T(\collectedmap(\collectedsym)) \nonumber \\
    & \leq \sum_{t \in \App(\Tpref)} (\ell + 4) \cdot g(6k + 6) \cdot \height_T(t) \nonumber \\
    & \leq \sum_{t \in \App(\Tpref)} k^{\Oh{k}} \cdot \height_T(t).
    \end{align}
    
    By plugging \cref{eq:collmap-use-1} and \cref{eq:collmap-use-2} into the inequality \cref{eq:collmap-use-split}, we obtain the desired bound.
    \end{claimproof}

    \begin{claim}
    \label{clm:phi-bound3}
    The following inequality holds:
    \[
    \sum_{\collectedsym \in \collectedcomps{\Tpref}{X}} \Phi_{\Tc'}(\Tnew{\collectedsym}) \leq (\ell + 1) \cdot \left( \sum_{t \in \Fprim} g(|\bag(t)| - 1) \cdot \height_T(t) \right) + \sum_{t \not\in F} \Phi_\Tc(t).
    \]
    \end{claim}
    
    \begin{claimproof}
    We have that
    \[
    \sum_{\collectedsym \in \collectedcomps{\Tpref}{X}} \Phi_{\Tc'}(\Tnew{\collectedsym}) = \sum_{\collectedsym \in \collectedcomps{\Tpref}{X}}\ \sum_{t \in V(\Tnew{\collectedsym})} \Phi_{\Tc'}(t)
    \]
    
    For each node $t \in V(\Tnew{\collectedsym})$, consider the original copy $\origin(t)$ of the node in $T$. Note that $\origin(t)$ is either an~explored node and then $\origin(t) \in \Fprim$, or a~descendant of a~blockage and then $\origin(t) \notin F$. Thus:
    \begin{equation}
    \label{eq:sum-phi-collected}
        \sum_{\collectedsym \in \collectedcomps{\Tpref}{X}} \Phi_{\Tc'}(\Tnew{\collectedsym}) = \sum_{\collectedsym \in \collectedcomps{\Tpref}{X}} \sum_{\substack{t \in V(\Tnew{\collectedsym}) \\ \origin(t) \in \Fprim}} \Phi_{\Tc'}(t)
        + \sum_{\collectedsym \in \collectedcomps{\Tpref}{X}} \sum_{\substack{t \in V(\Tnew{\collectedsym}) \\ \origin(t) \not\in F}} \Phi_{\Tc'}(t).
    \end{equation}

    Recall that for every $t \in V(\Tnew{\collectedsym})$, if $\origin(t)$ is unexplored, then $|\bagnew{\collectedsym}(t)| = |\bag(\origin(t))|$ and $|\origininv(\origin(t))| = 1$ (property \pref{collected-unexplored}), and moreover $\height_{\Tnew{\collectedsym}}(t) \le \height_T(\origin(t))$ (property \pref{collected-height}).
    Since $\height_{T'}(t) = \height_{\Tnew{\collectedsym}}(t)$ for every $t \in V(\Tnew{\collectedsym})$, we get that
    \begin{equation}
    \label{eq:sum-phi-notexplored}
    \begin{split}
        \sum_{\collectedsym \in \collectedcomps{\Tpref}{X}} \sum_{\substack{t \in V(\Tnew{\collectedsym}) \\ \origin(t) \not\in F}} \Phi_{\Tc'}(t) &=
        \sum_{\collectedsym \in \collectedcomps{\Tpref}{X}} \sum_{\substack{t \in V(\Tnew{\collectedsym}) \\ \origin(t) \not\in F}} g(|\bag^{\collectedsym}(t)|) \cdot \height_{T'}(t) \\
        &\le \sum_{\collectedsym \in \collectedcomps{\Tpref}{X}} \sum_{\substack{t \in V(\Tnew{\collectedsym}) \\ \origin(t) \not\in F}} g(|\bag(\origin(t))|) \cdot \height_{T}(\origin(t)) \\
        &\le \sum_{t \notin F} g(|\bag(t)|) \cdot \height_T(t) = \sum_{t \not\in F} \Phi_\Tc(t).
    \end{split}
    \end{equation}
    Note that above there might exist unexplored nodes $t \in V(T) \setminus F$ for which $\origin^{-1}(t) = \emptyset$.

    Now, consider the explored nodes.
    By the fact that for every $t \in V(\Tnew{\collectedsym})$ with explored $\origin(t)$, we have $|\bagnew{\collectedsym}(t)| < |\bag(\origin(t))|$ and $|\origininv(\origin(t))| \leq \ell + 1$ (property \pref{collected-explored}), and moreover $\height_{\Tnew{\collectedsym}}(t) \le \height_T(\origin(t))$ (again property \pref{collected-height}), we find that
%
    \begin{align*}
    \sum_{\collectedsym \in \collectedcomps{\Tpref}{X}} \sum_{\substack{t \in V(\Tnew{\collectedsym}) \\ \origin(t) \in \Fprim}} \Phi_{\Tc'}(t) &= \sum_{\collectedsym \in \collectedcomps{\Tpref}{X}} \sum_{\substack{t \in V(\Tnew{\collectedsym}) \\ \origin(t) \in \Fprim}} g(|\bagnew{\collectedsym}(t)|)\cdot \height_{T'}(t) \\
    & \leq \sum_{\collectedsym \in \collectedcomps{\Tpref}{X}} \sum_{\substack{t \in V(\Tnew{\collectedsym}) \\ \origin(t) \in \Fprim}} g(|\bag(\origin(t))| - 1)\cdot \height_{T}(\origin(t)) \\
    & \leq \sum_{t \in \Fprim} (\ell + 1) \cdot g(|\bag(t) - 1|)\cdot \height_{T}(t).
    \end{align*}
    By plugging the bound above together with \cref{eq:sum-phi-notexplored} into \cref{eq:sum-phi-collected}, we obtain the desired inequality.
    \end{claimproof}
        We then combine the above to obtain a bound for the potential of $\Tc'$.
    \begin{claim}
    \label{clm:phi-bound4}
    The following inequality holds:
    \[
    \Phi(\Tc')
    \leq \Phi_\Tc(\Fprim) - |\Fprim| + \sum_{t \not\in F} \Phi_\Tc(t)
    + k^{\Oh{k}} \cdot \left( |\Tpref| + \sum_{t \in \App(\Tpref)} \height_T(t) \right) \cdot \log N
    \]
    \end{claim}
    
    \begin{claimproof} Recall that (\cref{eq:phi-tprim-form}):
    \[
    \Phi(\Tc') = \Phi_{\Tc'}(\Tnew{X})
    + \sum_{B \in \neighall{\Tpref}{X}} \Phi_{\Tc'} (\Tnewpref{B})
    + \sum_{\collectedsym \in \collectedcomps{\Tpref}{X}} \Phi_{\Tc'}(\Tnew{\collectedsym}).
    \]
    
    By applying inequalities from \cref{clm:phi-bound1}, \cref{clm:phi-bound2} and \cref{clm:phi-bound3}, we obtain that:
    \begin{align*}
    \Phi(\Tc')
    & \leq k^{\Oh{k}} \cdot \left( |\Tpref| + \sum_{t \in \App(\Tpref)} \height_T(t) \right) \cdot \log N \\
    & + 52(\ell + 4) \cdot \left( \sum_{t \in \Fprim} g(|\bag(t)| - 1) \cdot \height_T(t) \right) + \sum_{t \in \App(\Tpref)} k^{\Oh{k}} \cdot \height_T(t) \\
    & + (\ell + 1) \cdot \left( \sum_{t \in \Fprim} g(|\bag(t)| - 1) \cdot \height_T(t) \right) + \sum_{t \not\in F} \Phi_\Tc(t).
    \end{align*}
    
    From the inspection of the definition of $g(\cdot)$ in \cref{sec:refine-potential} it follows that \[g(x) \geq (53\ell + 209) \cdot g(x - 1) + 1\qquad \textrm{for every }x \in \Z_{\ge 1}.\]
    Using that we can bound two of the terms from the inequality above:
   \begin{align*}
     & 52(\ell + 4) \cdot \left( \sum_{t \in \Fprim} g(|\bag(t)| - 1) \cdot \height_T(t) \right) + (\ell + 1) \cdot \left( \sum_{t \in \Fprim} g(|\bag(t)| - 1) \cdot \height_T(t) \right) \\
    =\quad  & \sum_{t \in \Fprim} (53\ell + 209) \cdot g(|\bag(t)| - 1) \cdot \height_T(t) \\
    \leq\quad & \sum_{t \in \Fprim} \left( g(|\bag(t)|) \cdot \height_T(t) - 1 \right) = \Phi_\Tc(\Fprim) - |\Fprim|.
    \end{align*}
    Therefore,
    \[
    \Phi(\Tc')
    \leq \Phi_\Tc(\Fprim) - |\Fprim| + \sum_{t \not\in F} \Phi_\Tc(t)
    + k^{\Oh{k}} \cdot \left( |\Tpref| + \sum_{t \in \App(\Tpref)} \height_T(t) \right) \cdot \log N
    ,\]
    as desired.
    \end{claimproof}
    
%
%
%
    
    We are ready to prove \cref{lem:phi-difference}.
    Recall that from \cref{eq:phi-t}:
    \begin{equation*}
    \Phi(\Tc) \geq \sum_{t \in \Tpref} \height_T(t) + \Phi_\Tc(\Fprim) + \sum_{t \not\in F} \Phi_\Tc(t).
    \end{equation*}
    This inequality, combined with the bound from \cref{clm:phi-bound4}, gives us:
    \[
    \Phi(\Tc) - \Phi(\Tc') \geq \sum_{t \in \Tpref} \height_T(t) + |\Fprim| - k^{\Oh{k}} \cdot \left( |\Tpref| + \sum_{t \in \App(\Tpref)} \height_T(t) \right) \cdot \log N;
    \]
    this ends the proof of the lemma.
\end{proof}

Clearly, \cref{lem:phi-difference} implies property \amref{potential-difference} of the refinement operation.
Furthermore, we can derive from it a~bound on the size of~$\Fprim$:
\begin{align*}
    |\Fprim| \leq \Phi(\Tc) - \Phi(\Tc')  - \sum_{t \in \Tpref} \height_T(t) & + k^{\Oh{k}} \cdot \left( |\Tpref|
    + \sum_{t \in \App(\Tpref)} \height_T(t) \right) \cdot \log N.
\end{align*}

Recall that, by \cref{fact:refine-running-time}, the running time of the refinement call is upper-bounded by:
\[
2^{\Oh{k^9}}(|\Fprim| + |\Tpref| \cdot \log N)
\]
By using the obtained bound on $|\Fprim|$, we obtain the desired bound on the running time as given by property \rtref{running-time}.

\section{Height improvement} \label{sec:height}
\newcommand{\heightop}{\ensuremath{\mathsf{improveHeight}}\xspace}
\newcommand{\getunbalancedop}{\ensuremath{\mathsf{getUnbalanced}}\xspace}
\newcommand{\ubapp}{\ensuremath{\mathsf{UBApp}}}

In this section we leverage the refinement operation defined in \cref{sec:refinement} to produce a~data structure that allows us to maintain a~tree decomposition of small height.
As in \cref{sec:refinement}, we assume $\ell = 6k + 5$ and take $\Phi \coloneqq \Phi_\ell$ and $g \coloneqq g_\ell$ as defined in \cref{sec:refine-potential}.
We prove that:

\begin{lemma}[Height improvement data structure]
  \label{lem:height-data-structure}
  Fix $k \in \N$ and let $\ell = 6k + 5$.
  The $(\ell + 1)$-prefix-rebuilding data structure from \cref{lem:refinement-data-structure} maintaining a~tree decomposition $\Tc = (T, \bag)$ can be extended to additionally support the following operation:
  \ms{or should we say below that it returns the (sequence of) prefix rebuilding updates?}
  
  \begin{itemize}
    \item $\heightop()$: updates $\Tc$ through a~sequence of prefix-rebuilding updates, producing a~tree decomposition $\Tc' = (T', \bag')$ such that
    \[ \height(T') \leq 2^{\Oh{k \log k \sqrt{\log n \log \log n}}} \qquad\text{and}\qquad
    |V(T')| \leq k^{\Oh{k}} \cdot n^3. \]
    Also, $\Phi(\Tc') \leq \Phi(\Tc)$ and if the width of $\Tc$ is at most $\ell$, then the width of $\Tc'$ is also at most $\ell$.

    The worst-case running time of \heightop is bounded by
    \[ 2^{\Oh{k^9}} \cdot (\Phi(\Tc) - \Phi(\Tc')) + \Oh{1}. \]
  \end{itemize}
\end{lemma}

\cref{lem:height-data-structure} will be crucial in ensuring the efficiency of the data structure maintaining tree decompositions of graphs dynamically: after each update to the tree decomposition, we will call \heightop so as to ensure that the height of the decomposition stays sufficiently small.
Note here that all prefix-rebuilding updates performed by \heightop in order to decrease the height of the decomposition are essentially ``free'' in terms of amortized running time: the running time of the improvement is fully amortized by the decrease in the potential value, i.e. $\Phi(\Tc)-\Phi(\Tc')$. In particular, this decrease in potential is always nonnegative. 

The rest of this section is dedicated to the proof of \cref{lem:height-data-structure}.


\subsection{Algorithm description}
We begin with the description of \heightop.
Assume we are~given a~tree decomposition $(T, \bag)$, and assume we have access to the $(\ell+1)$-prefix-rebuilding data-structure from \cref{lem:height-maintenance}; in particular, in constant time we can access the height $\height(s)$ of each node $s$ and the size $\size(s)$, i.e., the number of nodes in the subtree of $T$ rooted at $s$.
The tree decomposition will be updated by \heightop only through the refinement operation $\mathsf{refine}(\cdot)$ implemented by \cref{lem:refinement-data-structure}, or through a recomputation of $\Tc$ from scratch.

First, we define a~recursive function \getunbalancedop that, assuming that the tree $T$ is ``unbalanced'', determines an unbalanced prefix $W$ of $T$ that we will then pass to the refinement operation from \cref{lem:refinement-data-structure}. Whether a~tree is balanced or not depends on its height (i.e., whether the tree is deep or shallow) and on its size (i.e., whether the tree is large or small): a~tree is unbalanced if it is both deep and small.
Note here that the definition of the word ``deep'' and ``small'' will depend on the current depth of the recursion.


We now formalize the description of \getunbalancedop.
Let us assume that we are given two sequences of positive integers: $h_1 > h_2 > \ldots > h_a$ and $n_1 > n_2 > \ldots > n_a$, where $h_j \le n_j$ for $j < a$ and $h_a > n_a$. For a node $t$ of $T$ and $j \in \{1,\ldots,a\}$, we say that the subtree $T_t$ of $T$ rooted at $t$ is {\em{$j$-shallow}} if $\height(t) < h_j$, and {\em{$j$-deep}} otherwise. Similarly, we say that the subtree is {\em{$j$-big}} if $\size(t) > n_j$; and {\em{$j$-small}} otherwise (not to be confused with the notion of $c$-small closures).
We define a~recursive function $\getunbalancedop(T_t, j)$, where $T_t$ is the subtree of $T$ rooted at $t \in V(T)$, and $j \in [a - 1]$, as follows.
We first assert that $T_t$ is both $j$-small and $j$-deep.
Let $P$ be any longest path from $t$ to a~leaf of $T_t$.
Define
\[ \ubapp(P, j) \coloneqq \{t' \in \App(P) \,\colon\, T_{t'} \text{ is } (j+1)\text{-deep and }(j+1)\text{-small}\} \]
as the set of unbalanced trees rooted at the appendices of $P$.
On the other hand, we will say that each tree rooted at a node of $\App(P) \setminus \ubapp(P, j)$, i.e., each balanced tree rooted at an~appendix of $P$, is \emph{$P$-found}.
Then, $\getunbalancedop(T_t, j)$ returns
\[ P \,\cup\, \bigcup_{t' \in \ubapp(P, j)} \getunbalancedop(T_{t'}, j+1). \]
Clearly, \getunbalancedop is well-defined: in each recursive call $\getunbalancedop(T_{t'}, j+1)$ it is guaranteed that $T_{t'}$ is both $(j+1)$-deep and $(j+1)$-small; and thus, since $h_a > n_a$ and the height of each tree is upper-bounded by its size, $\getunbalancedop(\cdot, a)$ is never called.
Moreover, it can be easily verified that $\getunbalancedop(T, 1)$ returns a~prefix of $T$.

Having described \getunbalancedop, we explain how to use it in the implementation of procedure \heightop. We assume that the data structure has access to the refinement prefix-rebuilding data structure from \cref{lem:refinement-data-structure}.
First, we ensure that $T$ is of reasonable size: if $|V(T)|$ is at least $B \cdot n^3$, where $B \coloneqq g_\ell(\ell+1) \leq k^{\Oh{k}}$, then we recompute the entire tree decomposition from scratch using the algorithm from \cref{thm:tw-opt} on $\Tc$ and make it binary, while not increasing its size more than twice, thus getting a new tree decomposition $\Tc^\star = (T^\star, \bag^\star)$. We call this process \textit{shrinking}. If shrinking occurs, we discard the old refinement data structure and initialize a fresh one with $\Tc^\star$ (note that this can be represented as a~prefix rebuilding update of size $|V(T)| + |V(T^\star)|$). Note that $\Tc^\star$ is a tree decomposition of width at most $k$ and size at most $\Oh{n}$, hence $\Phi(\Tc^\star) = \Oh{B \cdot n^2}$.
Observe that since $\Phi(\Tc) \geq B \cdot n^3 \gg \Phi(\Tc^\star)$, the decrease in the potential value will cover the computational cost incurred from the shrinking.


After ensuring that $T$ has reasonable size, we check whether $T$ is $1$-deep. If it is not, then we are done. If it is, then we extract an unbalanced prefix $W$ through $\getunbalancedop(T, 1)$, we call $\mathsf{refine}(W)$ and apply to $(T, \bag)$ the prefix-rebuilding update it returned. The parameter $n_1$ will be chosen so that $n_1 \ge B \cdot n^3 \geq |V(T)|$.
Thus, at the moment $\getunbalancedop(T, 1)$ is called, $T$ is $1$-deep and $1$-small, so the initial invariant of \getunbalancedop is satisfied.
We will prove later that for suitably chosen sequences $(n_i)$, $(h_i)$, the potential value will cover the time complexity of the refinement.

Now, after a single pass of this procedure, $T$ might still not be sufficiently shallow.
Therefore we repeat the sequence of operations above in a loop until the tree becomes such; we stress that at each iteration of the loop we also check whether $T$ is too big and potentially shrink it. As in each iteration of the loop, the potential $\Phi(\Tc)$ will drop significantly, this process will be finite; and in what follows, we will be able to reasonably bound its running time.



%
%
%

Before arguing why \getunbalancedop satisfies the bounds in the statement of \cref{lem:height-data-structure} and why it even terminates at all, we first specify the sequences $(h_i)$ and $(n_i)$.

\begin{lemma}
\label{lem:wojtek-wallofmath}
Let $N$ and $c$ be real numbers with $1 < c < N$.
There exist sequences of real numbers $h_1 > h_2 > \ldots > h_a$ and $n_1 > n_2 > \ldots > n_a$ such that 
\begin{itemize}
\item $h_1 \le 2^{\Oh{\sqrt{\log N \log c}}}$ and $n_1 \ge N$, 
\item $h_a > n_a > 1$, and
\item $n_i = n_{i+1} \cdot h_{i+1} \ge c \cdot h_{i+1} = h_i$ for all $1 \le i < a$.
\end{itemize}
\end{lemma}
\begin{proof}
Let $a$ be the smallest integer such that $c^{\frac{a(a+1)}{2}} \ge N$. We set
\[ h_i \coloneqq c^{a+2-i} \qquad \text{and} \qquad n_i \coloneqq c^{\frac{(a - i + 1) \cdot (a - i + 2)}{2}}. \]
We have $h_a = c^2 > c = n_a$ and $n_i = n_{i+1} \cdot h_{i+1} \ge c \cdot h_{i+1} = h_i$ for $1 \le i < a$. Also, note that $n_1 = c^{\frac{a(a+1)}{2}} \ge N$. Therefore, it remains to prove that $h_1 \leq 2^{\Oh{\sqrt{\log N \log c}}}$.

As $c < N$, we have that $a \ge 2$.
By the definition of $a$, we also have $c^{\frac{a(a-1)}{2}} < N \le c^{\frac{a(a+1)}{2}}$, which implies that $\frac{a(a-1)}{2} \cdot \log c < \log N \le \frac{a(a+1)}{2} \cdot \log c$. As $a^2 \ge \frac{a(a+1)}{2}$, we also have that $a^2 \log c \ge \log N$ and thus $a \ge \sqrt{\frac{\log N}{\log c}}$. On the other hand, as $a \ge 2$, we have $a^2 \le 4 \cdot \frac{a(a-1)}{2}$, so
\[ a \log c = \frac{a^2 \log c}{a} \le \frac{4\frac{a(a-1)}{2} \log c}{a}  \le \frac{4 \log N}{a} \le \frac{4 \log N}{\sqrt{\frac{\log N}{\log c}}} = 4 \sqrt{\log N \log c}. \]
Therefore, $h_1 = c^{a+1} \le c^{2a} = 2^{2a \log c} = 2^{\Oh{\sqrt{\log N \log c}}}$.
\end{proof}

\noindent We get the sequences $(h_i)$ and $(n_i)$ by applying \cref{lem:wojtek-wallofmath} with $N = |V(T)|$ and $c = 2^{\Oh{k \log k}} \log n$, to be exactly specified later.

\subsection{Bounding the potential decrease}

Let $\Tc'$ be the resulting tree decomposition after one iteration of the loop from the height improvement procedure, assuming that the shrinking procedure did not happen. We will now bound $\Phi(\Tc) - \Phi(\Tc')$ using the property \amref{potential-difference} from \cref{lem:refinement-data-structure}.

  Let $W$ be the prefix of $T$ returned by $\getunbalancedop(T, 1)$ and let $\Wc$ be the set of subtrees of $T$ rooted at the vertices of $\App(W)$.
  As $W$ is the disjoint union of the vertical paths found by \getunbalancedop, each subtree in $\Wc$ can be assigned the unique vertical path $P$ found by \getunbalancedop such that the subtree is $P$-found.
  
  Let us first focus on a~recursive call $\getunbalancedop(T_r, j)$, where $T_r$ is the subtree of $T$ rooted at $r$ and $j \in [a - 1]$.
  Let $P$ be the longest path rooted at $r$ found in the algorithm.
  Let $T_1, T_2, \dots, T_b$ be the set of $P$-found subtrees.
  As $T_r$ is binary, we have that $b \leq |P| + 1 \leq 2|P|$.
  From the invariant of \getunbalancedop we know that $T_r$ is $j$-deep, that is $|P| = \height(T_r) \ge h_j$. As each $P$-found subtree is either $(j+1)$-shallow or $(j+1)$-big, we have that
  \[ \sum_{i=1}^{b} \height(T_i) \le \sum_{i \text{ : } T_i \text{ is }(j+1)\text{-shallow}} \height(T_i) + \sum_{i \text{ : } T_i \text{ is }(j+1)\text{-big}} \height(T_i). \]
  (Note that a subtree can be both $(j+1)$-shallow and $(j+1)$-big.)
  We will now bound each sum on the right hand side separately.
  \mp{Don't we have $b\leq |P|-1$, because $P$ ends in a leaf?}
  \ms{Yes, but at this point of time I didn't want to change the numbers in the formulas and risk breaking stuff :)}
  
  First, each $(j+1)$-shallow subtree has height smaller than $h_{j + 1}$.
  As there are at most $b \leq 2|P|$ of them among the $P$-found subtrees, we get that
  \[ \sum_{i \text{ : } T_i \text{ is }(j+1)\text{-shallow}} \height(T_i) \le 2|P| \cdot h_{j + 1}. \]

  Second, $P$ is a~longest root-leaf path in $T_r$, so each $T_1, T_2, \dots, T_b$ has height at most $|P|$.
  Moreover, as $T_r$ is $j$-small, we have $|V(T_r)| \le n_j$; and for each $(j+1)$-big tree $T_i$ we have that $|V(T_i)| > n_{j+1}$.
  As $T_1, \dots, T_b$ are pairwise disjoint and are subtrees of $T_r$, we conclude that at most $\frac{n_j}{n_{j+1}}$ of $P$-found trees are $(j+1)$-big.
  Hence,
  \[ \sum_{i \text{ : } T_i \text{ is }(j+1)\text{-big}} \height(T_i) \le 2|P| \cdot \frac{n_{j}}{n_{j + 1}}. \]
  
  We conclude that
  \begin{equation}
  \label{eq:sum-of-appendix-heights}
  \sum_{i=1}^{b} \height(T_i) \le 2|P| \cdot \left(h_{j+1} + \frac{n_j}{n_{j+1}}\right).
  \end{equation}
  
  Also we observe that
  \begin{equation}
  \label{eq:sum-of-path-heights}
  \sum_{t \in P} \height(t) = 1 + 2 + \dots + |P| \geq \frac{|P|^2}{2} \geq \frac{|P| \cdot h_j}{2}.
  \end{equation}

Let us now recall property \amref{potential-difference} that, adjusted to our current notation, states that

$$\Phi(\Tc) - \Phi(\Tc') \geq \sum_{t \in W} \height_T(t) - 2^{\Oh{k \log k}} \cdot \left( |W| + \sum_{t \in \App(W)} \height_T(t) \right) \cdot \log |V(T)|. $$

At that point of computation we have that $|V(T)| \le B \cdot n^3$, where $B = k^{\Oh{k}}$, hence $\log |V(T)| \leq \Oh{\log n + k \log k}$.
Therefore, we can simplify this bound to the following:

\begin{equation}
\label{eq:refinement-potential-rephrased}
\Phi(\Tc) - \Phi(\Tc') \geq \sum_{t \in W} \height_T(t) - 2^{\Oh{k \log k}} \cdot \left( |W| + \sum_{t \in \App(W)} \height_T(t) \right) \cdot \log n.
\end{equation}

Let us now expand the $\cal O(\cdot)$ notation in \cref{eq:refinement-potential-rephrased} and fix a constant $1< C \in 2^{\Oh{k \log k}}$, depending only on $k$, such that the $2^{\Oh{k \log k}}$ term in \cref{eq:refinement-potential-rephrased} is upper-bounded by $C$. With that, we can rewrite the inequality again as:

\begin{equation} \label{eq:pot-diff}
\Phi(\Tc) - \Phi(\Tc') \geq \sum_{t \in W} \height_T(t) - C \cdot \left( |W| + \sum_{t \in \App(W)} \height_T(t) \right) \cdot \log n.
\end{equation}

Let us now inspect the contribution to the right hand side of \cref{eq:pot-diff} coming from the terms introduced by a~single recursive call $\getunbalancedop(T_r, j)$: the rooted path $P$ added to $W$ and the $P$-found subtrees. It is equal to 
\[\sum_{t \in P} \height(t) - C \cdot \left (|P| + \sum_{i=1}^{b}  \height(T_i)\right) \cdot \log n.\]

\tk{Would it be difficult to know explicitly what is ``sufficiently big $n$''}
\ms{I think it's $n \in 2^{\Omega(k \log k)}$. But do we want a~separate case for small $n$?}
 Let us also set now the value $c$ that we used in \cref{lem:wojtek-wallofmath} to define sequences $h_1, \ldots, h_a$ and $n_1, \ldots, n_a$ as $c \coloneqq 20 C \log n$ (let us note that for sufficiently large $n$, it is smaller than $N$, as required); the reason for such choice will become clear in a moment. Having that specified, we can state the following claim: 

\begin{claim}\label{cl:bd-sumh}
The following inequality holds:
	\[C \cdot \left(|P| + \sum_{i=1}^{b}  \height(T_i)\right) \cdot \log n \le \frac12 \sum_{t \in P} \height(t)\]
\end{claim}
\begin{claimproof}
 Let us recall that $h_j = c \cdot h_{j+1}$, $h_j \ge 1$, and $\frac{n_j}{n_{j+1}} = h_{j+1}$. Because of that, we can derive the following bounds:
\[
\begin{split}
C \cdot \left(|P| + \sum_{i=1}^{b}  \height(T_i)\right) \cdot \log n &\stackrel{(\ref{eq:sum-of-appendix-heights})}{\le} C\left(|P| + 2|P| \left(h_{j+1} + \frac{n_j}{n_{j+1}}\right)\right) \cdot \log n \\
&\le C(|P| \cdot h_{j+1} + 2|P|(h_{j + 1} + h_{j+1})) \cdot \log n \\
&= 5C \cdot |P| \cdot h_{j+1} \cdot \log n.
\end{split}
\]

On the other hand, we have that
\[ \sum_{t \in P} \height(t) \stackrel{(\ref{eq:sum-of-path-heights})}{\ge} \frac{|P| \cdot h_j}{2} = \frac{|P| \cdot c \cdot h_{j+1}}{2} = 10C \cdot |P| \cdot h_{j+1} \cdot \log n \] and the claim follows.
\end{claimproof}

Summing the inequality from \cref{cl:bd-sumh} over all paths $P$ corresponding to the calls of \getunbalancedop, we immediately get the following claim:

\begin{claim} \label{cl:bd-sumh2}
The following inequality holds:
	\[C  \cdot \left( |W| + \sum_{t \in \App(W)} \height_T(t) \right) \cdot \log n \le \frac12 \sum_{t \in W} \height(t).\]
\end{claim}

Applying that to \cref{eq:pot-diff} we get that $\Phi(\Tc) - \Phi(\Tc') \ge \frac12 \sum_{t \in W} \height(t)$.

The running time of $\textsf{refine}(W)$, based on property \rtref{running-time}, can be bounded as:
\[ 2^{\Oh{k^9}} \cdot \left(|W| + \sum_{t \in \App(W) }\height(t)\right) \cdot \log n + \Oh{\max(\Phi(\Tc) - \Phi(\Tc'), 1)}. \]
(As before, we used the fact that $\log |V(T)| \leq \Oh{\log n + k \log k}$.)
 However, from \cref{cl:bd-sumh2} we have: $$C \cdot \left( |W| + \sum_{t \in \App(W)} \height_T(t) \right) \cdot \log n \le \frac12 \sum_{t \in W} \height(t) \le \Phi(\Tc) - \Phi(\Tc').$$
 It follows that the running time can be bounded by $2^{\Oh{k^9}}(\Phi(\Tc) - \Phi(\Tc'))$. The total running time of $\getunbalancedop(T, 1)$ can be easily bounded as $\Oh{|W|}$ by implementing it by using \cref{lem:height-maintenance}, which again can, too, be bounded by $\Oh{\Phi(\Tc) - \Phi(\Tc')}$. As such, if the shrinking procedure did not happen, the running time of one pass of the loop from the height improvement operation, can be bounded as $2^{\Oh{k^9}}(\Phi(\Tc) - \Phi(\Tc'))$ too.
 

Let us now bound the running time of the shrinking procedure. If shrinking occurs, we have that $|V(T)| \ge B \cdot n^3$. We know that $\Phi(\Tc) \ge |V(T)|$. Running the algorithm from \cref{thm:tw-opt} on $T$ with parameter $k$ and post-processing the returned tree decomposition to make it binary takes $2^{\Oh{k^3}} \cdot |V(T)|$ time and produces $\Tc^\star = (T^\star, \bag^\star)$ such that $\Phi(\Tc^\star) \leq \Oh{B n^2}$. Then, resetting the refinement data structure with $\Tc^\star$ takes time $2^{\Oh{k^8}} \cdot |V(T)| = 2^{\Oh{k^8}} \cdot \Oh{B n^2}$.

As $\Phi(\Tc) \ge Bn^3$ and $\Phi(\Tc^\star) \le \Oh{Bn^2}$, we have that $\Oh{\Phi(\Tc)} = \Oh{\Phi(\Tc) - \Phi(\Tc')}$ for $n \in \Omega(1)$ large enough. Therefore, the running time of the shrinking procedure can be bounded as $2^{\Oh{k^8}} \cdot |V(T)| \le 2^{\Oh{k^8}} \Phi(\Tc) =  2^{\Oh{k^8}} (\Phi(\Tc) - \Phi(\Tc'))$. 

%
As $\Phi(T)$ strictly decreases after each iteration of the main loop, it has to terminate at some point.
As the running times of both the shrinking procedure and height improvement can be bounded by $2^{\Oh{k^9}}$ times the potential decrease, the same is true for the full single iteration of the main loop of \heightop. Consequently, the worst-case bound on the running time of \heightop from the statement of \cref{lem:height-data-structure} is proved.


We recall that $h_1 = 2^{\Oh{\sqrt{\log N \log c}}}$. We set $c = 2^{\Oh{k \log k}}  \log n$ and $N \le B \cdot n^3$, where $B = 2^{\Oh{k \log k}}$, so $\log c = \Oh{k \log k + \log \log n}$ and $\log N = \Oh{k \log k + \log n}$.
Therefore, $$h_1 =  2^{\Oh{\sqrt{\log N \log c}}} =  2^{\Oh{k \log k \sqrt{\log n \log \log n}}}.$$
Since \heightop only terminates when the height of the tree decomposition drops below $h_1$, the height condition from the statement of \cref{lem:height-data-structure} also holds.
Also, since \heightop only updates the tree decomposition through $\mathsf{refine}$ or replaces the tree decomposition with a~fresh decomposition of width at most $k < \ell$, it is guaranteed that whenever \heightop is invoked with a~tree decomposition $\Tc$ of width at most $\ell$, the resulting decomposition $\Tc'$ also has width at most $\ell$.
Finally, it is straightforward that $V(T') \leq k^{\Oh{k}} \cdot n^3$ and $\Phi(\Tc') \leq \Phi(\Tc)$.
This concludes the proof of \cref{lem:height-data-structure}.



\section{Proof of \cref{lem:weak-treewidth-ds}} \label{sec:wrap-up}

After having done all the necessary preparations, we are ready to prove \cref{lem:weak-treewidth-ds}.
For convenience, we recall its statement.

\torestateWeakTreewidthDs*

Again, as in \cref{sec:refinement}, we fix $\ell = 6k + 5$ and take $\Phi \coloneqq \Phi_\ell$ and $g \coloneqq g_\ell$ as defined in \cref{sec:refine-potential}.


First, we are going to describe how the data structure is implemented; then, we will bound the running time of a series of edge updates.

\subsection{Data structure}
In order to initialize a~dynamically changing annotated tree decomposition $\Tc = (T, \bag, \edges)$, we instantiate the following $(6k+6)$-prefix-rebuilding data structures:
\begin{itemize}
  \item $\D$: the refinement data structure with overhead $2^{\Oh{k^8}}$ from \cref{lem:refinement-data-structure}, additionally supporting \heightop (\cref{lem:height-data-structure});
  \item $\mathbb{H}$: the data structure with overhead $\Oh{1}$ from \cref{lem:height-maintenance}, allowing us to query, for each vertex $v \in V(G)$, the top bag $\Top(v) \in V(T)$ containing $v$; and for each $t \in V(T)$, the height of $t$ in $V(T)$.
\end{itemize}
Both data structures store the same decomposition $(T, \bag, \edges)$: all prefix-rebuilding updates performed by $\D$ are forwarded to $\mathbb{H}$.
After the initialization and after each query, we maintain the following invariant: $(T, \bag, \edges)$ contains an~annotated tree decomposition of $G$ of width at most $\ell = 6k + 5$, height at most $2^{\Oh{k \log k \sqrt{\log n \log \log n}}}$ and size at most $k^{\Oh{k}} \cdot n^3$.
After each query, we return the sequence of prefix-rebuilding updates performed by $\D$.

As for the initialization of our structure for an empty graph, we can initialize $\Tc = (T, \bag, \edges)$ with $T$ being a~complete binary tree (of height $\Oh{\log n}$) with $n$ nodes and each bag containing a different vertex of $V(G)$. Obviously, for each $t \in V(T)$, we have $\edges(t) = \emptyset$.
The initialization of $\D$ and $\mathbb{H}$ with it takes $2^{\Oh{k^8}} \cdot n$ time.
We remark that such a~decomposition $\Tc$ satisfies $\Phi(\Tc) \leq \Oh{kn}$ since the average height of a~node in $T$ is $\Oh{1}$, and $g(1) \in \Oh{k}$.


First assume that we are to add an~edge $uv$ to $G$.
Let $G'$ be equal to the graph $G$ with the edge $uv$ added.
Towards that goal, we must first ensure that the edge condition is satisfied for the new edge, i.e., both $u$ and $v$ belong to the same bag of $\Tc$.
Let $t_u = \Top(u)$ and $t_v = \Top(v)$ be the top bags of $u$ and $v$, respectively, in $T$.
If it is the case that $v \in \bag(t_u)$ or $u \in \bag(t_v)$, then the edge condition is already satisfied; however, we still need to update the function $\edges$ with the newly added edge $uv$.
Let us recall from the definition of $\edges$ that we should include the edge $uv$ in exactly one set: the set $\edges(t)$ for the topmost bag $t$ containing both $u$ and $v$. It is easy to verify that this bag is actually one of $t_u$ or $t_v$, whichever is at the smaller height in $T$.
Without loss of generality, assume it is $t_v$.
We include $uv$ in $\edges(t_v)$ by performing a~simple prefix-rebuilding update on the path from the root of $T$ to $t_v$.
The prefix rebuilt has size at most $\height(T)$; hence, the update takes time $2^{\Oh{k^8}} \cdot \height(T)$ and produces a~tree decomposition $\Tc'$ of $G'$.
The structure of the tree decomposition remains unchanged here, so the invariant is preserved by the update.

%

Let us assume now that $v \not\in \bag(t_u)$ and $u \not\in \bag(t_v)$; in this case, we must expand some bags in the decomposition so as to satisfy the edge condition.
Let $P_u$ and $P_v$ be the paths from the root of $T$ to $t_u$ and $t_v$, respectively.
Then, the update of the decomposition proceeds in two steps: first, we add $v$ to all the bags $\bag(t)$ for $t \in P_u \cup P_v$, and only then we add the edge $uv$ to an~appropriate set $\edges$.
The former is done by obtaining a~weak description $\widehat{u}$ of a~prefix-rebuilding update adding $v$ to the required bags, invoking $\bar{u} \coloneqq \mathsf{strengthen}(\widehat{u})$ (\cref{lem:prds-strengthen}) and then applying the resulting prefix-rebuilding operation $\bar{u}$.
Then, the insertion of the edge $uv$ to the appropriate set $\edges(t)$ is conducted as in the previous case.
Note that all of the above can be performed in time $2^{\Oh{k^8}} \cdot \height(T)$.

Now, $\Tc' = (T, \bag', \edges')$ is a~correct annotated tree decomposition of $G'$, but adding $u$ to the bags in $P_u \cup P_v$ might have increased the width of the decomposition to $\ell + 1 = 6k + 6$, thus invalidating the width invariant.
In order to counteract this, we call $\mathsf{refine}(P_u \cup P_v)$.
Note that $P_u \cup P_v$ covers all bags of size $\ell + 2 = 6k + 7$, so the call satisfies the precondition of \cref{lem:refinement-data-structure} and the annotated tree decomposition $\Tc''$ of $G'$ produced by $\mathsf{refine}$ has width at most $\ell$ (condition \wref).
However, $\Tc''$ might now have too large height or size; we resolve this issue by invoking the height improvement (\heightop), resulting in a~tree decomposition $\Tc'''$ of $G'$ of width at most $\ell$, height at most $2^{\Oh{k \log k \sqrt{\log n \log \log n}}}$ and size at most $k^{\Oh{k}} \cdot n^3$.
The final tree decomposition $\Tc'''$ satisfies all the prescribed invariants.

Deleting an~edge $uv$ from $G$ is much simpler: we do not change the tree structure $T$ or the contents of the bags of the tree decomposition.
However, we still need to remove $uv$ from the appropriate set $\edges(t)$ for $t \in V(T)$ using a~prefix-rebuilding update.
Locating the appropriate $t \in V(T)$ is done as before: let $t_u = \Top(u)$, $t_v = \Top(v)$ and $uv$ belongs to either $\edges(t_u)$ or $\edges(t_v)$, depending on whether $t_u$ or $t_v$ has smaller height.
Assuming without loss of generality that $uv \in \edges(t_v)$, we construct a~prefix-rebuilding update removing $uv$ from $\edges(t_v)$ by rebuilding the path from the root of $T$ to $t_v$.


%
%

\subsection{Complexity analysis} \label{subsec:complexity}
Now, we are going to bound the amortized running time of the data structure. As already mentioned, we start with an edgeless graph and a decomposition $\Tc = (T, \bag, \edges)$ of it with potential $\Phi(\Tc) = \Oh{kn}$.

We are only going to focus on the analysis of edge insertions; edge removals are more straightforward.
Computing $P_u, P_v$ and performing the introductory prefix-rebuilding updates takes time
\begin{equation}
\label{eq:wrapup-time-initial}
 2^{\Oh{k^8}} \cdot \height(T) \leq 2^{\Oh{k^8}} \cdot 2^{\Oh{k \log k \sqrt{\log n \log \log n}}}.
\end{equation}
Let $G'$ be the graph $G$ with an~edge $uv$ added and $\Tc' = (T, \bag', \edges')$ be the tree decomposition of $G'$ after the initial prefix-rebuilding updates. (We stress that the decompositions $\Tc$ and $\Tc'$ have the same shape $T$.)
If $u$ was not added to any bags, the potential of $\Tc'$ has not increased and the entire query took time $2^{\Oh{k \log k \sqrt{\log n \log \log n}}}$.
Assume now that $u$ has been added to some new bags.
Recall then that $u$ was added to at most $\Oh{\height(T)}$ new bags, increasing the size of each of them to at most $\ell + 2 = 6k+7$; the sizes of other bags remained unchanged.
Therefore,
\begin{equation}
\label{eq:wrapup-potential-initial}
\begin{split}
\Phi(\Tc') - \Phi(\Tc) & \leq g(6k + 7) \cdot \Oh{\height(T)^2} \leq 2^{\Oh{k \log k}} \cdot \height(T)^2 \\
& \leq 2^{\Oh{k \log k \sqrt{\log n \log \log n}}}.
\end{split}
\end{equation}
Let $\Tc''$ be the decomposition produced by applying $\mathsf{refine}(\Tpref)$ on $\Tc'$ with a~prefix $\Tpref$ with $|\Tpref| \leq \Oh{\height(T)}$.
Recall the property \amref{potential-difference} bounding the potential change:

\[\Phi(\Tc') - \Phi(\Tc'') \geq \sum_{t \in \Tpref} \height_T(t)  - 2^{\Oh{k \log k}} \cdot \left( |\Tpref| + \sum_{t \in \App(\Tpref)} \height_T(t) \right) \cdot \log |V(T)|.\]

%
%

Therefore,
\begin{equation}
\label{eq:wrapup-potential-refine}
\begin{split}
\Phi(\Tc'') - \Phi(\Tc') &\le
  -\sum_{t \in \Tpref} \height_T(t) + 2^{\Oh{k \log k}} \cdot \left( |\Tpref| + \sum_{t \in \App(\Tpref)} \height_T(t) \right) \cdot \log |V(T)| \\
  &\le 2^{\Oh{k \log k}} \cdot |\Tpref| \cdot \height(T) \cdot \log |V(T)| \\
  &\le 2^{\Oh{k \log k}} \cdot \height(T)^2 \cdot \log n \\
  &\le 2^{\Oh{k \log k \sqrt{\log n \log \log n}}}.
\end{split} 
\end{equation}
Here, we used the facts that $|\Tpref| \leq \Oh{\height(T)}$, $|\App(\Tpref)| \le |\Tpref| + 1$ and $\log|V(T)| \leq \Oh{k \log k + \log n}$ (following from the invariants and the fact that $T$ is binary).

Finally, let $\Tc'''$ be the final tree decomposition produced by running the height improvement operation (\heightop) on $\Tc''$.
Since \heightop never increases the potential value, we have $\Phi(\Tc''') \leq \Phi(\Tc'')$.
We conclude that

\[
\Phi(\Tc''') - \Phi(\Tc) \leq 2^{\Oh{k \log k}} \cdot \height(T)^2 \cdot \log n \leq 2^{\Oh{k \log k \sqrt{\log n \log \log n}}}.
\]

For the running time of the edge insertion, recall from property \rtref{running-time} that the application of the refinement operation on $\Tc'$ takes worst-case time
\begin{equation} \label{eq:wrapup-time-refine} \begin{split}
&\quad 2^{\Oh{k^9}}\left(|\Tpref| + \sum_{t \in \App(\Tpref)} \height(t) + \max(\Phi(\Tc') - \Phi(\Tc''), 0)\right) \\
\leq&\quad 2^{\Oh{k^9}} \left(\height(T)^2 + \max(\Phi(\Tc') - \Phi(\Tc''), 0)\right) \\
\stackrel{(\ref{eq:wrapup-potential-refine})}{\leq}&\quad 2^{\Oh{k^9}} \left(\height(T)^2 + (\Phi(\Tc') - \Phi(\Tc'')) + 2^{\Oh{k \log k}} \cdot \height(T)^2 \cdot \log n \right) \\
\leq&\quad 2^{\Oh{k^9}} \left(2^{\Oh{k \log k \sqrt{\log n \log \log n}}} + \Phi(\Tc') - \Phi(\Tc'')\right);
\end{split} \end{equation}
immediately followed by the height improvement, which by \cref{lem:height-data-structure} takes worst-case time
\begin{equation}
\label{eq:wrapup-time-height}
2^{\Oh{k^9}} \cdot (\Phi(\Tc'') - \Phi(\Tc''')) + \Oh{1}.
\end{equation}
Thus, the total running time is bounded by the sum of \cref{eq:wrapup-time-initial,eq:wrapup-time-refine,eq:wrapup-time-height}:
\[
\begin{split}
& \quad 2^{\Oh{k^9 + k\log k \sqrt{\log n \log \log n}}} + 2^{\Oh{k^9}}(\Phi(\Tc') - \Phi(\Tc''')) \\
\stackrel{(\ref{eq:wrapup-potential-initial})}{\leq} & \quad 2^{\Oh{k^9 + k \log k \sqrt{\log n \log \log n}}} + 2^{\Oh{k^9}}(\Phi(\Tc) - \Phi(\Tc''')).
\end{split}
\]

Since each update increases the potential value by at most $2^{\Oh{k \log k \sqrt{\log n \log \log n}}}$, it follows that the amortized time complexity of each update is $2^{\Oh{k^9 + k \log k \sqrt{\log n \log \log n}}}$, as claimed.

\section{Conclusions}\label{sec:conclusions}

\wn{Should we get rid of these $\Oh[k]$ everywhere and be more specific?}
We presented a data structure for the dynamic treewidth problem that achieves amortized update time $2^{\Oh[k]{\sqrt{\log n\log \log n}}}$. The obvious open question is to improve this complexity. It is plausible that some optimization of the current approach could result in shaving off the $\sqrt{\log \log n}$ factor in the exponent, but complexity of the form $2^{\Oh[k]{\sqrt{\log n}}}$ seems inherent to the recursive approach presented in \cref{sec:height}. Nevertheless, we conjecture that it should be possible to achieve update time that is polylogarithmic in $n$ for every fixed $k$, that is, of the form $\log^{\Oh[k]{1}} n$, or maybe even $\Oh[k]{\log^c n}$ for some universal constant $c$. The ultimate goal would be to get closer to the $\Oh{\log n}$ bound achieved by Bodlaender for the case $k=2$~\cite{Bodlaender93a}.

Apart from the above, we hope that our result may open new directions in the design of parameterized dynamic data structures. More precisely, dynamic programming on tree decompositions is used as a building block in multiple different techniques in algorithm design, so \cref{thm:main} may be useful for designing dynamic counterparts of those techniques. Here is list some possible directions.
\begin{itemize}
 \item \cref{thm:main} provides the dynamic variant of the most basic formulation of Courcelle's Theorem. It seems that dynamic variants of the optimization formulation and the counting formulation, due to Arnborg et al.~\cite{ArnborgLS91}, just follow from combining the data structure of \cref{thm:main} with suitably constructed automata. It would be interesting to see whether \cref{thm:main} can be used also in the context of more general problems concerning $\CMSO_2$ queries on graphs of bounded treewidth, for instance the problem of query enumeration; see e.g.~\cite{Bagan06,KazanaS13}.
 \item {\em{Bidimensionality}} and {\em{shifting}} are two basic techniques used for designing parameterized algorithms in planar graphs, or more generally, graphs excluding a fixed minor; see~\cite[Section~7.7]{platypus}. As both are based on solving the problem efficiently on graphs of bounded treewidth, one could investigate whether dynamic counterparts could be developed using \cref{thm:main}.
 \item Related to the point above, it would be interesting to see whether \cref{thm:main} could be used to design a dynamic counterpart of Baker's technique~\cite{Baker94}. This would result in dynamic data structures for maintaining approximation schemes for problems like {\sc{Independent Set}} or {\sc{Dominating Set}} on planar graphs.
 \item The {\em{irrelevant vertex technique}} is a classic principle in parameterized algorithms, based on iteratively deleting vertices from a graph while not changing the answer to the problem, until the graph in question has bounded treewidth and can be tackled directly using dynamic programming. While certainly challenging, it does not seem impossible to derive dynamic variants of some of the algorithms obtained using the irrelevant vertex technique. A concrete candidate here would be the {\sc{Disjoint Paths}} problem on planar graphs, which admits a relatively simple and well-understood irrelevant vertex rule~\cite{AdlerKKLST17}
 \item {\em{Meta-kernelization}}, due to Bodlaender et al.~\cite{BodlaenderFLPST16}, is a powerful meta-technique for obtaining small kernels for parameterized problems on topologically-constrained graphs. The fundamental concept in meta-kernelization is {\em{protrusion}}: a portion of the graph in question that induces a subgraph of bounded treewidth and communicates with the rest of the graph through a bounded-size interface. Kernelization algorithms obtained through meta-kernelization use reduction rules based on {\em{protrusion replacement}}: finding a large protrusion, understanding it using tree-decomposition-based dynamic programming, and replacing it with a small gadget of the same functionality. It would be interesting to see if in some basic settings, \cref{thm:main} could be applied in combination with meta-kernelization to obtain data structures for maintaining small kernels in dynamic graphs. {\sc{Dominating Set}} on planar graphs would probably be the first problem to look into, as it admits a simple kernelization algorithms that predates (and inspired) meta-kernelization~\cite{AlberFN04}. We remark that dynamic kernelization has already been studied as a way to obtain dynamic parameterized data structures~\cite{AlmanMW20,IwataO14}.
\end{itemize}

\paragraph*{Acknowledgements.} The authors thank Anna Zych-Pawlewicz for many insightful discussions throughout the development of this project.

\bibliographystyle{plain}
\bibliography{references}

\appendix

\section{Dynamic automata}
\label{sec:dynamic-dynamic-programming}
In this section we introduce a framework for dynamic maintenance of dynamic programming tables on tree decompositions under prefix-rebuilding updates. Concrete outcomes of this are proofs of \cref{lem:height-maintenance}, \cref{cor:dynamic-cmso}, and \cref{lem:closure-maintenance}, but the introduced framework is general enough to also capture maintenance of any reasonable dynamic programming scheme. 

We remark that maintenance of runs of automata on dynamic forests has already been investigated in the literature, and even for the much more general problem of dynamic enumeration; see for instance the works of Niewerth~\cite{Niewerth18} and of Amarilli et al.~\cite{AmarilliBMN19} and the bibliographic discussion within. In particular, many (though not all) results contained in this section could be in principle derived from~\cite[Lemma~7.3]{abs-1812-09519}, but not in a black-box manner and without concrete bounds on update time. Therefore, for the sake of completeness, in this section we provide a self-contained presentation.

\subsection{Tree decomposition automata}

Our framework is based on a notion of automata processing tree decompositions. While this notion is tailored here to our specific purposes, the idea of processing tree decompositions using various kinds of automata or dynamic programming procedures dates back to the work of Courcelle~\cite{Courcelle90} and is a thoroughly researched topic; see appropriate chapters of textbooks~\cite{DowneyF13,platypus,FlumG06} and bibliographic notes within. Hence, the entirety of this section can be considered a formalization of folklore.

\newcommand{\Uni}{\Omega}

Throughout this section we assume that all vertices of considered graphs come from a fixed, totally ordered, countable set of vertices $\Uni$. Further, we assume that elements of $\Uni$ can be manipulated upon in constant time in the RAM model. The reader may assume that $\Uni=\N$.

\paragraph*{Boundaried graphs.} We will work with graphs with specified boundaries, as formalized next.

\begin{definition}
 A {\em{boundaried graph}} is an undirected graph $G$ together with a vertex subset $\bnd G\subseteq V(G)$, called the {\em{boundary}}, such that $G$ has no edge with both endpoints in $\bnd G$.
 A {\em{boundaried tree decomposition}} of a boundaried graph $G$ is a triple $(T,\bag,\edges)$ that is an annotated tree decomposition of $G$ (treated as a normal graph) where in addition we require that $\bnd G$ is contained in the root~bag.
\end{definition}

When speaking about a boundaried tree decomposition $(T,\bag,\edges)$ of a boundaried graph~$G$, we redefine the adhesion of the root of $T$ to be $\bnd G$, rather than the empty set.

Suppose $(T,\bag,\edges)$ is a boundaried tree decomposition of a boundaried graph $G$, and $x$ is a node of $T$. Then we say that $x$ {\em{induces}} a boundaried graph $G_x$ and its boundaried tree decomposition $(T_x,\bag_x,\edges_x)$, defined as follows: if $X$ is the set of descendants of $x$ in $T$, then
\begin{eqnarray*}
 G_x & = & \left(\bigcup_{y\in X} \bag(y),\bigcup_{y\in X} \edges(y)\right);\\
 \bnd G_x & = & \adhesion{x};\\
 (T_x,\bag_x,\edges_x) & = & \funrestriction{(T,\bag,\edges)}{X}.
\end{eqnarray*}
It is clear that $(T_x,\bag_x,\edges_x)$ defined as above is a boundaried tree decomposition of $G_x$. We use the above notation only when the boundaried tree decomposition $(T,\bag,\edges)$ is clear from the context.

\paragraph*{Automata.} We now introduce our automaton model.

\newcommand{\Bb}{\mathcal{BK}}
\newcommand{\Rr}{\mathcal{R}}
\newcommand{\init}{\iota}
\newcommand{\trans}{\delta}
\newcommand{\run}{\rho}

\begin{definition}\label{def:tda}
 A {\em{(deterministic) tree decomposition automaton}} of {\em{width}} $\ell$
 consists of 
 \begin{itemize}
  \item a {\em{state set}} $Q$;
  \item a set of {\em{accepting states}} $F\subseteq Q$;
  \item an {\em{initial mapping}} $\init$ that maps every boundaried graph $G$ on at most $\ell+1$ vertices to a state $\init(G)\in Q$; and
  \item a {\em{transition mapping}} $\trans$ that maps every $7$-tuple of form $(B,X,Y,Z,J,q',q'')$, where $B\subseteq \Omega$ is a set of size at most $\ell+1$, $X,Y,Z\subseteq B$, $J\in \binom{B}{2}\setminus \binom{X}{2}$, $q'\in Q$, and $q'' \in Q \cup \{\bot\}$ to a state $\trans(B,X,Y,Z,J,q',q'')\in Q$.
 \end{itemize}
 The {\em{run}} of a tree decomposition automaton $\Aa$ on a binary boundaried tree decomposition $(T,\bag,\edges)$ of a boundaried graph $G$ is the unique labelling $\run_\Aa\colon V(T)\to Q$ satisfying the following properties:
 \begin{itemize}
  \item For every leaf $l$ of $T$, we have $$\run_\Aa(l)=\init(G_l).$$
  \item For every non-leaf node $x$ of $T$ with one child $y$, we have
  $$\run_\Aa(x)=\trans(\bag(x),\adhesion{x},\adhesion{y},\emptyset,\edges(x),\run_\Aa(y),\bot).$$
  \item For every non-leaf node $x$ of $T$ with two children $y$ and $z$, we have
  $$\run_\Aa(x)=\trans(\bag(x),\adhesion{x},\adhesion{y},\adhesion{z},\edges(x),\run_\Aa(y),\run_\Aa(z)).$$
 \end{itemize}
 A tree decomposition automaton $\Aa$ {\em{accepts}} $(T,\bag,\edges)$ if $\run_\Aa(r)\in F$, where $r$ is the root of $T$.
\end{definition}

Note that in the transitions described above, nodes with one child are treated by passing a ``dummy state'' $\bot \notin Q$ to the transition function instead of a state. Note that this allows $\trans$ also to recognize when there is only one child.
Also, the automata model presented above could in principle distinguish the left child $y$ from the right child $z$ and treat states passed from them differently. However, this will never be the case in our applications: in all constructed automata, the transition mapping will be symmetric with respect to swapping the role of the children $y$ and $z$.

\medskip

We say that a tree decomposition automaton $\Aa$ has {\em{evaluation time}} $\tau$ if functions $\init$ and $\trans$ can be evaluated on any tuple of their arguments in time $\tau$, and moreover for a given $q\in Q$ it can be decided whether $q\in F$ in time $\tau$. Note that we do {\em{not}} require the state space $Q$ to be finite. In fact, in most of our applications it will be infinite, but we will be able to efficiently represent and manipulate the states.

We will often run a tree decomposition automaton on a non-boundaried annotated tree decomposition of a non-boundaried graph $G$. In such cases, we simply apply all the above definitions while treating $G$ as a boundaried graph with an empty boundary.

We will also use {\em{nondeterministic tree decomposition automata}}, which are defined just like in \cref{def:tda}, except that $\init$ and $\trans$ are the {\em{initial relation}} and the {\em{transition relation}}, instead of mappings. That is, $\init$ is a relation consisting of pairs of the form $(G,q)$, where $G$ is a boundaried graph on at most $\ell+1$ vertices, and $q\in Q$. Similarly, $\trans$ is a relation consisting of pairs of the form $((B,X,Y,Z,J,q',q''),q)$, where $(B,X,Y,Z,J,q',q'')$ is a $7$-tuple like in the domain of the transition mapping, and $q\in Q$. Then a run of a nondeterministic automaton $\Aa$ on a boundaried binary tree decomposition $(T,\bag,\edges)$ is a labelling $\run$ of the nodes of $T$ with states such that $(G_l,\run(l))\in \init$ for every leaf $l$, $((\bag(x),\adhesion{x},\adhesion{y},\emptyset,\edges(x),\run(y),\bot),\run(x))\in \trans$ for every node $x$ with one child $y$, and $((\bag(x),\adhesion{x},\adhesion{y},\adhesion{z},\edges(x),\run(y),\run(z)),\run(x))\in \trans$ for every node $x$ with two children $y$ and $z$. Note that a nondeterministic tree decomposition automaton may have multiple runs on a single tree decomposition. We say that $\Aa$ accepts $(T,\bag,\edges)$ if there is a run of $\Aa$ on $(T,\bag,\edges)$ that is {\em{accepting}}: the state associated with the root node is accepting.

In the context of nondeterministic automata, by evaluation time we mean the time needed to decide whether a given pair belongs to any of the relations $\init$ or $\trans$, or to decide whether a given state is accepting.
Note that if $\Aa$ is a nondeterministic tree decomposition automaton with a finite state space $Q$, then we can determinize it --- find a deterministic automaton $\Aa'$ that accepts the same tree decompositions --- using the standard powerset construction. Then the state space of $\Aa'$ is $2^Q$.
In the following, all automata are deterministic unless explicitly stated.

\subsection{Automata constructions}

In subsequent sections we will use several automata. 
We now present four automata constructions that we will use.

\paragraph*{Tree decomposition properties automata.} We first construct three very simple automata that are used in \cref{lem:height-maintenance} for maintaining properties of the tree decomposition itself.

\begin{lemma}\label{lem:automaton-height}
For every $\ell\in \N$ there exists tree decomposition automata $\Hh_{\ell}$, $\mathcal{S}_{\ell}$, $\Cc_{\ell}$, each of width $\ell$, with the following properties: for any graph $G$, annotated binary tree decomposition $(T,\bag,\edges)$ of $G$ of width at most $\ell$, and any node $x$ of $T$:
\begin{itemize}
\item $\run_{\Hh_\ell}(x)$ is equal to $\height(T_x)$,
\item $\run_{\mathcal{S}_\ell}(x)$ is equal to $|V(T_x)|$, and
\item $\run_{\Cc_\ell}(x)$ is equal to $|\component{x}|$.
\end{itemize}
 The evaluation times of $\Hh_\ell$ and $\mathcal{S}_\ell$ are~$\Oh{1}$, and the evaluation time of $\Cc_\ell$ is~$\Oh{\ell}$.
\end{lemma}
\begin{proof}
The state sets of each of the automata are $\N$.
Let us define for all $n \in \N$ that $\max(n, \bot) = n$ and $n + \bot = n$.
For the height automaton $\Hh_\ell$, the initial mapping and the transition mapping are defined as follows (here $\_$ denotes any input value):
 \begin{eqnarray*}
  \init(\_) & = & 1 \\
  \trans(\_,\_,\_,\_,\_,q',q'') & = & 1+\max(q',q'').
 \end{eqnarray*}
 For the size automaton $\mathcal{S}_\ell$, the initial mapping and the transition mapping are defined as follows:
 \begin{eqnarray*}
 \init(\_) & = & 1 \\
 \trans(\_,\_,\_,\_,\_,q',q'') & = & 1+q'+q''.
 \end{eqnarray*}
 For the $|\component{x}|$ automaton $\Cc_\ell$, the initial mapping and the transition mapping are defined as follows:
 \begin{eqnarray*}
 \init(G) & = & |V(G) \setminus \bnd G| \\
 \trans(B, X, \_, \_, \_, q', q'') & = & q' + q'' + |B \setminus X|.
 \end{eqnarray*}
 It is straightforward to see that these automata satisfy the required properties.
\end{proof}

\paragraph*{$\CMSO_2$-types automaton.} The classic Courcelle's theorem~\cite{Courcelle90} states that there is an algorithm that given a $\CMSO_2$ sentence $\varphi$ and an $n$-vertex graph $G$ together with a tree decomposition of width at most $\ell$, decides whether $G\models \varphi$ in time $f(\ell,\varphi)\cdot n$, where $f$ is a computable function. One way of proving Courcelle's theorem is to construct a dynamic programming procedure that processes the provided tree decomposition in a bottom-up fashion. This dynamic programming procedure can be understood as a tree decomposition automaton in the sense of \cref{def:tda}, yielding the following result. The proof is a completely standard application of the concept of $\CMSO_2$-types, hence we only sketch it.

\newcommand{\tp}{\mathrm{tp}}
\newcommand{\Sent}{\mathrm{Sentences}}

\begin{lemma}\label{lem:automaton-CMSO}
 For every integer $\ell$ and $\CMSO_2$ sentence $\varphi$ there exists a tree decomposition automaton $\Aa_{\ell,\varphi}$ of width $\ell$ with the following property: for any graph $G$ and its annotated binary tree decomposition $(T,\bag,\edges)$ of width at most $\ell$, $\Aa_{\ell,\varphi}$ accepts $(T,\bag,\edges)$ if and only if $G\models \varphi$. The evaluation time is bounded by $f(\ell,\varphi)$ for some computable function $f$.   
\end{lemma}
\begin{proof}[Proof sketch]
 Let $p$ be the {\em{rank}} of $\varphi$: the maximum among the quantifier rank of $\varphi$ and the moduli of modular predicates appearing in $\varphi$. Consider any finite $X\subseteq \Omega$ and let $\CMSO_2(X)$ consists of all $\CMSO_2$ sentences that can additionally use elements of $X$ as constants (formally, these are $\CMSO_2$ sentences over the signature of graphs enriched by adding every $x\in X$ as a constant).
 It is well-known (see e.g.~\cite[Exercise~6.11]{Immerman99}) that for a given $X$ as above, one can compute a set $\Sent^p(X)$ consisting of at most $g(\ell,|X|)$ sentences of $\CMSO_2(X)$, for some computable $g$, such that for every $\psi\in \CMSO_2(X)$ of rank at most $p$, there is $\psi'\in \Sent^p(X)$ that is equivalent to $\psi$ in the sense of being satisfied in exactly the same graphs containing $X$. Also, the mapping $\psi\mapsto \psi'$ is computable.
 
 For a given boundaried graph $G$ with $X=\bnd G$, we define the {\em{$p$-type}} of $G$ as follows:
 $$\tp^p(G)=\{\psi\in \Sent^p(X)~|~G\models \psi\}.$$
 Now, we construct automaton $\Aa=\Aa_{\ell,\varphi}$ so that for every boundaried binary tree decomposition $(T,\bags,\edges)$ of a graph $G$, the run of $\Aa$ on $(T,\bags,\edges)$ is as follows:
 $$\run_{\Aa}(x)=\tp^p(G_x)\qquad\textrm{for all }x\in V(T).$$
 When constructing $\Aa$, the only non-trivial check is that one can define a suitable transition mapping $\trans$. For this, it suffices to show that for every node $x$ of $T$ with children $y$ and $z$, given types $\tp^p(G_y)$ and $\tp^p(G_z)$ together with the information about the bag of $x$ (consisting of $\bag(x)$, $\adhesion{x}$, $\adhesion{y}$, $\adhesion{z}$, and $\edges(x)$), one can compute the type $\tp^p(G_x)$; and same for nodes with one child. This follows from a standard argument involving Ehrenfeucht-Fra\"isse games, cf.~\cite{GroheK09,Makowsky04}.
\end{proof}

\paragraph*{Bodlaender-Kloks automaton.} In~\cite{DBLP:journals/jal/BodlaenderK96}, Bodlaender and Kloks gave an algorithm that given a graph~$G$, a binary tree decomposition of $G$ of width at most $\ell$, and a number $k\leq \ell$, decides whether the treewidth of $G$ is at most $k$ in time $2^{\Oh{k\ell^2}}\cdot n$, where $n$ is the vertex count of $G$. This algorithm proceeds by bottom-up dynamic programming on the provided tree decomposition of $G$, computing, for every node $x$, a table consisting of $2^{\Oh{k\ell^2}}$ boolean entries. Intuitively, each entry encodes the possibility of constructing a partial tree decomposition of the subgraph induced by the subtree at $x$ with a certain ``signature'' on the adhesion of $x$; the number of possible signatures is  $2^{\Oh{k\ell^2}}$.

Inspecting the proof provided in~\cite{DBLP:journals/jal/BodlaenderK96} it is not hard to see that this dynamic programming can be understood as a nondeterministic tree decomposition automaton. Thus, from the work of Bodlaender and Kloks we can immediately deduce the following statement.

\begin{lemma}\label{lem:automaton-BK}
 For every pair of integers $k\leq \ell$ there is a nondeterministic tree decomposition automaton $\Bb_{k,\ell}$ of width $\ell$ with the following property: for any graph $G$ and its binary annotated tree decomposition $(T,\bag,\edges)$ of width at most $\ell$, $\Bb_{k,\ell}$ accepts $(T,\bag,\edges)$ if and only if the treewidth of $G$ is at most $k$. The state space of $\Bb_{k,\ell}$ is of size $2^{\Oh{k\ell^2}}$ and can be computed in time $2^{\Oh{k\ell^2}}$. The evaluation time of $\Bb_{k,\ell}$ is $2^{\Oh{k\ell^2}}$ as well.
\end{lemma}

We remark that since the property of having treewidth at most $k$ can be expressed in $\CMSO_2$\footnote{This can be done, for instance, by stating that the given graph does not contain any of the forbidden minor obstructions for having treewidth at most $k$. It is known that the sizes of such obstructions are bounded by a doubly-exponential function in $k^5$, hence their number is at most triply-exponential in $k^5$~\cite{Lagergren98}.}, \cref{lem:automaton-BK} with an unspecified bound on the evaluation time also follows from \cref{lem:automaton-CMSO}. The reason behind formulating \cref{lem:automaton-BK} explicitly is to keep track of the evaluation time more precisely in further arguments.

\newcommand{\reps}{\mathsf{reps}}

\paragraph*{Closure automaton.} Finally, we introduce automata for computing small closures within subtrees of a tree decomposition. These will be used in the proof of \cref{lem:closure-maintenance}. We first need a few definitions.

Let $G$ be a boundaried graph. We say that two sets of non-boundary vertices $Y,Z\subseteq V(G)\setminus \bnd G$ are {\em{torso-equivalent}} if there is an isomorphism between $\torso[G]{Y\cup \bnd G}$ and $\torso[G]{Z\cup \bnd G}$ that fixes every vertex of $\bnd G$. Note that being torso-equivalent is an equivalence relation and for every integer $c$, let $\sim_{c,G}$ be the restriction of this equivalence relation to subsets of $V(G)\setminus \bnd G$ of cardinality at most $c$. Note $\sim_{c,G}$ has at most $2^{\Oh{(c+|\bnd G|)^2}}$ equivalence classes.

Suppose further that $\Tc = (T,\bag,\edges)$ is a boundaried tree decomposition of $G$.
We generalize the depth function $d_\Tc$ defined in \cref{sec:preliminaries} to boundaried tree decompositions as follows: let $d_\Tc \colon V(G)\setminus \bnd G \,\to\, \Z_{\geq 0}$ be such that $d_\Tc(u)$ is the depth of the top-most node of $T$ whose bag contains $u$.
In particular, $d_\Tc(u)$ depends on the vertex $u$ and the tree decomposition $(T,\bag,\edges)$.
Further, for a set of vertices $Y\subseteq V(G)\setminus \bnd G$, we define $d_\Tc(Y)=\sum_{u\in Y} d_\Tc(u)$.
Recalling that there is a total order $\preceq$ on the vertices of $G$ inherited from $\Uni$, subsets of $V(G)\setminus \bnd G$ can be compared as follows: for $Y,Y'\subseteq V(G)\setminus \bnd G$, we set $Y\preceq_\Tc Y'$ if
\begin{itemize}
 \item $d_\Tc(Y)<d_\Tc(Y')$, or
 \item $d_\Tc(Y)=d_\Tc(Y')$ and $Y$ is lexicographically not larger than $Y'$ with respect to $\preceq$.
\end{itemize}
For a nonempty set of subsets $\mathcal{S}$ of $V(G)\setminus \bnd G$, we let $\min_\Tc \mathcal{S}$ be the $\preceq_\Tc$-smallest element of $\mathcal{S}$. This allows us to define the {\em{$c$-small torso representatives}} as follows:
\begin{align*}
&\reps^c(G,(T,\bag,\edges)) \coloneqq\\
&\left\{(d_\Tc({\textstyle \min_\Tc} \mathcal{K}),\torso[G]{{\textstyle \min_\Tc} \mathcal{K} \cup \bnd G})\colon \mathcal{K}\textrm{ is an equivalence class of }\sim_{c,G}\right\}.
\end{align*}
Note that each member of an equivalence class of $\sim_{c,G}$ has the same size, so at this point we do not optimize for the size even though it will later be needed for the proof of \cref{lem:closure-maintenance}.
The set $\reps^c(G,(T,\bag,\edges))$ depends on both the boundaried graph $G$ and its boundaried tree decomposition $(T,\bag,\edges)$. Also, the cardinality of $\reps^c(G,(T,\bag,\edges))$ is equal to the number of equivalence classes of $\sim_{c,G}$, which, as noted, is at most $2^{\Oh{(c+|\bnd G|)^2}}$.

The closure automata we are going to use are provided by the following statement.

\begin{lemma}\label{lem:closure-automaton}
 For every pair of integers $c,\ell$ there is a tree decomposition automaton $\Rr=\Rr_{c,\ell}$ with the following property: for any graph $G$ and its annotated binary tree decomposition $(T,\bag,\edges)$ of width at most $\ell$, the run of $\Rr$ on $(T,\bag,\edges)$ satisfies
 $$\run_{\Rr}(x)=\reps^c(G_x,(T_x,\bag_x,\edges_x))\qquad\textrm{for all }x\in V(T).$$
 The evaluation time of $\Rr$ is $2^{\Oh{(c+\ell)^2}}$.
\end{lemma}
\begin{proof}
 For the state space $Q$ of $\Rr$ we take the set of all sets of pairs of the form $(p,H)$, where $p\in \N$ and $H$ is a  graph on at most $c+\ell+1$ vertices contained in $\Uni$. The final states of $\Rr$ are immaterial for the lemma statement, hence we can set $F=\emptyset$. As for the initial mapping $\init$, for a boundaried graph $H$ on at most $\ell+1$ vertices we can set $\init(H)=\reps^c(H,(T_0,\bags_0,\edges_0))$, where $(T_0,\bags_0,\edges_0)$ is the trivial one-node tree decomposition of $H$ in which all vertices and edges are put in the root bag. It is straightforward to see that $\init(H)$ can be computed in time $2^{\Oh{(c+\ell)^2}}$ directly from the definition.
 
 It remains to define the transition mapping $\delta$. For this, it suffices to prove the following. Suppose $(T,\bags,\edges)$ is a boundaried tree decomposition of a boundaried graph $G$ and $x$ is a node of $G$ with children $y$ and $z$. Then knowing $$R_y\coloneqq \reps^c(G_y,(T_y,\bags_y,\edges_y)),\qquad R_z\coloneqq \reps^c(G_z,(T_z,\bags_z,\edges_z)),$$
 as well as $\bag(x)$, $\edges(x)$, and adhesions of $x,y,z$, one can compute $$R_x\coloneqq \reps^c(G_x,(T_x,\bags_x,\edges_x))$$ in time $2^{\Oh{(c+\ell)^2}}$. Formally, we would also need such an argument for the case when $x$ has only one child $y$, but this follows from the argument for the case of two children by considering a dummy second child $z$ with an empty bag. So we focus only on the two-children case.
 
 For any pair of sets $Y\subseteq \component{y}$ and $Z\subseteq \component{z}$,
 define $G_x(Y,Z)$ to be the graph with vertex set $Y\cup Z\cup \bag(x)$ and edge set consisting of the union of the edge sets of the following graphs: 
 $$\torso[G_y]{Y\cup \adhesion{y}},\qquad \torso[G_z]{Z\cup \adhesion{z}},\qquad \textrm{and} \qquad (\bag(x),\edges(x)).$$
 The following claim is straightforward.
 
 \begin{claim}\label{cl:prt-torso}
  For any triple of sets $X\subseteq \bag(x)\setminus \adhesion{x}$, $Y\subseteq \component{y}$ and $Z\subseteq \component{z}$, we have
  $$\torso[G_x]{X\cup Y\cup Z\cup \adhesion{x}}=\torso[G_x(Y,Z)]{X\cup Y\cup Z\cup \adhesion{x}}.$$
 \end{claim}

 To compute $R_x$,
 we first construct a family of candidates $C$ as follows. Consider every pair of pairs $(p_y,H_y)\in R_y$ and $(p_z,H_z)\in R_z$, and every $X\subseteq \bag(x)\setminus \adhesion{x}$. Let $Y=V(H_y)\setminus \adhesion{y}$ and $Z=V(H_z)\setminus \adhesion{z}$, and note that the graph $G_x(Y,Z)$ is the union of graphs $H_y$, $H_z$, and $(\bag(x),\edges(x))$. If $|X\cup Y\cup Z|\leq c$, then we add to $C$ the pair
 $$(p_y+p_z+|Y|+|Z|,\torso[G_x(Y,Z)]{X\cup Y\cup Z\cup \adhesion{x}}).$$
 Otherwise, if $|X\cup Y\cup Z|>c$, no pair is added to $C$. 
The first coordinate of the pair added to $C$ is equal to $d_{T_x}(X \cup Y \cup Z)$ because $X \subseteq \bag(x)$ and $Y,Z \subseteq V(G_x) \setminus \bag(x)$.
By \cref{cl:prt-torso}, the second coordinate of the pair added to $C$ is equal to $\torso[G_x]{X\cup Y\cup Z\cup \adhesion{x}})$.

 Further, we have 
 $$|C|\leq |R_y|\cdot |R_z|\cdot 2^{\ell+1}\leq 2^{\Oh{(c+\ell)^2}},$$
 and $C$ can be computed in time $2^{\Oh{(c+\ell)^2}}$.
 
 Next, the candidates are filtered as follows. As long as in $C$ there is are distinct pairs $(p,H)$ and $(p',H')$ such that $H$ and $H'$ are isomorphic by an isomorphism that fixes $\adhesion{x}$, we remove the pair that has the larger first coordinate; if both pairs have the same first coordinate, remove the one where the vertex set of the second coordinate is larger in $\preceq$. Clearly, this filtering procedure can be performed exhaustively in time $2^{\Oh{(c+\ell)^2}}$.

 Thus, after filtering, all second coordinates of the pairs in $C$ are pairwise non-equivalent in $\sim_{c,G_x}$. It is now straightforward to see using a simple exchange argument that $R_x$ is equal to $C$ after the filtering. This constitutes the definition and the algorithm computing the transition function~$\trans$. 
\end{proof}

\subsection{Dynamic maintenance of automata runs}

Having defined the automata we are going to use, we now show how to maintain their runs effectively under prefix-rebuilding updates.

\begin{lemma}\label{lem:automaton-maintenance}
  Fix $\ell\in \N$ and a tree decomposition automaton $\Aa=(Q,F,\init,\trans)$ of width $\ell$ and evaluation time $\tau$. 
  Then there exists an~$\ell$-prefix-rebuilding data structure with overhead $\tau$ that additionally implements the following operation:
  \begin{itemize}
    \item $\mathsf{query}(x)$: given a~node $x$ of $T$, returns $\run_\Aa(x)$. Runs in worst-case time $\Oh{1}$.
  \end{itemize}
\end{lemma}
\begin{proof}
 At every point in time, the data structure stores the decomposition $(T,\bag,\edges)$, where every node $x$ is supplied with a pointer to its parent, a pair of pointers to its children, and the state $\rho_{\Aa}(x)$ in the run of $\Aa$ on $(T,\bag,\edges)$. This allows for answering queries in constant time, as requested.
 
 For initialization, we just compute the run of $\Aa$ on $(T,\bag,\edges)$ in a bottom-up manner: the states for leaves are computed according to the initialization mapping, while the states for internal nodes are computed according to the transition mapping bottom-up. This requires time $\tau$ per node, so $\Oh{\tau\cdot |V(T)|}$ in total.
 
 For applying a prefix-rebuilding update $\tup u=(\Tpref, \Tpref', T^\star, \bag^\star, \edges^\star, \pi)$, the pointer structure representing the decomposition can be easily rebuilt in time $\ell^{\Oh{1}}\cdot |\tup u|$ by building the tree $\Tpref'$ and reattaching all appendices of $\Tpref$ in $T$ according to $\pi$, using a single pointer change per appendix. Observe here that the information about the run of $\Aa$ on the reattached subtrees does not need to be altered, except for the appendices of $\Tpref'$, for which the run could have to be altered because their adhesions could change.
 Hence, it remains to compute the states associated with the nodes of $\Tpref' \cup \App(\Tpref')$ in the run of $\Aa$ on the new decomposition $(T',\bag',\edges')$. This can be done by processing $\Tpref' \cup \App(\Tpref')$ in a bottom-up manner, and computing each consecutive state using either the initialization mapping $\init$ (for nodes in $\Tpref'$ that are leaves of $T'$) or the transition mapping $\trans$ (for the other nodes in $\Tpref'$), in total time $\Oh{\tau\cdot |\Tpref' \cup \App(\Tpref')|}\leq \Oh{\tau\cdot |\tup u|}$.
\end{proof}

Now, \cref{lem:automaton-maintenance} combined with \cref{lem:automaton-CMSO} immediately implies \cref{cor:dynamic-cmso}.
Similarly, an $\ell$-prefix-rebuilding data structure implementing the three first operations of \cref{lem:height-maintenance} follows immediately by applying \cref{lem:automaton-maintenance} to the automaton provided by \cref{lem:automaton-height}.
Let us here complete the proof of \cref{lem:height-maintenance}.

\begin{proof}[Proof of \cref{lem:height-maintenance}]
By above discussion, it suffices to implement an $\ell$-prefix-rebuilding data structure with overhead $\Oh{1}$ that implements the operation $\Top(v)$, that given a vertex $v \in V(G)$ returns the unique highest node $t$ of $T$ such that $v \in \bag(t)$.
Consider a prefix-rebuilding update changing $(T,\bag,\edges)$ to $(T',\bag',\edges')$.
Observe that a prefix-rebuilding update can change the highest node where $v$ occurs only if $v \in \bags_{T}(\Tpref \cup \App(\Tpref))$, and in particular, in that case the highest node of $(T',\bag',\edges')$ where $v$ occurs will be in $\Tpref' \cup \App(\Tpref')$.
Both $|\Tpref \cup \App(\Tpref)|$ and $|\Tpref' \cup \App(\Tpref')|$ are linear in $|\overline{u}|$, so we simply maintain the mapping $\Top(v)$ explicitly by recomputing it for all vertices $v \in \bags_{T}(\Tpref \cup \App(\Tpref))$.
\end{proof}

\subsection{Closures and blockages}
\label{ssec:closures-and-blockages}

We can now give a proof of \cref{lem:closure-maintenance}. This will require more work, as apart from dynamically maintaining relevant information we also need to implement methods for extracting a closure and blockages.

\begin{proof}[Proof of \cref{lem:closure-maintenance}]
 For the proof, we fix the following two automata:
 \begin{itemize}
  \item $\Rr=\Rr_{c,\ell}$ is the closure automaton for parameters $c$ and $\ell$, provided by \cref{lem:closure-automaton}.
  \item $\Bb=\Bb_{2k+1,c+\ell}$ is the Bodlaender-Kloks automaton for parameters $2k+1$ and $c+\ell$, provided by \cref{lem:automaton-BK}.
 \end{itemize}
 Let $\Bb=(Q,F,\init,\trans)$. Recall that $\Bb$ is nondeterministic, $|Q|\leq 2^{\Oh{k(c+\ell)^2}}$, $Q$ can be computed in time $2^{\Oh{k(c+\ell)^2}}$, and membership in $F$, $\init$, or $\trans$ for relevant objects can be decided in time $2^{\Oh{k(c+\ell)^2}}$.
 
 Our data structure just consists of the data structure provided by \cref{lem:automaton-maintenance} for the automaton~$\Rr$. Thus, the initialization time is $2^{\Oh{(c+\ell)^2}}\cdot |V(T)|$ and the update time is $2^{\Oh{(c+\ell)^2}}\cdot |\tup u|$ as requested. It remains to implement method $\mathsf{query}(\Tpref)$. For this, we may assume that the stored annotated tree decomposition $(T,\bag,\edges)$ is labelled with the run $\rho_\Rr$ of $\Rr$ on $(T,\bag,\edges)$. This means that for every $x\in V(T)$, we have access to $\reps^c(T_x,\bag_x,\edges_x)$.
 
 Consider a query $\mathsf{query}(\Tpref)$. We break answering this query into two steps. In the first step, we compute a $d_\Tc$-minimal $c$-small $k$-closure $X$ of $\bags(\Tpref)$, together with $\torso[G]{X}$. This will take time $2^{\Oh{k(c+\ell)^2}}\cdot |\Tpref|$. In the second step, we find $\blockages{\Tpref}{X}$. This will take time $2^{\Oh{(c+\ell)^2}}\cdot |\exploration{\Tpref}{X}|$.
 
 \paragraph*{Step 1: finding the closure and its torso.}
 Let $A$ be the set of appendices of $\Tpref$. For $a\in A$, let
 $$R(a)\coloneqq \reps^c(T_a,\bag_a,\edges_a);$$
 recall that $R(a)$ is stored along with $a$ in the data structure.
 Let $\Lambda$ be the set of all mappings $\lambda$ with domain $A$ such that $\lambda(a)\in R(a)$ for all $a$. For $\lambda\in \Lambda$ and $a\in A$, let $d^{\lambda}(a)$ and $H^{\lambda}(a)$ be the first, respectively the second coordinate of $\lambda(a)$, and let $s^\lambda(a) = |V(H^{\lambda}(a)) \setminus \adhesion{a}|$.
 For $x \in V(T)$ we denote by $d_T(x)$ the depth of $x$ in $(T, \bag, \edges)$.
For $\lambda\in \Lambda$ we define
 \begin{eqnarray*}
  H^\lambda & \coloneqq & \left(\bags(\Tpref)\cup \bigcup_{a\in A} V(H^{\lambda}(a)),\edges(\Tpref)\cup \bigcup_{a\in A} E(H^{\lambda}(a))\right),\\
  s^\lambda & \coloneqq & \sum_{a \in A} s^\lambda(a),\\
  d^\lambda & \coloneqq & \sum_{a\in A} \left(d^\lambda(a) + d_T(a) \cdot s^\lambda(a)\right).
 \end{eqnarray*}
 Note that by the definition of $\reps^c$, we have that
 \begin{itemize}
 \item $H^\lambda=\torso[G]{V(H^\lambda)}$,
 \item $s^\lambda=|V(H^\lambda)|-|\bags(\Tpref)|$, and
 \item $d^\lambda=d_\Tc(V(H^\lambda))-d_\Tc(\bags(\Tpref))$
\end{itemize}
 for all $\lambda\in \Lambda$.

 Further, for each $a\in A$, the set $R(a)$ comprises all possible non-isomorphic torsos that can be obtained by picking at most $c$ vertices within $\component{a}$, and with each possible torso $R(a)$ stores a realization with the least possible total depth. This immediately implies the following statement.
 
 \begin{claim}\label{cl:lambda}
  Let $\lambda\in \Lambda$ be such that the treewidth of $H^\lambda$ is at most $2k+1$ and, among such mappings $\lambda$, $s^\lambda$ is minimum, and among those $d^\lambda$ is minimum. Then $V(H^\lambda)$ is a $d_\Tc$-minimal $c$-small $k$-closure of $\bags(\Tpref)$. Further, if no $\lambda$ as above exists, then there is no $c$-small $k$-closure of $\bags(\Tpref)$.
 \end{claim}
 So, by \cref{cl:lambda}, it suffices to find $\lambda\in \Lambda$ that primarily minimizes $s^\lambda$ and secondarily $d^\lambda$ such that $H^\lambda$ has treewidth at most $k$, or conclude that no such $\lambda$ exists. Indeed, then we can output $X\coloneqq V(H^\lambda)$ and $\torso[G]{X}=H^\lambda$ as the output to the query. Note here that once a mapping $\lambda$ as above is found, one can easily construct $H^\lambda$ in time $(\ell+c)^{\Oh{1}}\cdot |\Tpref|$ right from the definition. Intuitively, to find a suitable $\lambda$ we analyze possible runs of the Bodlaender-Kloks automaton $\Bb$ on the natural tree decomposition of $H^\lambda$ inherited from $\Tpref$, for different choices of $\lambda$.
 
 \newcommand{\Tprefx}{T_{\mathrm{pref},x}}
 
 Let $S=\Tpref \cup A$.
 For $x \in S$ let $S_x$ be the subset of $S$ consisting of descendants of $x$ and $A_x = S_x \cap A$, and $\Lambda_x$ be defined just like $\Lambda$, but for domain $A_x$ instead of $A$.
 For $x \in V(T)$, let $\Tprefx \subseteq \Tpref$ consist of the nodes in $\Tpref$ that are descendants of $x$.
 For $x,y \in V(T)$, let $d_T(x,y)$ denote their distance in $T$. 
 Further, for $\lambda\in \Lambda_x$, define
 \begin{eqnarray*}
  H_x^\lambda & \coloneqq & \left(\bags(\Tprefx)\cup \bigcup_{a\in A_x} V(H^{\lambda}(a)),\edges(\Tprefx)\cup \bigcup_{a\in A_x} E(H^{\lambda}(a))\setminus \binom{\adhesion{x}}{2}\right),\\
  s_x^\lambda & \coloneqq & \sum_{a \in A_x} s^\lambda(a),\\
  d_x^\lambda & \coloneqq & \sum_{a\in A_x} \left(d^\lambda(a) + d_T(x,a) \cdot s^\lambda(a)\right).
 \end{eqnarray*}
 We treat $H_x^\lambda$ as a boundaried graph with boundary $\adhesion{x}$. Then, $H_x^\lambda$ has a boundaried tree decomposition $(T^\lambda_x,\bag^\lambda_x,\edges^\lambda_x)$ naturally inherited from $(T,\bag,\edges)$ as follows:
 \begin{itemize}
  \item $T^\lambda_x=T[S_x]$;
  \item $\bag^\lambda_x(y)=\bag(y)$ for all $y\in \Tprefx$, and $\bag^\lambda_x(a)=V(H^\lambda(a))$ for all $a\in A_x$;
  \item $\edges^\lambda_x(y)=\edges(y)\cup \bigcup_{a\in A_y} E(H^\lambda(a))\cap (\binom{\bag(y)}{2} \setminus \binom{\adhesion{y}}{2})$ for all $y\in \Tprefx$, and $\edges^\lambda_x(a)=E(H^\lambda(a))\setminus \binom{\adhesion{a}}{2}$ for all $a\in A_x$. 
 \end{itemize}
 Note that the width of this tree decomposition is at most $c + \ell$.

\newcommand{\Zcmp}{\overline{\Z}}

Let $\Zcmp$ denote the set of non-negative integers together with $+\infty$.
  Let us use a total order on pairs in $\Zcmp \times \Zcmp$ where we first compare the first elements and if they are equal the second elements.
 In a bottom-up fashion, for every $x\in S$ we compute the mapping $\zeta_x\colon Q\times 2^{\binom{\adhesion{x}}{2}}\to \Zcmp \times \Zcmp$ defined as follows: for $q\in Q$ and $W\subseteq \binom{\adhesion{x}}{2}$, $\zeta_x(q,W)$ is the minimum value of $(s^\lambda_x, d^\lambda_x)$ among $\lambda\in \Lambda_x$ such that
 \begin{itemize}
  \item $\Bb$ has a run on $(T^\lambda_x,\bag^\lambda_x,\edges^\lambda_x)$ in which $x$ is labelled with $q$; and
  \item $\bigcup_{a \in A_x} E(H^\lambda(a)) \cap \binom{\adhesion{x}}{2}=W$.
 \end{itemize}
 In case there is no $\lambda$ as above, we set $\zeta_x(q,W)=(+\infty, +\infty)$.
 
 We now argue that the mappings $\zeta_x$ can be computed in a bottom-up manner. This follows from the following rules, whose correctness is straightforward.
 \begin{itemize}
  \item For every $a\in A$, $\zeta_a(q,W)$ is the minimum pair $(s,d)$ such that there is $(d,H)\in R(a)$ with the following properties: $\left(H- \binom{\adhesion{a}}{2},q\right)\in \init$, $E(H)\cap \binom{\adhesion{a}}{2}=W$, and $|V(H) \setminus \adhesion{a}|=s$.
  \item For every $x\in \Tpref$ with no children, $\zeta_x(q,W)=(0,0)$ if $(G_x,q)\in \init$ and $W=\emptyset$, and $\zeta_x(q,W)=(+\infty,+\infty)$ otherwise.
  \item For every $x\in \Tpref$ with one child $y$, $\zeta_x(q,W)$ is the minimum $(s,d)$ such that the following holds: There exist $q'\in Q$ and $W'\subseteq \binom{\adhesion{y}}{2}$ with $(s',d') = \zeta_y(q',W')$  such that 
  \begin{gather*}
   \left(\left(\bag(x),\adhesion{x},\adhesion{y},\emptyset,\edges(x)\cup W'\setminus \binom{\adhesion{x}}{2},q',\bot\right),q\right)\in \trans,\\
   W = W'\cap \binom{\adhesion{x}}{2},\\
   s = s', \text{ and}\\
   d = d'+s'.
  \end{gather*}
  If there are no $q',q'',W'$ as above, then $\zeta_x(q,W)=+\infty$. 
  \item For every $x\in \Tpref$ with two children $y$ and $z$, $\zeta_x(q,W)$ is the minimum $(s,d)$ such that the following holds: There exist $q',q''\in Q$, $W'\subseteq \binom{\adhesion{y}}{2}$, and $W''\subseteq \binom{\adhesion{z}}{2}$ with $(s',d') = \zeta_y(q',W')$ and $(s'',d'') = \zeta_z(q'',W'')$ such that 
  \begin{gather*}
   \left(\left(\bag(x),\adhesion{x},\adhesion{y},\adhesion{z},\edges(x)\cup W'\cup W''\setminus \binom{\adhesion{x}}{2},q',q''\right),q\right)\in \trans,\\
   W = (W'\cup W'')\cap \binom{\adhesion{x}}{2},\\
   s = s'+s'', \text{ and}\\
   d = d'+d''+s'+s''.
  \end{gather*}
  If there are no $q',q'',W',W''$ as above, then $\zeta_x(q,W)=+\infty$.
 \end{itemize}
 Using the rules above, all mappings $\zeta_x$ for $x\in S$ can be computed in total time $2^{\Oh{k(\ell+c)^2}}\cdot |\Tpref|$, because $|Q|\leq  2^{\Oh{k(\ell+c)^2}}$ and the evaluation time of $\Bb$ is $2^{\Oh{k(\ell+c)^2}}$.
 
 By the properties of $\Bb$ asserted in \cref{lem:automaton-BK}, the minimum $(s^\lambda, d^\lambda)$ among those $\lambda\in \Lambda$ for which $H^\lambda$ has treewidth at most $k$ is equal to $\min_{q\in F} \zeta_r(q,\emptyset)$, where $r$ is the root of $T$. The latter minimum can be computed in time $2^{\Oh{k(\ell+c)^2}}$ knowing $\zeta_r$. Finally, to find $\lambda\in \Lambda$ witnessing the minimum, it suffices to retrace the dynamic programming in the standard way. That is, when computing mappings $\zeta_x$, for every computed value of $\zeta_x(q,W)$ we memorize how this value was obtained. After finding $q\in F$ that minimizes $\zeta_r(q,\emptyset)$ we recursively retrace how the value of $\zeta_r(q,\emptyset)$ was obtained along $S$ in a top-down manner, up to values computed in the nodes of $A$; the ways in which these values were obtained give us the mapping $\lambda$. This concludes the construction of the closure $X=V(H^\lambda)$ and its torso $H^\lambda$.
 
 \newcommand{\profile}{\mathsf{profile}}
 
 \paragraph*{Step 2: finding blockages.} Having constructed $X=V(H^\lambda)$ together with $\torso[G]{X}=H^\lambda$, we proceed to finding the set $\blockages{\Tpref}{X}$. We may assume that every vertex of $X$ has been marked as belonging to $X$ (which can be done after computing $X$ in time $\Oh{|X|}\leq (c+\ell)^{\Oh{1}}\cdot |\Tpref|$), hence checking whether a given vertex belongs to $X$ can be done in constant time.
 
 First, we observe that for every node $x\in V(T)\setminus \Tpref$, we can efficiently find out the information about the behavior of $X$ in the subtree rooted at $x$. Denote $X_x\coloneqq X\cap \component{x}$. 
 
 \begin{claim}\label{cl:conn-down}
  Given a node $x\in V(T)\setminus \Tpref$, one can compute $X_x$ and $\torso[G_x]{X_x\cup \adhesion{x}}$ in time $2^{\Oh{(c+\ell)^2}}$.
 \end{claim}
 \begin{proof}
  For every $(d,H)\in \reps^c(T_x,\bag_x,\edges_x)$, call $H$ a {\em{candidate}} if $V(H)\setminus \adhesion{x}\subseteq X$. Note that we can find all candidates in time $(c+\ell)^{\Oh{1}}\cdot |\reps^c(T_x,\bag_x,\edges_x)|\leq 2^{\Oh{(c+\ell)^2}}$ by inspecting all elements of $\reps^c(T_x,\bag_x,\edges_x)$ one by one. Now, a simple exchange argument shows that $X_x\cup \adhesion{x}$ is equal to the largest (in terms of the number of vertices) candidate, and $\torso[G_x]{X_x\cup \adhesion{x}}$ is equal to this candidate.
 \end{proof}
 
 For technical purposes, we need to set up a simple data structure for checking adjacencies in $H^\lambda$. The idea is based on the notion of degeneracy and can be considered folklore.


\begin{claim}\label{cl:degeneracy}
In time $(c+\ell)^{\Oh{1}}\cdot |\Tpref|$ one can set up a data structure that for given vertices $u,v\in X$, can decide whether $u$ and $v$ are adjacent in $H^\lambda$ time $\Oh{\ell}$. 
 \end{claim}
 \begin{proof}
 
The treewidth of $H^\lambda$ is at most $2k+1 = \Oh{\ell}$, implying that we can in time $\Oh{k \cdot |X|}$ compute a total order $\preceq$ on $X$ so that every vertex $u\in X$ has only at most $2k+1$ neighbors that are earlier than $u$ in $\preceq$.
 Therefore, for every vertex $u$ of $X$ we construct a list $L(u)$ consisting of all neighbors of $u$ that are earlier than $u$ in $\preceq$. These lists can be constructed by examining the edges of $H^\lambda$ one by one, and for every next edge $uv$, either appending $u$ to $L(v)$ or $v$ to $L(u)$, depending whether $u$ or $v$ is earlier in $\preceq$. Then whether given $u$ and $v$ are adjacent can be decided in time $\Oh{k}$ by checking whether $u$ appears on $L(v)$ or $v$ appears on $L(u)$.
 \end{proof}

 Next, for a node $x\in V(T)$, we define $\profile(x)\subseteq \binom{\bag(x)}{2}$ to be the set comprising of all pairs $\{u,v\}\subseteq \bag(x)$ such that in $G$ there is path connecting $u$ and $v$ that is internally disjoint with~$X$. (Note that this definition concerns $u$ and $v$ both belonging and not belonging to $X$.) 
 Note that if $x\in \Tpref$, we have $\bag(x)\subseteq X$ and $\profile(x)$ consists of all edges of $\torso[G]{X}=H^\lambda$ with both endpoints in $\bag(x)$. Consequently, for such $x$ we can compute $\profile(x)$ in time $\ell^{\Oh{1}}$ using the data structure of \cref{cl:degeneracy}.

 Next, we show that knowing the profile of a parent we can compute the profile of a child.
 
 \begin{claim}\label{cl:conn-up}
  Suppose $x$ is the parent of $y$ in $T$ and $y\notin \Tpref$. Then given $\profile(x)$, one can compute $\profile(y)$ in time $2^{\Oh{(c+\ell)^2}}$.
 \end{claim}
 \begin{proof}
  Let $J$ be the graph on vertex set $\bag(y)$ whose edge set is the union of
  \begin{itemize}[nosep]
   \item $\edges(y)$;
   \item all edges present in $\profile(x)$ with both endpoints in $\adhesion{x}$; and
   \item all edges present in $\torso[G_z]{X_z\cup \adhesion{z}}$ that have both endpoints in $\adhesion{z}$, for every child $z$ of $y$.
  \end{itemize}
 Note that $J$ can be constructed in time $2^{\Oh{(c+\ell)^2}}$ using \cref{cl:conn-down}. Now, it is straightforward to see that $\profile(y)$ consists of all pairs $\{u,v\}\subseteq \bag(y)$ such that in $J$ there is a path connecting $u$ and $v$ that is internally disjoint with $X\cap \bag(y)$. Using this observation, $\profile(y)$ can be now constructed in time $\ell^{\Oh{1}}$, because $J$ has at most $\ell+1$ vertices.
 \end{proof}

 Having established \cref{cl:conn-up}, we can finally describe the procedure that finds blockages. The procedure inspects the appendices $a\in A$ one by one, and upon inspecting $a$ it finds all blockages that are descendants of $a$. To this end, we start a depth-first search in $T_a$ from $a$. At all times, together with the node $x\in V(T_a)$ which is currently processed in the search, we also store $\profile(x)$. Initially, $\profile(a)$ can be computed using \cref{cl:conn-up}, where the profile of the parent of $a$ --- which belongs to $\Tpref$ --- can be computed in time $\ell^{\Oh{1}}$, as argued. When the search enters a node $y$ from its parent $x$, $\profile(y)$ can be computed from $\profile(x)$ using \cref{cl:conn-up} again. Observe that knowing $\profile(x)$, it can be determined in time $(c+\ell)^{\Oh{1}}$ whether $x$ is a blockage:
 \begin{itemize}
  \item If $\bag(x)\subseteq X$ and $\profile(x)=\binom{\bag(x)}{2}$, then $x$ is a clique blockage.
  \item If $\bag(x)\cap X\neq \emptyset$ and there exists $u\in \bag(x)\setminus X$ such that $\{u,v\}\in \profile(x)$ for all $v\in \bag(x)\setminus \{u\}$, then $x$ is a component blockage.
 \end{itemize}
 Consequently, if $x$ is a blockage, we output $x$ and do not pursue the search further. Otherwise, if $x$ is not a blockage, the search recurses to the children of $x$.
 
 Clearly, the total number of nodes visited by the search is bounded by $|\exploration{\Tpref}{X}|$, and for each node we use time $2^{\Oh{(c+\ell)^2}}$. Hence, the algorithm outputs all blockages in total time $2^{\Oh{(c+\ell)^2}}|\exploration{\Tpref}{X}|$.
\end{proof}

\section{Proof of \cref{thm:main}}
\label{sec:ultimate-proof}

In this section we complete the proof of \cref{thm:main}.

Let us roughly sketch the proof.
Recall that by \cref{lem:weak-treewidth-ds} there exists a~data structure that, for an~initially empty dynamic graph $G$, updated by edge insertions and deletions, maintains an~annotated tree decomposition $(T, \bag, \edges)$ of $G$ of width at most $6k + 5$ under the promise that the treewidth of $G$ never grows above $k$.
Now we progressively improve the data structure to become more resilient to queries causing $\tw{G}$ to exceed $k$, and then add support for testing arbitrary $\CMSO_2$ properties:
\begin{itemize}
  \item The first improvement (\cref{lem:semiweak-treewidth-ds}) allows the data structure to detect the updates increasing the treewidth of $G$ above $k$, which enables us to reject them (i.e., refuse to perform such updates and leave the state of the data structure intact).
  This is achieved by: (a) maintaining a~dynamic graph $G'$, equal to $G$ after each update, and an~additional instance of the data structure of \cref{lem:weak-treewidth-ds} maintaining a~tree decomposition of $G'$ of width at most $6k + 11$ under the promise that $\tw{G'} \leq k + 1$; (b) also maintaining the Bodlaender--Kloks automaton (\cref{lem:automaton-BK}) that, given the augmented tree decomposition of $G'$ of width at most $6k + 11$, verifies whether the treewidth of $G$ is at most $k$; (c) applying each update to $G'$ first, using the Bodlaender--Kloks automaton in order to verify whether $\tw{G'} \leq k$ after the update, and, depending on whether this is the case, accepting or rejecting the update.
  
  \item The next improvement (\cref{lem:strong-treewidth-ds}) allows the data structure to actually process the updates that would increase $\tw{G}$ above $k$.
  At each point of time, the data structure will maintain the information on whether $\tw{G} \leq k$ and an~annotated tree decomposition of the most recent snapshot of $G$ of treewidth at most $k$.
  This is a~fairly standard application of the technique of \emph{delaying invariant-breaking insertions} by Eppstein et al.~\cite{EppsteinGIS96}.
  This improvement resolves the former assertion of \cref{thm:main}.
  
  \item In order to resolve the latter assertion of the theorem (i.e., the maintenance of arbitrary $\CMSO_2$ properties), we use the data structure of \cref{lem:strong-treewidth-ds} along with the prefix-rebuilding data structure of \cref{cor:dynamic-cmso}.
\end{itemize}

All three steps should be considered standard and fairly straightforward; for instance, the first two steps are used to perform an~analogous improvement to the data structure of Majewski et al.~\cite{MajewskiPS23} maintaining $\CMSO_2$ properties of graphs with bounded feedback vertex number.

\begin{lemma}
  \label{lem:semiweak-treewidth-ds}
  There is a~data structure that for an~integer $k \in \N$, fixed upon initialization, and a~dynamic graph $G$, updated by edge insertions and deletions, maintains an~annotated tree decomposition $(T, \bag, \edges)$ of $G$ of width at most $6k + 5$ using prefix-rebuilding updates, only accepting the updates if the treewidth of the updated graph is at most $k$.
  More precisely, at every point in time the graph has treewidth at most $k$ and the data structure contains an~annotated tree decomposition of $G$ of width at most $6k + 5$.
  The data structure can be initialized on $k$ and an~edgeless $n$-vertex graph $G$ in time $2^{k^{\Oh{1}}} \cdot n$, and then every update:
  \begin{itemize}
    \item takes amortized time $2^{k^{\Oh{1}} \cdot \sqrt{\log n \log \log n}}$;
    \item if the treewidth of the graph after the update would be larger than $k$, then the update is rejected and \ttl{} is returned;
    \item otherwise, the update is accepted and the data structure returns a~sequence of prefix-rebuilding updates used to obtain the annotated tree decomposition of the new graph from $(T, \bag, \edges)$.
  \end{itemize}
  
  \begin{proof}
    We implement the required data structure by setting up three instances of already existing data structures:
    \begin{itemize}
      \item $\D_k$: the data structure from \cref{lem:weak-treewidth-ds} maintaining an~annotated tree decomposition of a~dynamic graph $G_k$ of width at most $6k + 5$ using prefix-rebuilding updates under the promise that $\tw{G_k} \leq k$;
      \item $\D_{k+1}$: an~analogous data structure, maintaining an~annotated tree decomposition of a~dynamic graph $G_{k+1}$ of width at most $6k + 11$ using prefix-rebuilding updates under the promise that $\tw{G_{k+1}} \leq k + 1$;
      \item $\mathbb{BK}$: a~$(6k + 11)$-prefix-rebuilding data structure with overhead $2^{\Oh{k^3}}$ maintaining whether the dynamic graph maintained by an~annotated tree decomposition has treewidth at most $k$; an~existence of such a~data structure follows from the combination of \cref{lem:automaton-BK,lem:automaton-maintenance}.
    \end{itemize}
    
    All the data structures are initialized with an~edgeless graph on $n$ vertices.
    Between the updates to the data structure, we maintain the following invariant: all three data structures store the description of the same dynamic graph of treewidth at most $k$; and moreover, the annotated tree decompositions stored by $\D_{k+1}$ and $\mathbb{BK}$ are identical: each prefix-rebuilding update performed by $\D_{k+1}$ is also applied to $\mathbb{BK}$.
    On each successful update, the data structure returns all prefix-rebuilding updates applied on the annotated tree decomposition stored in $\D_k$.

    It remains to show how the updates are handled.
    \begin{itemize}
      \item Assume an~edge $uv$ is removed from the graph.
        Note that edge removals cannot increase the treewidth of the maintained graph, so the removal can be safely relayed to all data structures.

      \item When an~edge $uv$ is to be added to the graph, we first update $\D_{k+1}$ by adding the edge.
      As the addition of an~edge to a~graph may increase its treewidth by at most one, $\D_{k+1}$ will handle this update.
       If, after the update, $\mathbb{BK}$ confirms that the treewidth of the stored graph is still at most $k$, we accept the update by adding $uv$ also to the graph stored by $\D_k$.
       
       On the other hand, if $\mathbb{BK}$ returns that the treewidth of the graph after the update is larger then $k$, then we reject the update: we update $\D_{k+1}$ by removing $uv$ from the stored graph and return \ttl.\footnote{Note that even though the update is rejected in this case and the graphs stored by $\D_k$ and $\D_{k+1}$ do not change, the annotated tree decomposition stored by $\D_{k+1}$ may ultimately be modified by the update.}
    \end{itemize}
    
    It is easy to show that the updates maintain the prescribed invariants and that on each query to the data structure, $\D_k$ and $\D_{k+1}$ are updated a~constant number of times, so the required time complexity bounds follow for these data structures.
    It also follows from \cref{lem:weak-treewidth-ds} that the total size of prefix-rebuilding updates performed by $\D_{k+1}$ over the first $q$ queries is $2^{k^{\Oh{1}}} \cdot n\, +\, 2^{k^{\Oh{1}} \cdot \sqrt{\log n \log \log n}} \cdot q$.
    As each update of $\D_{k+1}$ is also applied to $\mathbb{BK}$ (a~prefix-rebuilding data structure with overhead $2^{\Oh{k^3}}$), the first $q$ updates to the data structure are processed by $\mathbb{BK}$ in time $2^{k^{\Oh{1}}} \cdot n\, +\, 2^{k^{\Oh{1}} \cdot \sqrt{\log n \log \log n}} \cdot q$.
    Thus, the amortized time complexity of the data structure is proved.
  \end{proof}
\end{lemma}

We now strengthen the data structure from \cref{lem:semiweak-treewidth-ds} so as to accept the edge insertions increasing the treewidth of the graph above $k$.
To this end, we use the aforementioned technique of postponing invariant-breaking insertions.
Here, we use a~formulation of the technique by Chen et al.~\cite{abs-2006-00571}.

Suppose $U$ is a~universe.
A~family of subsets $\Fc \subseteq 2^U$ is said to be \emph{downward closed} if $F \in \Fc$ implies $E \in \Fc$ for all $E \subseteq F$.
Assume that $\F$ is a~data structure maintaining a~dynamic subset $X \subseteq U$ under insertions and removals of elements.
We say that $\F$:
\begin{itemize}
  \item \emph{strongly supports $\Fc$ membership} if $\F$ supports a~boolean query $\mathsf{member}()$ which returns whether $X \in \Fc$;
  \item \emph{weakly supports $\Fc$ membership} if $\F$ maintains an~invariant that $X \in \Fc$, rejecting the insertions that would break the invariant.
\end{itemize}
Then, the following statement provides a~way of turning weak support for $\Fc$ membership into the strong support:

\begin{lemma}[\cite{abs-2006-00571}, Lemma 11.1]
  \label{lem:chen-delaying-invariant-breaking}
  Suppose $U$ is a~universe and we have access to a~dictionary $\Lb$ on $U$.
  Let $\Fc \subseteq 2^U$ be downward closed and suppose that there is a~data structure $\D$ weakly supporting $\Fc$ membership.
  
  Then there is a~data structure $\D'$ strongly supporting $\Fc$ membership, where each $\mathsf{member}()$ query takes $\Oh{1}$ time and each update uses amortized $\Oh{1}$ time and amortized $\Oh{1}$ calls to operations on $\Lb$ and $\D$.
  Moreover, $\F'$ maintains an~instance of $\F$ and whenever $\mathsf{member}() = \mathsf{true}$, then $\F$ stores the same set $X \in \Fc$ as $\F'$.
\end{lemma}

The last statement of \cref{lem:chen-delaying-invariant-breaking} was not stated formally; however, it is immediate from the proof in~\cite{abs-2006-00571}.

We proceed to show how \cref{lem:chen-delaying-invariant-breaking} can be applied to the data structure of \cref{lem:semiweak-treewidth-ds}.

\begin{lemma}
  \label{lem:strong-treewidth-ds}
  There is a~data structure that for an~integer $k \in \N$, fixed upon initialization, and a~dynamic graph $G$, updated by edge insertions and deletions, maintains:
  \begin{itemize}
    \item an~annotated tree decomposition $(T, \bag, \edges)$ of width at most $6k + 5$ of the most recent snapshot of $G$ of treewidth at most $k$;
    \item a~boolean information on whether $\tw{G} \leq k$.
  \end{itemize}
  The data structure can be initialized on $k$ and an~edgeless $n$-vertex graph $G$ in time $2^{k^{\Oh{1}}} \cdot n$, and then every update:
  \begin{itemize}
    \item takes amortized time $2^{k^{\Oh{1}} \cdot \sqrt{\log n \log \log n}}$;
    \item returns \ttl{} if the treewidth of the graph after the update is larger than $k$; otherwise, returns the sequence of prefix-rebuilding updates used to modify $(T, \bag, \edges)$.
  \end{itemize}
  
  \begin{proof}
    Following the notation of Chen et al., define the universe $U \coloneqq \binom{V(G)}{2}$.
    Moreover, define a~family $\Fc \subseteq 2^U$ as follows: $X \in \Fc$ if and only if the graph $H$ with $V(H) = V(G)$ and $E(H) = X$ has treewidth at most $k$.
    It is immediate that $\Fc$ is downward closed.
    Let also $\F$ be the data structure of \cref{lem:semiweak-treewidth-ds}.
    Observe that $\F$ weakly supports $\Fc$: it handles edge additions and removals, accepting them exactly when the treewidth of the graph after the update is at most $k$.
    Therefore, by \cref{lem:chen-delaying-invariant-breaking}, there exists a~data structure $\F'$ strongly supporting $\Fc$.
    Moreover, $\F'$ maintains an~instance of $\F$ and whenever the treewidth of the maintained graph is at most $k$, then $\F$ stores the same graph as $\F'$.
    
    Now, our data structure maintains the instances of $\F$ and $\F'$ as above and the sequence $s$ of prefix-rebuilding updates applied to $\F$ since the last time the treewidth of $G$ was $k$ or less.
    Each update is relayed verbatim to $\F'$; in turn, $\F'$ updates $\F$ through edge insertions and deletions, causing $\F$ to rebuild its annotated tree decomposition using a~sequence of prefix-rebuilding updates.
    Such updates are appended to $s$.
    If, after the update to our data structure, $\F'$ returns $\mathsf{member}() = \mathsf{false}$, then we return \ttl{}. Otherwise, we return the sequence $s$ of prefix-rebuilding updates and clear $s$ before the next query.

    It is straightforward to show that the described data structure satisfies all the requirements of the lemma.
  \end{proof}
\end{lemma}

The proof of \cref{thm:main} is now direct:

\begin{proof}[Proof of \cref{thm:main}]
  The first part of the statement of the theorem is immediate from \cref{lem:strong-treewidth-ds}.
  Now, assume we are given a $\CMSO_2$ formula $\varphi$ and we are to verify whether $G \models \varphi$ whenever $\tw{G} \leq k$.
  To this end, we invoke \cref{cor:dynamic-cmso} and instantiate a~$(6k + 5)$-prefix-rebuilding data structure $\M$ with overhead $\Oh[k,\varphi]{1}$ that can be queried whether $G \models \varphi$ in worst-case time $\Oh[k, \varphi]{1}$.
  Now, whenever the data structure of \cref{lem:strong-treewidth-ds} returns that $\tw{G} \leq k$, it returns a~sequence of prefix-rebuilding updates, which we immediately forward to $\M$.
  Then, we query $\M$ to verify whether $G \models \varphi$.
  
  From \cref{lem:strong-treewidth-ds} it immediately follows that the total size of all prefix-rebuilding updates over the first $q$ queries is at most $2^{k^{\Oh{1}}} \cdot n\, +\, 2^{k^{\Oh{1}} \cdot \sqrt{\log n \log \log n}} \cdot q$.
  Hence, $\M$ processes the first $q$ queries in time $\Oh[k, \varphi]{2^{k^{\Oh{1}}} \cdot n\, +\, 2^{k^{\Oh{1}} \cdot \sqrt{\log n \log \log n}} \cdot q}$, which satisfies the postulated amortized time complexity bounds.
\end{proof}

\end{document}